\documentclass[
submission
]{dmtcs-episciences-tampered}

\usepackage[utf8]{inputenc}
\usepackage{subfigure}

\usepackage[square,numbers]{natbib}

\author{Clemens Grabmayer\affiliationmark{1}
                                            }

\title{{Linear Depth Increase of Lambda Terms
        along Leftmost-Outermost Beta-Reduction}}   

\affiliation{
  Gran Sasso Science Institute, Viale F. Crispi, 7, 67100 L'Aquila AQ, Italy}
\received{2019-09-30}

\usepackage{amsmath,amsthm,amssymb}
\usepackage[english]{babel}
\usepackage{stmaryrd}
\usepackage{authblk}
\usepackage{color,graphicx}
\usepackage{hyperref,url}
\usepackage{tikz}
\usepackage[normalem]{ulem}
\usepackage{calc}
 
\usepackage{tikz}
\usetikzlibrary{arrows,shapes,calc,positioning}
 
\newcommand{\funin}{\mathrel{:}}
\newcommand{\fap}[2]{#1({#2})}
\newcommand{\bfap}[3]{{#1}({#2},\hspace*{0.05pt}{#3})}
\newcommand{\tfap}[4]{{#1}({#2},\hspace*{0.05pt}{#3},\hspace*{0.05pt}{#4})}
\newcommand{\iap}[2]{#1 _{#2}}
\newcommand{\indap}[2]{#1 _{#2}}
\newcommand{\subap}[2]{#1 _{#2}}
\newcommand{\supap}[2]{#1 ^{#2}}
\newcommand{\bap}{\subap}
\newcommand{\pap}{\supap}
\newcommand{\bpap}[3]{#1 _{#2}^{#3}}
\newcommand{\pbap}[3]{#1 _{#3}^{#2}}

\newcommand{\nb}{\nobreakdash}

\newcommand{\nf}{\normalfont}

\newcommand{\sdefdby}{{:=}}
\newcommand{\defdby}{\mathrel{\sdefdby}}

\newcommand{\punc}[1]{\ensuremath{\hspace*{3pt}{#1}}}

\definecolor{azure}{rgb}{0.94,1.00,1.00}
\definecolor{brown}{rgb}{.75,.25,.25}
\definecolor{cyan}{rgb}{0.25,0.88,0.82}
\definecolor{chocolate}{rgb}{0.82,0.41,0.12}
\definecolor{darkcyan}{rgb}{0.5,0,1}
\definecolor{darkgreen}{rgb}{0,0.39,0}
\definecolor{darkmagenta}{rgb}{0.5,0,0.5}
\definecolor{darkgoldenrod}{RGB}{184,134,11}
\definecolor{firebrick}{RGB}{175,25,25}
\definecolor{forestgreen}{rgb}{0.13,0.55,0.13}
\definecolor{goldenrod}{RGB}{218,165,32}
\definecolor{lightcyan}{rgb}{0.88,1.00,1.00}
\definecolor{lightpink}{rgb}{1.00,0.71,0.76}
\definecolor{myyellow}{RGB}{235,235,0}
\definecolor{lightyellow}{rgb}{1.00,1.00,0.88}
\definecolor{lightgoldenrod}{rgb}{0.83,0.97,0.51}
\definecolor{lightgoldenrodyellow}{rgb}{0.98,0.98,0.82}
\definecolor{lightskyblue}{rgb}{0.53,0.81,0.98}
\definecolor{moccasin}{rgb}{1.00,0.89,0.71}
\definecolor{magenta}{rgb}{1,0,1}
\definecolor{navyblue}{rgb}{0,0,0.5}
\definecolor{orange}{rgb}{1.0,0.65,0.0}
\definecolor{orangered}{rgb}{1.0,0.27,0.0}
\definecolor{palegreen}{rgb}{0.60,0.98,0.60}
\definecolor{powderblue}{rgb}{0.69,0.88,0.90}
\definecolor{purple}{rgb}{1,0.5,1}
\definecolor{royalblue}{RGB}{65,105,225}
\definecolor{mediumblue}{RGB}{0,0,205}
\definecolor{cornflowerblue}{RGB}{100,149,237}
\definecolor{springgreen}{rgb}{0.0,1.0,0.5}
\definecolor{turquoise}{rgb}{0.25,0.88,0.82}
\definecolor{snow}{rgb}{1.00,0.98,0.98}
\definecolor{tan}{rgb}{0.82,0.71,0.55}
\definecolor{red}{rgb}{1,0,0}
\definecolor{violetred}{RGB}{208,32,144}

\newcommand{\tuple}[1]{\langle #1 \rangle}
\newcommand{\tuplespace}{\hspace*{0.5pt}}
\newcommand{\pair}[2]{\tuple{#1, \tuplespace #2}}

\newcommand{\descsetexpmid}{\mathrel{\vert}}
\newcommand{\descsetexp}[2]{\left\{{#1}\descsetexpmid{#2}\right\}}

\newcommand{\descsetexpbig}[2]{\bigl\{{#1}\descsetexpmid{#2}\bigr\}}
\newcommand{\setexp}[1]{\left\{{#1}\right\}}
\newcommand{\setexpbig}[1]{\bigl\{{#1}\bigr\}}

\renewcommand{\emptyset}{\varnothing}

\newcommand{\family}[2]{\setexp{#1}_{#2}}

\newcommand{\nat}{\mathbb{N}}

\newcommand{\slognot}{\neg}
\newcommand{\lognot}[1]{\slognot{#1}}

\newcommand{\existsst}[2]{\exists{\hspace*{1pt}#1}{#2}}

\newcommand{\sbigOmicron}{O}
\newcommand{\bigOmicron}{\fap{\sbigOmicron}}

\newcommand{\sbigOmega}{\Omega}
\newcommand{\bigOmega}{\fap{\sbigOmega}}

\newcommand{\avar}{x}
\newcommand{\bvar}{y}
\newcommand{\cvar}{z}
\newcommand{\avari}{\indap{\avar}}
\newcommand{\bvari}{\indap{\bvar}}
\newcommand{\cvari}{\indap{\cvar}}

\newcommand{\ater}{s}
\newcommand{\bter}{t}
\newcommand{\cter}{u}

\newcommand{\cteracc}{\cter'}
\newcommand{\ateri}{\indap{\ater}}
\newcommand{\bteri}{\indap{\bter}}
\newcommand{\cteri}{\indap{\cter}}

\newcommand{\cteracci}{\indap{\cteracc}}

\newcommand{\asig}{\Sigma}
\newcommand{\asigi}{\indap{\asig}}
\newcommand{\asigmin}{\supap{\asig}{-}}

\newcommand{\asiglosim}{\indap{\asig}{\scriptlosim}}

\newcommand{\asiglop}{\indap{\asig}{\scriptlop}}
\newcommand{\asiglopsim}{\indap{\asig}{\scriptlopsim}}

\newcommand{\asiglambda}{\indap{\asig}{\lambda}}
\newcommand{\asigexpand}{\indap{\asig}{\scriptexpand}}
\newcommand{\asigexp}{\indap{\asig}{\scriptexp}}

\newcommand{\arules}{R}

\newcommand{\ruleslopsim}{\indap{\arules}{\scriptlopsim}}

\newcommand{\rulesexp}{\indap{\arules}{\scriptexp}}
\newcommand{\rulesexpprime}{\indap{\arules}{\scriptexp'}}

\newcommand{\aTRS}{{\cal R}}
\newcommand{\alTRS}{{\cal L}}

\newcommand{\TRS}{TRS}
\newcommand{\TRSs}{TRSs}

\newcommand{\stermsover}{\text{\nf Ter}}
\newcommand{\termsover}{\fap{\stermsover}}

\newcommand{\scontextsover}{\textit{Cxt}}
\newcommand{\contextsover}{\fap{\scontextsover}}
\newcommand{\scontextsnover}{\bap{\scontextsover}}
\newcommand{\contextsnover}[1]{\fap{\scontextsnover{#1}}}

\newcommand{\sarity}{\textit{ar}}
\newcommand{\arityof}{\fap{\sarity}}

\newcommand{\vars}{\textit{Var}}

\newcommand{\sfolapp}{@}
\newcommand{\sfonlabs}{\lambda}
\newcommand{\sfolabs}[1]{(\lambda{#1})}

\newcommand{\folapp}{\bfap{\sfolapp}}

\newcommand{\folabs}[1]{\fap{\sfolabs{#1}}}

\newcommand{\afovar}{\mathsf{v}}
\newcommand{\afovari}[1]{\indap{\afovar}{\hspace*{-0.25pt}#1}}

\newcommand{\afoscopesym}{f}
\newcommand{\bfoscopesym}{g}
\newcommand{\cfoscopesym}{h}
\newcommand{\dfoscopesym}{i}
\newcommand{\afoscopesymi}{\indap{\afoscopesym}}
\newcommand{\afoscope}{\fap{\afoscopesym}}

\newcommand{\dfoscope}{\fap{\dfoscopesym}}

\newcommand{\cxtap}[2]{{#1}[#2]}

\newcommand{\acxt}{C}
\newcommand{\bcxt}{D}
\newcommand{\ccxt}{E}
\newcommand{\acxtap}{\cxtap{\acxt}}
\newcommand{\bcxtap}{\cxtap{\bcxt}}

\newcommand{\acxti}{\bap{\acxt}}

\newcommand{\ccxti}{\bap{\ccxt}}

\newcommand{\afoscopecxt}{F}
\newcommand{\afoscopecxtap}{\cxtap{\afoscopecxt}}
\newcommand{\hole}{\Box}
\newcommand{\holei}{\iap{\hole}}
\newcommand{\holes}{\pmb{\hole}}

\newcommand{\nary}[1]{{$#1$}\nb-ary}

\newcommand{\ruleof}[1]{\indap{\rho}{\hspace*{-0.3pt}#1}}

\newcommand{\smyleftspoon}{\hbox{${\circ}\kern-1.3pt{\relbar}$}}
\newcommand{\smyrightspoon}{\hbox{${\relbar}\kern-1.2pt{\circ}$}}
\newcommand{\sdependson}{\smyleftspoon}
\newcommand{\dependson}{\mathrel{\sdependson}}
\newcommand{\sisnestedinto}{\smyrightspoon}

\newcommand{\depth}[1]{\lvert{#1}\rvert} 

\newcommand{\holedepth}[1]{\lvert{#1}\rvert_{\hole}}

\newcommand{\expdepth}[1]{\lvert{#1}\rvert_{\scriptexp}}
\newcommand{\expdepthbig}[1]{\bigl\lvert{#1}\bigr\rvert_{\scriptexp}}

\newcommand{\expholedepth}[1]{\lvert{#1}\rvert_{\scriptexp,\hole}}
\newcommand{\expholedepthbig}[1]{\bigl\lvert{#1}\bigr\rvert_{\scriptexp,\hole}}

\newcommand{\holedepthnotexp}[1]{\lvert{#1}\rvert_{\hole}^{\scriptnotexp}}
\newcommand{\holedepthnotexpbig}[1]{\big\lvert{#1}\big\rvert_{\hole}^{\scriptnotexp}}

\newcommand{\scriptnotexp}{\text{\sout{\nf\hspace*{1pt}exp\hspace*{1pt}}}}
\newcommand{\depthnotexp}[1]{\lvert{#1}\rvert^{(\scriptnotexp)}}
\newcommand{\depthnotexpbig}[1]{\big\lvert{#1}\big\rvert^{(\scriptnotexp)}}

\newcommand{\snestdepth}{\textit{d}_{\text{\nf nest}}}
\newcommand{\nestdepth}{\fap{\snestdepth}}
\newcommand{\smaxnestdepth}{\textit{D}_{\text{\nf nest}}}
\newcommand{\maxnestdepth}{\fap{\smaxnestdepth}}

\newcommand{\size}[1]{\left\lVert{#1}\right\rVert} 

\newcommand{\expsize}[1]{\left\lVert{#1}\right\rVert\bap{}{\scriptexp}}

\newcommand{\slop}{\textit{lop}}

\newcommand{\slopn}{\subap{\slop}}
\newcommand{\slopstart}{\slop}
\newcommand{\lopstart}{\fap{\slop}}
\newcommand{\lopn}[1]{\fap{\slopn{#1}}}
\newcommand{\slopni}[2]{\slop_{#1,\hspace*{1pt}#2}}
\newcommand{\lopni}[2]{\fap{\slopni{#1}{#2}}}  
\newcommand{\ssubst}{\textit{subst}}
\newcommand{\subst}{\tfap{\ssubst}}

\newcommand{\slopsimTRS}{{\cal L\hspace*{-0.75pt}O}} 
\newcommand{\lopsimTRSwrt}{\fap{\slopsimTRS}}

\newcommand{\sexpandTRS}{{\cal E}}

\newcommand{\sexpandTRSlopsim}{\indap{\sexpandTRS}{\scriptlopsim}}
\newcommand{\expandTRSwrt}{\fap{\sexpandTRS}}

\newcommand{\expandTRSlopsimwrt}{\fap{\sexpandTRSlopsim}}

\newcommand{\denlterrepwrt}[2]{\llbracket{#2}\rrbracket^{#1}}

\newcommand{\denlterwrt}[2]{\llbracket{#2}\rrbracket^{#1}_{\sslabs}}
\newcommand{\denlter}[1]{\llbracket{#1}\rrbracket_{\sslabs}}

\newcommand{\lopsimTRS}{lopsim-TRS}

\newcommand{\alter}{M}
\newcommand{\blter}{N}
\newcommand{\clter}{L}
\newcommand{\dlter}{P}
\newcommand{\elter}{Q}

\newcommand{\dlteracc}{\dlter'}

\newcommand{\alteri}{\indap{\alter}}
\newcommand{\blteri}{\indap{\blter}}
\newcommand{\clteri}{\indap{\clter}}
\newcommand{\dlteri}{\indap{\dlter}}

\newcommand{\dlteracci}{\indap{\dlteracc}}

\newcommand{\twochurch}{\textit{two}}

\newcommand{\sslabs}{\lambda}
\newcommand{\slabs}[1]{\sslabs{#1}.}
\newcommand{\labs}[2]{\slabs{#1}{#2}}
\newcommand{\lapp}[2]{{#1}{#2}}

\newcommand{\substin}[2]{{#1}[{#2}]}
\newcommand{\substinfor}[3]{\substin{#1}{{#2}\defdby{#3}}}

\newcommand{\sred}{\to}
\newcommand{\red}{\mathrel{\sred}}
\newcommand{\sredi}{\indap{\sred}}
\newcommand{\redi}[1]{\mathrel{\sredi{#1}}}

\newcommand{\smred}{\twoheadrightarrow}
\newcommand{\mred}{\mathrel{\smred}}
\newcommand{\smredi}{\indap{\smred}}
\newcommand{\mredi}[1]{\mathrel{\smredi{#1}}}
\newcommand{\sredn}{\supap{\sred}}

\newcommand{\sredin}[2]{\bpap{\sred}{#1}{#2}}

\newcommand{\snfred}{{\downarrow}}

\newcommand{\ssyntequal}{{\equiv}}
\newcommand{\syntequal}{\mathrel{\ssyntequal}}
\newcommand{\snotsyntequal}{{\not\equiv}}
\newcommand{\notsyntequal}{\mathrel{\snotsyntequal}}

\newcommand{\scriptlop}{\text{\nf lop}}
\newcommand{\scriptlosim}{\text{\nf losim}}

\newcommand{\scriptlopsim}{\text{\nf lopsim}}

\newcommand{\scriptexpand}{\text{\nf expand}}
\newcommand{\scriptexp}{\text{\nf exp}}
\newcommand{\scriptsearch}{\text{\nf search}}
\newcommand{\scriptcontract}{\text{\nf contr}}
 
\newcommand{\scriptsubst}{\text{\nf subst}}

\newcommand{\scriptinit}{\text{\nf init}}
\newcommand{\scriptdescendinfolapp}{\iap{\text{\nf desc}}{\sfolapp}}
\newcommand{\scriptdescendinfolabs}{\iap{\text{\nf desc}}{\sfonlabs}}
\newcommand{\scriptcontractn}[1]{\iap{\text{\nf contr}}{{#1}}}

\newcommand{\scriptvar}{\text{\nf var}}

\newcommand{\ssearchred}{\sredi{\scriptsearch}}
\newcommand{\searchred}{\mathrel{\ssearchred}}
\newcommand{\ssearchmred}{\smredi{\scriptsearch}}
\newcommand{\searchmred}{\mathrel{\ssearchmred}}
\newcommand{\ssubstred}{\sredi{\scriptsubst}}
\newcommand{\substred}{\mathrel{\ssubstred}}
\newcommand{\ssubstmred}{\smredi{\scriptsubst}}
\newcommand{\substmred}{\mathrel{\ssubstmred}}
\newcommand{\scontractred}{\sredi{\scriptcontract}}
\newcommand{\contractred}{\mathrel{\scontractred}}
\newcommand{\scriptlobeta}{\text{\nf lo}\beta}
\newcommand{\slobetared}{\sredi{\scriptlobeta}}
\newcommand{\lobetared}{\mathrel{\slobetared}}
\newcommand{\slobetaredn}{\sredin{\scriptlobeta}}
\newcommand{\lobetaredn}[1]{\mathrel{\slobetaredn{#1}}}

\newcommand{\slopsimred}{\sredi{\slop}}
\newcommand{\lopsimred}{\mathrel{\slopsimred}}
\newcommand{\slopsimmred}{\smredi{\slop}}
\newcommand{\lopsimmred}{\mathrel{\slopsimmred}}

\newcommand{\sexpand}{\textit{exp}}

\newcommand{\sexpandi}{\indap{\sexpand}}
\newcommand{\expandi}[1]{\fap{\sexpandi{#1}}}

\newcommand{\sexpred}{\sredi{\scriptexp}}
\newcommand{\expred}{\redi{\scriptexp}}

\newcommand{\expmred}{\mredi{\scriptexp}}

\newcommand{\sbetared}{\sred_{\beta}}
\newcommand{\betared}{\mathrel{\sbetared}}
\newcommand{\sbetaredn}{\sredin{\beta}}
\newcommand{\betaredn}[1]{\mathrel{\sbetaredn{#1}}}

\newcommand{\sexprednf}{{\downarrow_{\scriptexp}}}
\newcommand{\exprednf}[1]{{#1}\sexprednf}
\newcommand{\sexprednfi}[1]{{\bpap{\downarrow}{\scriptexp}{\hspace*{-1pt}(#1)}}}
\newcommand{\exprednfi}[2]{{\langle#2\rangle}\sexprednfi{#1}}

\newcommand{\sexprednfbigi}[1]{{\bpap{\smash{\big\downarrow}}{\scriptexp}{\hspace*{-1.5pt}(#1)}}}
\newcommand{\exprednfbigi}[2]{{#2}\sexprednfbigi{#1}}

\newcommand{\sinitred}{\sredi{\scriptinit}}
\newcommand{\initred}{\mathrel{\sinitred}}
\newcommand{\sdescendinfolappred}{\sredi{\scriptdescendinfolapp}}
\newcommand{\descendinfolappred}{\mathrel{\sdescendinfolappred}}
\newcommand{\sdescendinfolabsred}{\sredi{\scriptdescendinfolabs}}
\newcommand{\descendinfolabsred}{\mathrel{\sdescendinfolabsred}}

\newcommand{\svarred}{\sredi{\scriptvar}}

\newcommand{\svarnred}[1]{\sredi{\scriptvar_{#1}}}
\newcommand{\varnred}[1]{\mathrel{\svarnred{#1}}}

\newcommand{\arewseq}{\sigma}

\newcommand{\arewstep}{\rho}

\newcommand{\scomprewrels}[2]{{#1}\cdot{#2}}
\newcommand{\comprewrels}[2]{\mathrel{\scomprewrels{#1}{#2}}}

\newcommand{\lambdaabstraction}{$\sslabs$\nb-ab\-strac\-tion}

\newcommand{\lambdacalculus}{$\lambda$\nb-cal\-cu\-lus}

\newcommand{\lambdaterm}{$\lambda$\nb-term}
\newcommand{\lambdaterms}{\lambdaterm{s}}
\newcommand{\lambdadepth}{$\sslabs$\nb-depth}
\newcommand{\lambdalifting}{lambda-lif\-ting}

\newcommand{\betareduction}{$\beta$\nb-re\-duc\-tion}
\newcommand{\betacontraction}{$\beta$\nb-con\-trac\-tion}
\newcommand{\betaredex}{$\beta$\nb-re\-dex}
\newcommand{\alphaconversion}{$\alpha$\nb-con\-ver\-sion}
\newcommand{\alphaequivalent}{$\alpha$\nb-equi\-va\-lent}
\newcommand{\lo}{left\-most-outer\-most}

\newcommand{\loreduction}{\lo~re\-duc\-tion}

\newcommand{\lopbetareduction}{\lo-par\-al\-lel~\betareduction}
\newcommand{\loredex}{\lo\ re\-dex}
\newcommand{\nondeterministic}{non-de\-ter\-min\-istic}
\newcommand{\nontrivial}{non-triv\-i\-al}

\newcommand{\TRSrepresentation}{\TRS\nb-re\-pre\-sen\-ta\-tion}

\newcommand{\lTRSrepresentation}{\lTRS\nb-re\-pre\-sen\-ta\-tion}
\newcommand{\lTRS}{$\lambda$\hspace*{-0.5pt}\nb-\hspace*{-0.5pt}\TRS}
\newcommand{\lTRSs}{\lTRS{s}}

\newcommand{\welldefined}{well-de\-fined}
\newcommand{\wellfounded}{well-found\-ed}

\newcommand{\sseamred}{\smredi{\text{\nf s}}}
\newcommand{\sconred}{\sredi{\text{\nf c}}}

\newcommand{\lambdacal}{\lambda}
\newcommand{\graphcal}{\mathcal{G}}
\newcommand{\graphcalwrt}{\fap{\graphcal}}
\newcommand{\agraph}{G}
\newcommand{\agraphi}{\iap{\agraph}}
\newcommand{\sgraphred}{{\Longrightarrow}}
\newcommand{\graphred}{\mathrel{\sgraphred}}
\newcommand{\sgraphredrtc}{{\Longrightarrow^*}}
\newcommand{\graphredrtc}{\mathrel{\sgraphredrtc}}
\newcommand{\sgraphredn}[1]{{\pap{\sgraphred}{#1}}}
\newcommand{\graphredn}[1]{\mathrel{\sgraphredn{#1}}}

\newcommand{\stime}{\text{\nf Time}}
\newcommand{\timei}{\bap{\stime}}

\theoremstyle{plain}
\newtheorem{theorem}{Theorem}

\newtheorem*{corollary*}{Corollary}
\newtheorem{lemma}[theorem]{Lemma}
\newtheorem{proposition}[theorem]{Proposition}

\theoremstyle{definition}
\newtheorem{definition}[theorem]{Definition}

\newtheorem{example}[theorem]{Example}

\begin{document}
\publicationdetails{VOL}{2015}{ISS}{NUM}{SUBM}
\maketitle
\begin{abstract}
  Performing $n$ steps of \betareduction\ to a given term in the \lambdacalculus\ 
  can lead to an increase in the size of the resulting term that is exponential in $n$. 
  The same is true for the possible depth increase of terms along a \betareduction\ sequence. 
  We explain that the situation is different for the \lo\ strategy for \betareduction:
  while exponential size increase is still possible, depth increase is bounded linearly in the number of steps. 
  For every \lambdaterm~$\alter$ with depth $d$, 
  in every step of a \lo\ \betareduction\ rewrite sequence starting from $\alter$ 
  the term depth increases by at most $d$.
  Hence the depth of the $n$\nb-th reduct of $\alter$ in such a rewrite sequence~is~bounded~by~$d\cdot (n+1)$.
  
  We prove the lifting of this result to \lambdaterm\ representations as orthogonal first-order term rewriting systems,
  which can be obtained by the \lambdalifting\ transformation.
  For the transfer to \lambdacalculus, we rely on correspondence statements via \lambdalifting. 
  We argue that the linear-depth-increase property can be a stepping stone for an alternative proof of, and so can shed new light on,
  a result by Accattoli and Dal~Lago (2015) that states: 
  leftmost-outermost \betareduction\ rewrite sequences of length $n$ in the \lambdacalculus\ can be implemented on a reasonable machine with an overhead
  that is polynomial in $n$ and the size of the initial term. 
  
\end{abstract}

\section{Introduction}
%
Accattoli and Dal~Lago \cite{acca:lago:2014:beta-reduction-invariant:LICS,acca:lago:2016}
proved that the number of steps in a
\lo\ rewrite sequence to normal form provides an invariant cost model for the \lambdacalculus, in the following sense.\ 
There is an implementation $I$ on a reasonable machine (e.g., a Turing machine, or a random access machine)
of the partial function that maps a \lambdaterm\ to its normal form, whenever that exists,
such that $I$ has the following property: 
there are integer polynomials $\bfap{p}{x}{y}$ 
                                              and $\fap{q}{x}$
such that if a \lambdaterm~$\blter$ is the result of $n$ successive \lo\ \betareduction\ steps performed to a \lambdaterm~$\alter$ of size~$m$,
then $I$ obtains a compact representation~$C$ of $\blter$ from $\alter$
in time bounded by $\bfap{p}{n}{m}$, 
and $\blter$ can be obtained from $C$ in time bounded by $\fap{q}{\size{\blter}}$ where $\size{\blter}$ is the symbol size of 
the represented \lambdaterm~$\blter$.%
  \footnote{That the represented \lambdaterm~$\blter$ must be computable from its compact representation~$C$
            in time bounded by the size of $\blter$,
            which is implicit in the result of \cite{acca:lago:2014:beta-reduction-invariant:LICS,acca:lago:2016},
            is crucial to prevent `hiding' of reduction work in the computation of `pretty printing' $C$ as $\blter$.} 

To achieve this result, Accattoli and Dal~Lago describe how to simulate \lo\ \betareduction\ rewrite sequences in the \lambdacalculus\ by
`\lo\ useful' rewrite sequences in the linear explicit substitution calculus.
They employ substitution steps 
                                             only insofar as those are needed to 
create the \lo\ \betaredex\ (representation), or to make it visible. 
In this way they work with shared representations of \lambdaterms\ in order to avoid size explosion. 
Then they show that `\lo\ useful' rewrite sequences in the linear explicit substitution calculus can be implemented on a reasonable machine
with a polynomial overhead dependent on the length of the sequence, and the size of the initial~term.\pagebreak[4]   

My goal is to connect this result with  graph reduction techniques that are widely used for the compilation and runtime-evaluation of functional programs.
In particular, I would like to obtain a graph rewriting implementation for \lo\ \betareduction\ in the \lambdacalculus\
that demonstrates this result, but that is close in spirit to graph reduction as it is used in runtime evaluators for functional programming languages.  
My idea is to describe a port graph grammar \cite{stew:2002} implementation that is based on 
\TRS\ (term rewrite system) representation{s} of \lambdaterm{s}.
These \lambdaterm\ representations correspond closely to supercombinator systems that are obtained by \lambdalifting, 
as first described by Hughes~\cite{hugh:1982:report}.

That such an implementation is conceivable by employing subterm-sharing 
is suggested by a property of (plain, unshared) \lo\ \betareduction\ rewrite sequences in the \lambdacalculus\
that we will show.
The depth increase in each step of an arbitrarily long \lo\ \betareduction\ rewrite sequence from a \lambdaterm~$\alter$
is uniformly bounded by $\depth{\alter}$, the depth of $\alter$.
As a consequence, for the depth of the $n$\nb-th reduct~$\clteri{n}$ of a \lambdaterm~$\clteri{0}$ 
in a \lo\ \betareduction\ rewrite sequence $\clteri{0} \lobetared \clteri{1} \lobetared \cdots \lobetared \clteri{n} \lobetared \cdots$ 
it holds that: $\depth{\clteri{n}} \le \depth{\clteri{0}} \cdot (n+1)$,
and hence $\depth{\clteri{n}}/\depth{\clteri{0}} \in \bigOmicron{n}$. 

In the terminology of \cite{acca:lago:2014:beta-reduction-invariant:LICS,acca:lago:2016}
this property shows that \lo\ rewrite sequences do not cause `depth explosion' in \lambdaterms. 
This contrasts with general $\sbetared$ rewrite sequences,
along which the depth of terms may increase exponentially.
The example below provides an illustration. 

\begin{example}[`depth-exploding' family under \betareduction, from Asperti and L\'{e}vy \cite{aspe:levy:2013}] 
  Consider the following families $\family{\alteri{i}}{i\in\nat}$ and $\family{\blteri{i}}{i\in\nat}$ of \lambdaterms: 
  \begin{center}
    $
    \begin{aligned}
      \alteri{0} & \defdby \lapp{\avar}{\avar} \punc{,}
        & & \quad &   
      \blteri{0} & \defdby \alteri{0} = \lapp{\avar}{\avar} \punc{,}  
      \\
      \alteri{i+1} & \defdby \lapp{\lapp{\twochurch}{(\labs{\avar}{\alteri{i}})}}{\avar} 
        \quad \text{(for $n\in\nat$)} \punc{,}
        & & & 
      \blteri{i+1} & \defdby \substinfor{\blteri{i}}{\avar}{\blteri{i}} 
        \quad \text{(for $n\in\nat$)} \punc{,}
    \end{aligned}  
    $
  \end{center}
  where 
  $\twochurch \defdby \labs{\avar}{\labs{\bvar}{\lapp{\avar}{(\lapp{\avar}{\bvar})}}}$
  is the Church numeral for $2$.
  By induction on $i$ it can be verified that it holds:
  \begin{center}
    $
    \begin{aligned}
      \depth{\alteri{i}} 
        & =
       \left\{
       \begin{aligned}
         1 &      \;\;\text{ if }\, i = 0
         \\
         3(i+1) & \;\;\text{ if }\, i\ge 1
       \end{aligned}
       \,\right\} 
       \in \bigOmicron{i}
     \hspace*{6ex} & 
      \alteri{i}
        & \betaredn{4i}\;
      \blteri{i}
     \hspace*{6ex} & 
      \depth{\blteri{i}} & = 2^i \in \bigOmega{2^i}
    \end{aligned}
    $
  \end{center}
  and that the syntax tree of $\blteri{i}$ is the complete binary application tree 
  with $2^i$ occurrences of $\avar$ at depth $\depth{\blteri{i}} = 2^i$.
  The induction step for the statement on the rewrite sequence can be performed as follows:
  \begin{alignat*}{6}
    \alteri{i+1}
      =
    \lapp{\lapp{\twochurch}{(\labs{\avar}{\alteri{i}})}}{\avar}
      & \;\betaredn{4i}\; & & 
    \lapp{\lapp{\twochurch}{(\labs{\avar}{\blteri{i}})}}{\avar}  
      =
    \lapp{\lapp{(\labs{\avar}{\labs{\bvar}{\lapp{\avar}{(\lapp{\avar}{\bvar})}}})}{(\labs{\avar}{\blteri{i}})}}{\avar}  
      \hspace*{1.5ex}
    \text{(by the induction hypothesis)}    
    \\
      & \;\betared\; & & 
    \lapp{({\labs{\bvar}{\lapp{(\labs{\avar}{\blteri{i}})}{(\lapp{(\labs{\avar}{\blteri{i}})}{\bvar})}}})}{\avar}      
    \\
      & \;\betared\; & &
    \lapp{(\labs{\avar}{\blteri{i}})}{(\lapp{(\labs{\avar}{\blteri{i}})}{\avar})}   
      \;\betared\;
    \lapp{(\labs{\avar}{\blteri{i}})}{\blteri{i}} 
      \;\betared\; 
    \substinfor{\blteri{i}}{\avar}{\blteri{i}}
      =
    \blteri{i+1} 
  \end{alignat*}
  This $\sbetared$ rewrite sequence is not \lo, but it proceeds mainly in inside-out direction. 
  
  Let $i \ge 1$. Then for $n = 4i$ and $\alter \defdby \alteri{i}$
  it follows that $\alteri{i}$ reduces to its normal form $\blteri{i}$ in precisely $n$ \betareduction\ steps
  $\alter = \alteri{i} = \clteri{0} \betared \clteri{1} \betared \ldots \betared \clteri{n} = \blteri{i}$,
  for with reducts~$\clteri{0},\ldots,\clteri{n}$,
  that the depth of the initial term is
  $\depth{\clteri{0}} = \depth{\alter} = \depth{\alteri{i}} = 3 (i+1) \le 4i = n$,
  and the depth of the final term is
  $\depth{\clteri{n}} = \depth{\blteri{i}} = \depth{\blteri{i}} = 2^i = 2^{n/4}$.
  From this it follows
  $\depth{\clteri{n}}/\depth{\clteri{0}} \ge 2^{n/4}/n$. 
  
  This argument shows that for the relative depth increase of \betareduction\ rewrite sequences $\clteri{0} \betaredn{n} \clteri{n}$ of length $n$
  is exponential, because it holds in any case that
  $\depth{\clteri{n}}/\depth{\clteri{0}} \in \bigOmega{2^{((1/4)-\epsilon)n}}$ for every $\epsilon>0$.
\end{example}

\smallskip
Such an exponential depth increase with respect to general \betareduction\
contrasts sharply with the linear-depth-increase property of \lo\ \betareduction\ that we will show here.
We now lay out the basic insight that is at the basis of this result.

\paragraph{Underlying property, leading to the linear-depth-increase result.}
  For every leftmost-outermost \betareduction\ rewrite sequence 
  $\clteri{0} \lobetared \clteri{1} \lobetared \cdots \lobetared \clteri{n} \lobetared \clteri{n+1} (\,\lobetared \cdots)$
  in the \lambdacalculus\ the following property can be shown:
  if $\lapp{(\labs{\cvar}{\dlter})}{\elter}$ is the \lo\ \betaredex\ in the $n$\nb-th reduct $\clteri{n}$,
  its abstraction part has a representation
  $ \labs{\cvar}{\dlter} 
      \syntequal
    \substin{(\labs{\cvar}{\dlteri{0}})}{\cvari{1}\sdefdby\dlteri{1},\ldots,\cvari{k}\sdefdby\dlteri{k}} $  
  with `scope part' $\labs{\cvar}{\dlteri{0}}$ and `free subexpressions' $\dlteri{0},\dlteri{1},\ldots,\dlteri{k}$,
  where $\cvari{1},\ldots,\cvari{k}\notsyntequal\cvar$ are distinct variables that are free in $\dlteri{0}$, 
  such that an abstraction of the form
  $ \labs{\cvar}{\dlteracc} 
      \syntequal
    \substin{(\labs{\cvar}{\dlteri{0}})}{\cvari{1}\sdefdby\dlteracci{1},\ldots,\cvari{k}\sdefdby\dlteracci{k}} $
  with the same scope part $\labs{\cvar}{\dlteri{0}}$, but possibly with different free subexpressions $\dlteracci{1},\ldots,\dlteracci{k}$,   
  occurs already in $\clteri{0}$ (perhaps as an \alphaconversion\ equivalent variant).
  This implies $\depth{\dlteri{0}} \le \depth{\dlter} < \depth{\clteri{0}}$
  for the depth of $\dlteri{0}$ in relation to the depth of the initial term $\clteri{0}$ of the sequence.
  Now if $\clteri{n} \syntequal \acxtap{(\lapp{\labs{\cvar}{\dlter})}{\elter}}$ for some unary context $\acxt$
  with the \lo\ \betaredex\ highlighted, then the $n$\nb-th step is of the form:
  \begin{equation*}
    \begin{split}
      \clteri{n} 
        & {} \parbox{\widthof{${}\:\lobetared\:{}$}}{${}\:\syntequal\:{}$}
      \acxtap{\lapp{(\labs{\cvar}{\dlter})}{\elter}} 
        \:\syntequal\:
      \acxtap{\lapp{(\substin{(\labs{\cvar}{\dlteri{0}})}{\cvari{1}\,\sdefdby\,\dlteri{1},\ldots,\cvari{k}\,\sdefdby\,\dlteri{k}})}{\elter}}
      \\
        & {} \parbox{\widthof{${}\:\lobetared\:{}$}}{${}\:\syntequal\:{}$}
      \acxtap{\lapp{(\labs{\cvar}{\substin{\dlteri{0}}{\cvari{1}\,\sdefdby\,\dlteri{1},\ldots,\cvari{k}\,\sdefdby\,\dlteri{k}}})}{\elter}}  
      \\
        & {} \parbox{\widthof{${}\:\lobetared\:{}$}}{${}\:\lobetared\:{}$}
      \acxtap{\substin{(\substin{\dlteri{0}}{\cvari{1}\,\sdefdby\,\dlteri{1},\ldots,\cvari{k}\,\sdefdby\,\dlteri{k}})}{\cvar\,\sdefdby\,\elter}}
        \syntequal
      \acxtap{\substin{\dlteri{0}}{\cvari{1}\,\sdefdby\,\dlteri{1},\ldots,\cvari{k}\,\sdefdby\,\dlteri{k},\cvar\,\sdefdby\,\elter}}
        \syntequal
      \clteri{n+1} \punc{.}       
    \end{split} 
  \end{equation*}
  In order to move the substitutions for $\cvari{1},\ldots,\cvari{k}$ inside of the abstraction $\labs{\cvar}{\dlter}$,
  we have assumed here, for simplicity, that $\cvar$ does not occur free in one of $\dlteri{1},\ldots,\dlteri{k}$
  (otherwise \alphaconversion\ would be needed to rename $\cvar$ in $\labs{\cvar}{\dlter}$ first).
  This justifies taking up the substitution of $\elter$ for $\cvar$ into the simultaneous substitution expression 
  after the $\slobetared$ step. 
  Now from the form of the step $\clteri{n} \lobetared \clteri{n+1}$
  we see that any depth increase can only stem from the substitution of $\elter$ for one of the occurrences of $\cvar$ in $\dlteri{0}$.
  This can move the argument $\elter$ of the \betaredex\ deeper by at most $\depth{\dlteri{0}}$.
  So by using $\depth{\dlteri{0}} < \depth{\clteri{0}}$, see above, we obtain 
    $\depth{\clteri{n+1}} < \depth{\clteri{n}} + \depth{\clteri{0}}$.
  In this way we recognize that the depth increase in the $n$\nb-th \lo\ \betareduction\ step is always bounded by the depth $\depth{\clteri{0}}$ of the initial term $\clteri{0}$ of the sequence.

\paragraph{Concepts for showing the underlying property.}
  For showing that scope parts of abstractions in \lo\ redexes of \lo\ \betareduction\ rewrite sequences trace back to the initial term of the sequence,
  we will use representations of \lambdaterms\ as orthogonal first-order term rewrite systems.
  We call these TRS representations \lTRSs.
  They are closely connected to systems of supercombinators \cite{hugh:1982:report,hugh:1982},
  which are widely used for the compilation of functional programs. 
  Supercombinator translations 
                               are obtained by `\lambdalifting' \cite{peyt:1987}.
  This transformation rewrites higher-order terms with bindings (such as named abstractions in \lambdaterms)
  into applicative first-order terms, and a finite number of combinator definitions.
  For functional programs \lambdalifting\ is applied by construeing them as 
  generalized \lambdaterms\ with \textsf{case} and \textsf{letrec} constructs.
  A program is compiled into a finite number of combinator definitions 
  of the form $ C x_1 \ldots x_n = \bcxtap{x_1,\ldots,x_n} $ where $\bcxt$ is an applicative combinator context.
  
  Supercombinator representations are well-suited for the evaluation 
  via \lo\ evaluation. This is because evaluation can proceed by repeatedly applying a combinator definition
  to occurrences of combinators with their sufficient number of arguments.
  In the example as above these are occurrences of applicative terms of the form $C \ateri{1} \ldots \ateri{n}$. 
  In this way evaluation becomes a process of applying combinator definitions locally 
  without having to carry out the substitutions of arguments for variable occurrences that are needed for \betareduction\ on \lambdaterms.
  Moreover, \lo\ \betareduction\ can be simulated by evaluating combinator terms in a \lo\ manner. 
  In the \lTRS\ formulation,
  supercombinator definitions are modeled by rewrite rules 
  $ \folapp{\afoscope{x_1 \ldots x_n}}{y} \red \afoscopecxtap{x_1,\ldots,x_n,y} $ where $\afoscopesym$ is a scope symbol,
  and $\afoscopecxt$ an applicative context that may contain other scope symbols.
  \lTRSs\ correspond to systems of supercombinators that are obtained by `fully-lazy \lambdalifting' \cite{hugh:1982:report,peyt:1987}. 
  
  This construction of first-order term representations of \lambdaterms\ 
  guarantees that every redex of a term in the representing \lTRS\
  corresponds to a \betaredex\ via the translation to \lambdaterms. 
  Indeed, the \lo\ redex on a \lTRS\ term representation of a \lambdaterm\ corresponds to the \lo\ \betaredex\ on the represented \lambdaterm. 
  But conversely, typically not all \betaredex{es} in a \lambdaterm\ will correspond directly to a redex on the \lTRSrepresentation.
  Crucially, after a number of (typically \lo) \betareduction\ steps $\bteri{0} \mred \bteri{n}$ 
    have been simulated from a \lTRS\nb-term $\bteri{0}$ that represents a \lambdaterm~$\alter$,
  every redex $\folapp{\afoscope{\bteri{1} \ldots \bteri{n}}}{\cter}$ in $\bteri{n}$ will involve a scope symbol $\afoscopesym$ that, 
  under the translation to \lambdacalculus,
  represents the scope of $\cvar$ in a a subterm~$\labs{\cvar}{\clter}$ that already occurred (modulo \alphaconversion) in $\alter$. 
  
  \smallskip


While the linear-depth-increase statement will be shown for rewrite sequences in a \TRS\
for simulating \lo\ \betareduction, its transfer to \lambdaterms\ via a lifting theorem along \lambdalifting\
will only be sketched. The lifting and projection statements needed for this part are similar
to proofs for the correctness of fully-lazy \lambdalifting\ as described by Balabonski~\cite{bala:2012}.   

Notwithstanding the linear-depth-increase property for \lo\ rewrite sequences that we show here,
it is important to realize that `size explosion' (exponential size increase) can in fact take place.
There are infinitely many \lambdaterms~$\alteri{n}$ of size $\bigOmicron{n}$ (linear size in $n$)
such that $\alteri{n}$ reduces in $n$ \lo\ \betareduction\ steps
to a term of size $\bigOmega{2^n}$ (properly exponential size in $n$). 

\begin{example}[`size-exploding' family under \lo\ $\beta$-red., from Accattoli and \mbox{Dal~Lago~\cite{acca:lago:2014:beta-reduction-invariant:LICS,acca:lago:2016}}] 
    \label{expl:lo:size:exploding:family}
  Consider the following two families $\family{\alteri{n}}{n\in\nat}$ and $\family{\blteri{n}}{n\in\nat}$ of \lambdaterms:  
  \begin{center}
    $
  \begin{aligned}
    \alteri{0} & \defdby \lapp{\bvar}{\lapp{\avar}{\avar}} \punc{,}
      & 
    \blteri{0} & \defdby 
                         \lapp{\bvar}{\lapp{\avar}{\avar}} \punc{,}
    \\
    \alteri{n+1} & \defdby \lapp{(\labs{\avar}{\alteri{n}})}{\alteri{0}} 
      \quad \text{(for $n\in\nat$)} \punc{,} \qquad\qquad
      & \hspace*{-6ex}
    \blteri{n+1} & \defdby \lapp{\bvar}{\lapp{\blteri{n}}{\blteri{n}}}  
      & \quad \text{(for $n\in\nat$)} \punc{.}
  \end{aligned} 
   $
  \end{center}
  Every term $\blteri{n}$, for $n\in\nat$ is a normal form, 
  and it holds that:
  \begin{equation}\label{eq:expl:lo:size:exploding:family}
    \substinfor{\blteri{n}}{\avar}{\blteri{0}} 
      \: = \:
    \blteri{n+1} 
      \quad \text{(for all $n\in\nat$).}   
  \end{equation}
  This can be shown by induction.
  Furthermore the term $\blteri{n}$ is the normal form of $\alteri{n}$, for $n\in\nat$,
  because there is a \lo\ \betareduction\ rewrite sequence of length~$n$ from $\blteri{n}$ to $\alteri{n}$:
  \begin{center}
    $ 
    \alteri{n}
      \;\lobetaredn{n}\;
    \blteri{n} \quad \text{(for all $n\in\nat$).}
    $  
  \end{center}
  The induction step in a proof of this statement can be verified as follows:
  \begin{center}
    $
    \begin{alignedat}{4}
      \alteri{n+1}
        =
      \lapp{(\labs{\avar}{\alteri{n}})}{\alteri{0}}
        & \;\lobetared\;\, & &
      \substinfor{\alteri{n}}{\avar}{\alteri{0}} &
      \\[-1ex]  
        & \;\lobetaredn{n}\;\, & &
      \substinfor{\blteri{n}}{\avar}{\alteri{0}}   
        & & \qquad 
      \parbox{\widthof{(follows by the ind.\ hyp.\ $\alteri{n} \lobetaredn{n} \blteri{n}$}}%
             {(follows by the ind.\ hyp.\ $\alteri{n} \lobetaredn{n} \blteri{n}$
              \\\phantom{(}%
              by using that $\alteri{0}$ is normal form)}
      \\[-0.35ex]  
        & \; = \;\, & & 
      \substinfor{\blteri{n}}{\avar}{\blteri{0}}   
        & & \qquad 
      \text{(by definition of $\blteri{0}$ and $\alteri{0}$ coincide)}
      \\[-0.35ex]  
        & \;=\; & & 
      \blteri{n+1} & & \qquad
        \text{(by using \eqref{eq:expl:lo:size:exploding:family})} \punc{.}
    \end{alignedat}
    $
  \end{center}
  Finally, the size of terms in $\family{\alteri{n}}{n}$ grows linearly,
  and the size of terms $\family{\blteri{n}}{n}$ exponentially: 
  \begin{center}
    $
    \begin{aligned}
      \size{\alteri{n}} 
        & =
      5 + 8 n \in \bigOmicron{n} \punc{,} 
     &   
      \size{\blteri{n}} & = 2^{n+4} \in \bigOmega{2^n} \punc{,}
    \end{aligned}
    $
  \end{center}
  where by the size of the \lambdaterm\ we understand the size of its syntax tree plus the number of symbols in variable occurrences.   
  %
\end{example}

Therefore naive implementations of \lo\ \betareduction\ that operate
directly on \lambdaterms\ cannot avoid exponential runtime cost, simply because the result
of $n$ \lo\ \betareduction\ steps can be exponentially larger than the initial term.
However, Accattoli and Dal Lago recognized that a \lo\ \betareduction\ sequence 
can also be implemented in the linear substitution calculus
by carrying out explicit-substitution steps of \betaredex\ contractions in a lazy manner
that only guarantees that the pattern of the next \lo\ \betaredex\ is always visibly created. 
They show that, in this way, the size of intermediate \lambdaterm\ representations stays 
polynomially bounded by the length of the sequence. 

The linear-depth-increase property along \lo\ rewrite sequences suggests 
an alternative proof, which is based on graph rewriting, of the result by Accattoli and Dal~Lago.
The crucial idea is to use directed acyclic graph representations of terms in a \lTRS\
with the property that the depth of a graph (which is defined due to acyclicity) corresponds closely to the depth of the represented term. 
Then the power of sharing is deployed to avoid size explosion of the graph representations.
In Section~\ref{sec:idea:graph:implementation} we sketch the basic idea for 
such a graph implementation, and estimate its complexity.


\paragraph{Overview.}
  In Section~\ref{sec:lo-simulation}
    we introduce representations of \lambdaterms\ as first-order terms,
    and define a TRS that simulates the \lo\ strategy 
    (and a \nondeterministic\ generalization) for \betareduction\ 
    on \lambdaterm\ representations. 
  In Section~\ref{sec:lTRSs} 
    we define \lTRSs, that is, representations of \lambdaterms\ as orthogonal term write systems
    that are closely related to supercombinator representations.
    We also define the expansion of \lTRS\ representations into first-order term represenations of \lambdaterms.
  In Section~\ref{sec:lo-simulation:lTRSs}
    we adapt the \lo\ \betareduction\ simulation TRS from Section~\ref{sec:lo-simulation} 
    to \lTRS\ representations of \lambdaterms.  
  In Section~\ref{sec:depth:increase}
    we show the linear-depth-increase result for simulated \lo\ \betareduction\ sequences:
    we prove it for all rewrite sequences in the simulation TRS on \lTRS\ respresentations.
  In Section~\ref{sec:transfer:lambda-calculus}
    we sketch how the linear-depth-increase result can be transferred from \lTRS\ representations to \lambdaterms.
  In Section~\ref{sec:idea:graph:implementation}
    we briefly lay out our idea of using the linear-depth-increase result 
    for developing an efficient graph rewriting system
    for simulating \lo\ \betareduction\ on \lTRS\ representations.

\section{Preliminaries}
  \label{prelims}

By $\nat = \setexp{0,1,2,\ldots}$ we denote the natural numbers including $0$.
For first-order term rewriting systems, terminology and notation from the standard text \cite{terese:2003} will be used.
Below we summarize the most important concepts and the notation that we will use. 

\paragraph{First-order signatures, variables, and context holes.}
  A \emph{(first-order) signature} $\boldsymbol{\asig} = \pair{\asig}{\sarity}$ is a set of function symbols that is equipped with an arity function $\sarity \funin \asig \to \nat$. 
  Such a signature may contain \emph{constants} by which we mean function symbols of arity~$0$. 
  When referring to signatures, we will mostly keep the arity function implicit, and write $\asig$ for $\boldsymbol{\asig}$. 
  
  In addition to first-order signatures we will use
  countably infinite sets $\vars$ of \emph{variables},
  and a countably infinite set $\holes \defdby \setexp{ \holei{1},\holei{2},\ldots }$ of \emph{context hole symbols}
  each of which carries an index.
  We will always tacitly assume that the set~$\vars$, the set~$\holes$, 
  and the union of the set of function symbols in signature $\asigi{1}$, $\asigi{2}$, \ldots, under consideration
  are disjoint.

\paragraph{Terms and contexts over first-order signatures.}   
  By $\termsover{\asig,\vars}$ we denote the set of \emph{terms over signature $\asig$ and set $\vars$ of variables} 
  that are formed with function symbols in $\asig$ and variables in $\vars$.
  By $\termsover{\asig} \defdby \termsover{\asig,\emptyset}$ (with an empty set of variables) we define 
  the set of \emph{ground terms} over $\asig$,
  that is, the set of terms that are formed from only the function symbols in $\asig$.
  We use $\syntequal$ to indicate syntactic equality of terms.
  
  For $n\in\nat$, $n>0$,  we denote by $\contextsnover{n}{\asig,\vars}$ the set of \emph{contexts} 
  that are formed with function symbols in $\asig$ and variables in $\vars$,
  and with $n$ kinds of holes $\holei{1},\ldots,\holei{n}$.
  Note that an \nary{n} context may contain zero, one or more occurrences of each of the $n$ holes;
  so it does not need to have any hole occurrence at all, in which case it is a term.
  As a consequence also $\termsover{\asig,\vars} \subsetneqq \contextsnover{n}{\asig,\vars} \subsetneqq \contextsnover{n+1}{\asig,\vars}$ holds for all $n\in\nat$, $n>0$.
  We also use $\syntequal$ to indicate syntactic equality of contexts.
  For $n\in\nat$, $n>0$, we define by $\contextsnover{n}{\asig} \defdby \contextsnover{n}{\asig,\emptyset}$  
  the set of \nary{n} \emph{ground contexts} over $\asig$,
  that is, the set of \nary{n} contexts that are formed from the function symbols in $\asig$.
  By $\contextsover{\asig,\vars} \defdby \bigcup_{n\in\nat, n>0} \contextsnover{n}{\asig,\vars}$
  we define the set of contexts over $\asig$ and $\vars$ and with some of the holes in $\holes$.
  Note again that $\termsover{\asig,\vars} \subsetneqq \contextsover{\asig,\vars}$ holds.
  By $\contextsover{\asig}$ we denote the set of ground contexts over $\asig$. 
  
  For unary (\nary{1}) contexts $\acxt\in\contextsnover{1}{\asig,\vars}$ we permit to drop the subscript `1'
  from the context hole $\holei{1}$, which then is the single context hole $\holei{1}$ that may occur in $\acxt$,
  and thus we permit to write $\hole$ for $\holei{1}$. 
  
  By $\contextsnover{n,1}{\asig,\vars}$ we denote the subset of $\contextsnover{n}{\asig,\vars}$ that is formed
  by the \emph{\underline{\smash{linear}}} \nary{n} contexts in which every context hole $\holei{i}$, for $i\in\setexp{1,\ldots,n}$
  is only permitted to occur once. By $\contextsnover{n,1}{\asig}$ we denote the set of linear, \nary{n}, ground contexts over $\asig$.
  
  Let $\acxt\in\contextsnover{n}{\asig,\vars}$ be an \nary{n} context.
  Then for terms $\bteri{1},\ldots,\bteri{n}\in\termsover{\asig,\vars}$
  we denote by $\cxtap{\acxt}{\bteri{1},\ldots,\bteri{n}}$ the term in $\termsover{\asig,\vars}$ 
  that results from $\acxt$ by replacing each hole $\holei{i}$ in $\acxt$ by $\bteri{i}$, for all $i\in\setexp{1,\ldots,n}$.
  Similarly, for contexts $\acxti{1},\ldots,\acxti{n}\in\contextsnover{m}{\asig,\vars}$ 
  we denote by $\cxtap{\acxt}{\acxti{1},\ldots,\acxti{n}}$ the context in $\contextsnover{m}{\asig,\vars}$ 
  that results from $\acxt$ by replacing each hole $\holei{i}$ in $\acxt$ by $\acxti{i}$, for all $i\in\setexp{1,\ldots,n}$.

\paragraph{Depth and size of terms. Depth, hole depth, and size of contexts.} 
  For a term $\bter$ we denote by $\depth{\bter}$ the \emph{depth of $\bter$}
  by which we mean the length of the longest (cycle-free) path in the syntax tree of $\ater$ from the root to a leaf. 
  For a context $\acxt$ the \emph{depth $\depth{\acxt}$ of $\acxt$} is defined analogously.  
  By the \emph{size $\size{\bter}$} of a term $\bter$, and the \emph{size $\size{\acxt}$} of a context $\acxt$
  we mean the size of the syntax tree of $\bter$, and $\acxt$, respectively. 
  
  By the \emph{hole depth $\holedepth{\acxt}$} of a context $\acxt$
    we mean the length of the longest (cycle-free) path in the syntax tree of $\acxt$ from the root to a leaf at which some hole occurs.
  We will use the following two lemmas that express easy properties concerning the connection between depth and hole depth in
  filled contexts.   

  \begin{lemma}\label{lem:holedepth:vs:depth}
    $\holedepth{ \acxtap{\ateri{1},\ldots,\ateri{n},\Box} } \:\le\: \depth{\acxt}\:$
    for all terms $\ateri{1},\ldots,\ateri{n}\in\termsover{\asig}$, where $n\in\nat$,
    and all contexts $\acxt\in\contextsnover{n+1}{\asig}$ in which there is at least one occurrence of $\holei{n+1}$. 
  \end{lemma}

  \begin{lemma}\label{lem:depth:cxtap:vs:depth:holedepth}
      $\depth{ \acxtap{\ateri{1},\ldots,\ateri{n}} }  
         \: = \: 
       \max \descsetexp{ \depth{\acxt},\, \holedepth{\acxt} + \depth{\ateri{i}} }
                       { i\in\setexp{1,\ldots,n} }\:$
    for contexts $\acxt\in\contextsnover{n}{\asig}$, and terms $\ateri{1},\ldots,\ateri{n}\in\termsover{\asig}$.
  \end{lemma}

\paragraph{Term rewriting systems.}
  A \emph{term rewriting system} (\TRS) is a pair $\pair{\asig}{\arules}$ that consists of a signature~$\asig$,
  and a set $\arules \subseteq \termsover{\asig,\vars}\times\termsover{\asig,\vars}$ of pairs of terms over $\asig$ that are called \emph{rules}.
  The rules are subject to two conditions: the left-hand side of a rule is not a variable,
  and the variables that occur on the right-hand side of a rule are a subset of the variables that occur on the left-hand side.
  
  A term $\ater\in\termsover{\asig,\vars}$ is a \emph{normal form of} a TRS~$\aTRS = \pair{\asig}{\arules}$ if no rule of $\aTRS$ is applicable to $\ater$.

\paragraph{Notation for rewrite relations.}
  Let $\aTRS$ be a \TRS\ with rewrite relation $\sred$. 
  Then we denote the many-step (zero, one or more step) rewrite relation of $\aTRS$ by $\smred$,
  and the $n$ step rewrite relation of $\aTRS$ by $\sredn{n}$, for $n\in\nat$. 
  By $\snfred$ we mean the many-step relation of $\aTRS$ to a normal form. 
  We will use the same notation convention for rewrite relations that are indexed by name abbreviations.

\paragraph{$\lambda$-calculus.}
  Contrasting with terms in a \TRS\ (first-order terms), 
  \lambdaterms\ are viewed as $\alpha$\nb-equivalence classes of pseudo-term representations with names for bound variables.
  For \lambdaterms, $\sbetared$ denotes \betareduction, and $\slobetared$ \lo\ \betareduction.

  A \betareduction\ redex in a \lambdaterm~$\alter$ is called \emph{\lo} if it is to the left, or outside of any other redex in $\alter$.
  The \emph{\lo\ reduction strategy} for the \lambdacalculus\ is a 1-step strategy that,
  for a given \lambdaterm~$\alter$ contracts the \lo\ \betaredex\ in $\alter$.

\paragraph{Termination$\hspace*{0.75pt}$/$\hspace*{0.75pt}$strong normalization of rewrite relations.} 
  Let $\sred$ be the rewrite relation (of a TRS or of \lambdacalculus), and let $\bter$ be a term.  
  We say that $\sred$ \emph{terminates from $\ater$}, 
     and also that \emph{$\sred$ is strongly normalizing from $\bter$}
  if there is no infinite rewrite sequence from $\ater$
  (and consequently all sufficiently long rewrite sequences from $\bter$ lead to a normal form with respect to $\sred$). 
  We say that $\sred$ \emph{terminates}, and also that \emph{is strongly normalizing}, if $\sred$ does not enable infinite rewrite sequences.

\section{Simulation of leftmost-outermost rewrite sequences}
  \label{sec:lo-simulation}

We start with the formal definition of first-order representations of \lambdaterms,
called \lambdaterm\ representations,
before describing a \TRS\ for simulating \lo\ \betareduction\ on \lambdaterm\ representations.

\begin{definition}[\lambdaterm\ representations, denoted \lambdaterms]%
    \label{def:ltermrep}
  Let
    $
    \asiglambda 
      \,\defdby\, 
    \descsetexp{ \afovari{j} }{ j\in\nat }
      \cup 
    \setexp{ \sfolapp }  
      \cup
    \descsetexp{ \sfolabs{\afovari{j}} }{ j\in\nat } 
    $
  be the signature that
  consists of the \emph{variable} symbols $\afovari{j}$, with $j\in\nat$, which are constants (nullary function symbols),
  the binary \emph{application symbol}~$\sfolapp$,
  and the unary \emph{named abstraction} symbols $\sfolabs{\afovari{j}}$, for $j\in\nat$.
  
  Now by a \emph{\lambdaterm\ representation} (a \emph{(first-order) representation of a \lambdaterm})
  we mean a ground term in $\termsover{\asiglambda}$. 
  A \lambdaterm\ representation $\ater$ denotes, by reading its symbols in the obvious way,
  and interpreting occurrences of variable symbols $\afovari{j}$ that are not bound,
  as the variable names $\avari{j}$, a unique \lambdaterm\ $\denlter{\ater}$.

\end{definition}

\begin{example}
  $\folabs{\afovari{0}}{\afovari{0}}$,
  $\folabs{\afovari{1}}{\folabs{\afovari{2}}{\afovari{1}}}$, and
  $\folabs{\afovari{0}}{\folabs{\afovari{1}}{\folabs{\afovari{2}}{\folapp{\folapp{\afovari{0}}{\afovari{1}}}{\folapp{\afovari{1}}{\afovari{2}}}}}}$
  are \lambdaterm\ representations that denote the \lambdaterms\
  $I = \labs{\avar}{\avar}$,
  $K = \labs{\avar\bvar}{\avar}$,
  and 
  $S = \labs{\avar\bvar\cvar}{\lapp{\lapp{\avar}{\cvar}}{(\lapp{\bvar}{\cvar})}}$,
  respectively.
\end{example}
    
Below we formulate a \TRS\ that facilitates the simulation, on \lambdaterm\ representations,
of the evaluation of \lambdaterms\ according to the \lo\ strategy.
We introduce this TRS as a motivation for a similar simulation \TRS\ on super\-com\-bi\-na\-tor-based \lambdaterm\ representations
that is introduced later in Definition~\ref{def:losimTRS}, and that will be crucial for obtaining the linear depth-increase result. 
While the \TRS\ is designed to reason about \lo\ rewrite sequences, it actually permits the simulation of generalizations of the \lo\ rewrite sequences:
\betaredex{es} may also be contracted if they are \lo\ in right subterms immediately below stable parts of the term. 
This is because the search process for \lo\ redexes will be initiated again in parallel positions just below stable spines. 
We will therefore use the abbreviation `lop' in symbol names to hint at 
the \nondeterministic\ evaluation strategy `\underline{l}eftmost-\underline{o}utermost, iterated in \underline{p}arallel positions below stable parts of the term'. 

The idea behind the simulation \TRS\ is as follows.
The process is started on a term $\lopstart{\ater}$, where $\ater$ is a \lambdaterm\ representation that is to be evaluated.
First $\lopstart{\ater}$ is initialized to $\lopn{0}{\ater}$ (via the rule ($\scriptinit$)), 
where the index (which here is $0$) will be used as a lower bound for yet unused variable indices. 
Then a term $\ateri{0}$ with an outermost applications in an expression $\lopn{n}{\ateri{0},\bteri{1},\ldots,\bteri{n}}$
is uncurried into a representing expression with a stack of applications
(by steps of the rule ($\scriptdescendinfolapp$))
when descending over applications along the spine of the term 
until a variable or an abstraction is encountered 
(detected by one of the rules ($\scriptdescendinfolabs$), ($\scriptvar_0$), or ($\scriptvar_{n+1}$)).
If an abstraction occurs, and the expression contains an argument for this abstraction,
the representation of a \lo\ \betaredex\ has been detected, 
which is then contracted by a step corresponding to a \betacontraction\
(applying the rule ($\scriptcontract$)); the evaluation continues similarly from there on.
If there is no argument for such an abstraction, 
then it is part of a head normal form context,
and the evaluation 
                   descends into the abstraction (applying the rule ($\scriptdescendinfolabs$)) 
to proceed recursively on the subterm.
If a variable occurs on the left end of the spine (detected by one of the rules ($\scriptvar_0$) or ($\scriptvar_{n+1}$)), 
then a head normal form context has been detected,
which consists of a single variable (in case the applicable rule is ($\scriptvar_0$)),
or of the variable together with the recently uncurried applications (in case the applicable rule is ($\scriptvar_{n+1}$)).
In the first case evaluation stops in the present subterm,
whereas in the second case the simulating evaluation can continue (after applying ($\scriptvar_{n+1}$)), 
possibly in parallel, from any immediate subterm of one of the recently uncurried applications.
The rules:\label{def:losim:TRS:ltermreps}%
\begin{align*}
  \lopstart{\avar}
    \;\; & \sred \;\;
  \lopn{0}{\avar}
  \tag{$\scriptinit$}
  \\
  \lopn{n}{\folapp{\avar}{\bvar},\bvari{1},\ldots,\bvari{n}}
    \;\; & \sred \;\;
  \lopn{n+1}{\avar,\bvar,\bvari{1},\ldots,\bvari{n}}  
  \tag{$\scriptdescendinfolapp$}
  \displaybreak[0]\\
  \lopn{0}{\folabs{\afovari{j}}{\avar}}
    \;\; & \sred \;\;
  \folabs{\afovari{j}}{\lopn{0}{\avar}}
  \tag{$\scriptdescendinfolabs$}
  \displaybreak[0]\\
  \lopn{n+1}{\folabs{\afovari{j}}{\avar},\bvari{1},\bvari{2},\ldots,\bvari{n+1}}
    \;\; & \sred \;\;
  \lopn{n}{\subst{\avar}{\afovari{j}}{\bvari{1}},\bvari{2},\ldots,\bvari{n+1}}
  \tag{$\scriptcontractn{n+1}$}
  \displaybreak[0]\\ 
  \lopn{0}{\afovari{j}}
    \;\; & \sred \;\;
  \afovari{j}
  \tag{$\scriptvar_0$}
  \displaybreak[0]\\
  \lopn{n+1}{\afovari{j},\bvari{1},\ldots,\bvari{n+1}}
    \;\; & \sred \;\;
  \folapp{\ldots{\folapp{\afovari{j}}{\lopn{0}{\bvari{1}}}}\ldots}
         {\lopn{0}{\bvari{n+1}}}    
  \tag{$\scriptvar_{n+1}$}
\end{align*}
have to be extended with appropriate rules for $\ssubst$ that implement capture-avoiding substitution,
which induce a rewrite relation $\ssubstred$. 
We do not provide those rules here, because the rewrite system above only serves us as a stepping stone
for a similar rewrite system in Section~\ref{sec:lo-simulation:lTRSs}
that operates on supercombinator representations of \lambdaterms\ (\lTRSs)
where substitution can be organized as context-filling.


Based on the simulation \TRS, we denote by $\scontractred$ the rewrite relation that is induced by the rule scheme ($\scriptcontractn{n+1}$) for $n\in\nat$.
It defines steps that initiate the simulation of a \betareduction\ step which then proceeds with $\ssubstred$ steps 
  that carry out the substitution in the contraction of the \betaredex.
By $\sinitred$, $\sdescendinfolappred$, $\sdescendinfolabsred$, and $\svarred$
we designate the rewrite relations that are induced by the rules
$(\scriptinit)$, $(\scriptdescendinfolapp)$, $(\scriptdescendinfolabs)$, and $(\scriptvar_n)$ for some $n\in\nat$, respectively.
By $\ssearchred$ we denote the union of $\sinitred$, $\sdescendinfolappred$, $\sdescendinfolabsred$, and $\svarred$,
because they organize the search 
for the next \loredex\ or of an outermost redex. 
Finally, we denote by $\slopsimred$ the rewrite relation that is induced by the entire TRS.


The labels for $\scontractred$ and $\ssearchred$ are motivated as follows:
In a $\scontractred$ step the representation of a \loredex\ is contracted,
or the representation of a `stacked' outermost redex that is \lo\ below a stable
part of the term (and  that is bound to become
a \loredex\ at some later stage, at least if the term has a normal form). 
And a $\ssearchred$ step is part of the search in the term 
for the representation of the next \loredex\ or of an outermost redex
that is bound to become a \loredex\ later.
 

\begin{example}\label{ex:lopsimred:ltermrep}
  We consider the \lambdaterm\
  $\alter = \labs{x}{\lapp{(\labs{y}{y})}{(\lapp{(\labs{z}{\labs{w}{\lapp{w}{z}}})}{x})}}$.
  Evaluating $\alter$ with the \lo\ rewrite strategy, symbolized by the rewrite relation $\slobetared$, gives rise to the rewrite sequence:
  \begin{gather}\label{rewseq1:ex:lopsimred:ltermrep}
    \labs{x}{\underline{\lapp{(\labs{y}{y})}{(\lapp{(\labs{z}{\labs{w}{\lapp{w}{z}}})}{x})}}}
      \;\lobetared\;
    \labs{x}{\underline{\lapp{(\labs{z}{\labs{w}{\lapp{w}{z}}})}{x}}}
      \;\lobetared\;
    \labs{x}{\labs{w}{\lapp{w}{x}}}    
  \end{gather}
  where the underlinings symbolize the \betaredex{es} that are contracted in the next step. 
  The term:
  \begin{equation*}
    \ater =
    \folabs{\afovari{0}}
           {\folapp{\folabs{\afovari{1}}{\afovari{1}}}
                   {\folapp{\folabs{\afovari{2}}
                                   {\folabs{\afovari{3}}
                                           {\folapp{\afovari{3}}{\afovari{2}}}}}
                           {\afovari{0}}}}
  \end{equation*}                         
  denotes $\alter$, that is, $\denlter{\ater} = \alter$; other variable names are possible modulo `\alphaconversion'.                    
  Simulating this \lo\ rewrite sequence by means of the simulation TRS above
  %
  \begin{alignat*}{2}
    \lopstart{\ater}
     & \;\initred\; & &
       \lopn{0}{\folabs{\afovari{0}}
                         {\folapp{\folabs{\afovari{1}}{\afovari{1}}}
                                 {\folapp{\folabs{\afovari{2}}
                                                 {\folabs{\afovari{3}}
                                                         {\folapp{\afovari{3}}{\afovari{2}}}}}
                                         {\afovari{0}}}}}
     \\
     & \;\descendinfolabsred\; & & 
       \folabs{\afovari{0}}
              {\lopn{0}{\folapp{\folabs{\afovari{1}}{\afovari{1}}}
                                 {\folapp{\folabs{\afovari{2}}
                                                 {\folabs{\afovari{3}}
                                                         {\folapp{\afovari{3}}{\afovari{2}}}}}
                                         {\afovari{0}}}}}
    \displaybreak[0]\\
     & \;\descendinfolappred\; & & 
       \folabs{\afovari{0}}
              {\lopn{1}{\folabs{\afovari{1}}{\afovari{1}},
                         {\folapp{\folabs{\afovari{2}}
                                         {\folabs{\afovari{3}}
                                                 {\folapp{\afovari{3}}{\afovari{2}}}}}
                                 {\afovari{0}}}}}
    \displaybreak[0]\\
    & \;\contractred\; & & 
       \folabs{\afovari{0}}
              {\lopn{0}{\subst{\afovari{1}}
                                {\afovari{1}}
                                {\folapp{\folabs{\afovari{2}}
                                                {\folabs{\afovari{3}}
                                                        {\folapp{\afovari{3}}{\afovari{2}}}}}
                                        {\afovari{0}}}}}
    \displaybreak[0]\\
    & \;\substred\; & & 
       \folabs{\afovari{0}}
              {\lopn{0}{\folapp{\folabs{\afovari{2}}
                                         {\folabs{\afovari{3}}
                                                 {\folapp{\afovari{3}}{\afovari{2}}}}}
                                 {\afovari{0}}}}
    \displaybreak[0]\\
    & \;\descendinfolappred\; & & 
       \folabs{\afovari{0}}
              {\lopn{1}{\folabs{\afovari{2}}
                                 {\folabs{\afovari{3}}
                                         {\folapp{\afovari{3}}{\afovari{2}}}},
                          \afovari{0}}}
    \displaybreak[0]\\
    & \;\contractred\; & & 
       \folabs{\afovari{0}}
              {\lopn{0}{\subst{\folabs{\afovari{3}}
                                        {\folapp{\afovari{3}}{\afovari{2}}}}
                                {\afovari{2}}
                                {\afovari{0}}}}
    \displaybreak[0]\\
    & \;\substmred\; & & 
       \folabs{\afovari{0}}
              {\lopn{0}{\folabs{\afovari{3}}
                                 {\folapp{\afovari{3}}{\afovari{0}}}}}
    \displaybreak[0]\\
    & \;\descendinfolabsred\; & & 
       \folabs{\afovari{0}}
              {\folabs{\afovari{3}}
                      {\lopn{0}{\folapp{\afovari{3}}{\afovari{0}}}}}
    \displaybreak[0]\\
    & \;\descendinfolappred\; & & 
       \folabs{\afovari{0}}
              {\folabs{\afovari{3}}
                      {\lopn{1}{\afovari{3},\afovari{0}}}}
    \displaybreak[0]\\
    & \;\varnred{1}\; & & 
       \folabs{\afovari{0}}
              {\folabs{\afovari{3}}
                      {\folapp{\afovari{3}}{\lopn{0}{\afovari{0}}}}}
    \\
    & \;\varnred{0}\; & & 
       \folabs{\afovari{0}}
              {\folabs{\afovari{3}}
                      {\folapp{\afovari{3}}{\afovari{0}}}}
  \end{alignat*}        
  Note that the $\scontractred$ steps indeed initiate, and the $\ssubstred$ steps complete,
  the simulation of corresponding \betareduction\ steps in the $\slobetared$ rewrite sequence on \lambdaterms\ above,
  while the other steps organize the search for the next (\lambdaterm\ representation of a) \lo\ \betaredex.
  The $\slobetared$ rewrite sequence \eqref{rewseq1:ex:lopsimred:ltermrep}
  can be viewed as the projection of the $\lopsimred$ rewrite sequence above
  under an extension of the denotation operation $\denlter{\cdot}$ on \lambdaterm\ representations yielding \lambdaterms\
  (which works out substitutions, and interprets uncurried application expressions $\lopn{n}{\ater,\bteri{1},\ldots,\bteri{n}}$ appropriately).
  Hereby $\scontractred$ steps project to $\slobetared$ steps,
  but all other steps vanish under the projection.
  %
\end{example}

While the \TRS\ above facilitates the faithful representation of \lo\ rewrite sequences on \lambdaterms\
(which can be formulated formally analogous to Proposition~\ref{prop:lifting:lobeta:lo-losim:rewseqs}, see page~\pageref{lem:lifting}),
it does not lend itself well to the purpose of proving the linear-depth-increase result.
This is because it is not readily clear which invariant for reducts $\bter$ of a term $\ater$ 
in rewrite sequences $\arewseq \funin \ater \lopsimmred \bter \lopsimred \cter$
could make it possible to prove that the depth increase in the final step of $\arewseq$ 
is bounded by a constant $d$ that only depends on the initial term $\ater$ of the sequence (but not on $\bter$).
In the next section, however, we develop a concept that can overcome this problem.
We define extensions of first-order \lambdaterm\ representations
in which the abstraction parts of representations of \lo\ \betaredex{es} 
are built up from contexts that trace back to contexts in the initial term of the rewrite sequence.
This will guarantee that after a \lo\ \betareduction\ rewrite sequence $\alteri{0} \lobetaredn{n} \alteri{n}$
a scope part of the abstraction part $\labs{\cvar}{\clter}$ of the next \lo\ \betaredex\ $\lapp{(\labs{\cvar}{\clter})}{\dlter}$ in $\alteri{n}$ 
does already occur in $\alteri{0}$.

\section{\protect\lTRS\ representations of lambda terms}
  \label{sec:lTRSs}

We now introduce \lTRS{s} as orthogonal \TRSs\ that are able to represent \lambdaterms.
The basic idea is that, for a \lambdaterm~$\alter$, function symbols that are called `scope symbols'
are used to represent abstraction scopes.
Hereby the scope of an abstraction $\labs{\avar}{\clter}$ in $\alter$ 
includes the abstraction $\sslabs{\avar}$ and all occurrences of the bound variable $\avar$, but may leave room
for subterms in $\clter$ without occurrences of $\avar$ bound by the abstraction.
For example, the \lambdaterm~$\labs{\avar}{\lapp{\lapp{\lapp{\cvar}{\avar}}{\bvar}}{\avar}}$
may be denoted as the term $\afoscope{\cvar,\bvar}$
where the binary scope symbol $\afoscopesym$ represents the scope context $(\labs{\avar}{\lapp{\lapp{\lapp{\holei{1}}{\avar}}{\holei{2}}}{\avar}})$.
In our formalization of \lambdaterm\ representations the free variables $\cvar$ and $\bvar$
will be replaced by variable constants, yielding for example the term  $\afoscope{\afovari{2},\afovari{1}}$.
Furthermore, scopes are assumed to be strictly nested.
Every scope symbol defines a rewrite rule that governs the behavior of the application of the scope to an argument.
In the case of the \lambdaterm~$\labs{\avar}{\lapp{\lapp{\lapp{\cvar}{\avar}}{\bvar}}{\avar}}$
this leads to the first-order rewrite rule 
$\folapp{\afoscope{\cvar,\bvar}}{\avar} \red \folapp{\folapp{\folapp{\cvar}{\avar}}{\bvar}}{\avar}$
for the scope symbol $\afoscopesym$ that corresponds to the \lambdaterm\ scope context $(\labs{\avar}{\lapp{\lapp{\lapp{\holei{1}}{\avar}}{\holei{2}}}{\avar}})$.
Such a translation facilitates a correspondence between \betareduction\ steps in the \lambdacalculus,
and first-order term rewriting steps on terms with adequately defined scope symbols. In the example here the correspondence is between the steps: 
\begin{alignat}{3}
  \lapp{(\labs{\avar}{\lapp{\lapp{\lapp{\cvar}{\avar}}{\bvar}}{\avar}})}
       {\alter} 
    & \;\; \betared \;\;  & & 
  \lapp{\lapp{\lapp{\cvar}{\alter}}{\bvar}}{\alter}     
    \tag*{(\betareduction\ in the \lambdacalculus),}
  \\
  \folapp{\afoscopesym}{\ater}
    & \;\; \red \;\; & &
  \folapp{\folapp{\folapp{\cvar}{\ater}}{\bvar}}{\ater}
    \tag*{(application of the corresponding \lTRS-rule),}
\end{alignat}
provided that the \lTRS\nb-term $\ater$ represents the \lambdaterm~$\alter$.

\lTRSs\ are \TRSrepresentation{s} of systems of supercombinators that are obtained by the \lambdalifting\ transformation.
I have been introduced to these \lambdaterm\ representations by orthogonal \TRSs\ by Vincent van Oostrom
(personal communication, in the framework of the NWO-research project `Realising Optimal Sharing',
 and our collaboration on `nested term graphs'~\cite{grab:oost:2015}). 
He strongly shaped my understanding of them, and pointed me to the studies of optimal reduction for weak \betareduction\ 
(\betareduction\ outside of abstractions or in `maximal free' subexpressions)
by Blanc, L\'{e}vy, and Maranget \cite{blan:levy:mara:2005}.
Also, he encouraged work by Balabonski~\cite{bala:2012} on characterizations of optimal-sharing implementations for weak \betareduction\ 
by term labelings.
Later I discovered the direct connection with `fully-lazy \lambdalifting',
which was introduced in the early 1980-ies by Hughes \cite{hugh:1982:report,hugh:1982}. 

\begin{definition}[\lTRS{s}]\label{def:lTRS}
  A \emph{\lTRS}  
  is a pair~$\alTRS = \pair{\asig}{\arules}$, 
  where $\asig$ is a signature containing the binary application symbol~$\sfolapp$,
  and the \emph{scope symbols} in $\asigmin \defdby \asig \setminus \setexp{\sfolapp}$,
  and where $\arules = \descsetexp{ \ruleof{\afoscopesym} }{ \afoscopesym\in\asigmin }$
  consists of the \emph{defining rules} $\ruleof{\afoscopesym}$ for scope symbols $\afoscopesym\in\asigmin$ with arity~$k$
  that are of the form: 
  \begin{align*}
    (\ruleof{\afoscopesym}) \;\;\;\;
    \folapp{\afoscope{\avari{1},\ldots,\avari{k}}}{\bvar}
      \;\sred\;
    \afoscopecxtap{\avari{1},\ldots,\avari{k},\bvar}
  \end{align*}
  with $\afoscopecxt$ a $(k+1)$\nb-ary context of $\alTRS$ 
  that is called the \emph{scope context} for $\afoscopesym$.
  For scope symbols $\afoscopesym,\bfoscopesym\in\asigmin$ 
  we say that $\afoscopesym$ \emph{depends on} the scope symbol $\bfoscopesym$,
  denoted by $\afoscopesym \dependson \bfoscopesym$,
  if $\bfoscopesym$ occurs in the scope context $\afoscopecxt$ for $\afoscopesym$.
  We say that $\alTRS$ is \emph{finitely nested} 
  if the converse relation of $\sdependson$, the \emph{nested-into} relation $\sisnestedinto$, is well-founded,
  or equivalently (using the axiom of dependent choice),
  if there is no infinite chain of the form
  $\afoscopesymi{0} \dependson \afoscopesymi{1} \dependson \afoscopesymi{2} \dependson \ldots$
  on scope symbols $\afoscopesymi{0},\afoscopesymi{1},\afoscopesymi{2},\ldots\in\asigmin$.  
\end{definition}

\begin{example}\label{ex:lTRS}
  Let $\alTRS = \pair{\asig}{\arules}$ be the \lTRS\ 
  with $\asigmin = \setexp{ \afoscopesym, \bfoscopesym, \cfoscopesym, \dfoscopesym }$,
  where $\arityof{\afoscopesym} = 2$, $\arityof{\bfoscopesym} = \arityof{\cfoscopesym} = 0$, and $\arityof{\dfoscopesym} = 1$,
  and the following set $\arules$ of rules:
  \begin{center}
    $
  \begin{alignedat}{4}
    (\ruleof{\afoscopesym}) & \;\;\;\; & 
    \folapp{\afoscope{\avari{1},\avari{2}}}{\avar}
      & {} \red 
    \folapp{\avari{1}}{\folapp{\avari{2}}{\avar}}
    & \qquad\qquad
    (\ruleof{\cfoscopesym}) & \;\;\;\; &
    \folapp{\cfoscopesym}{\avar} 
      & {} \red
    \dfoscope{\avar}
    \\
    (\ruleof{\bfoscopesym}) & \;\;\;\; & 
    \folapp{\bfoscopesym}{\avar}
      & {} \red
    \avar
    &
    (\ruleof{\dfoscopesym}) & \;\;\;\; &
    \folapp{\dfoscope{\avari{1}}}{\avar}
      & {} \red
    \folapp{\avar}{\avari{1}}        
  \end{alignedat}
    $
  \end{center}  
  This finite \lTRS\
  is also finitely nested, as the depends-on relation consists only of a single link: $\cfoscopesym \dependson \dfoscopesym$.  
  It facilitates to denote the \lambdaterm~$\alter$ in Example~\ref{ex:lopsimred:ltermrep},
  see the expansion of $\afoscope{\bfoscopesym,\cfoscopesym}$ in Example~\ref{ex:lTRS:expred} below. 
\end{example}

In order to explain how \lTRS\ terms 
                                     denote \lambdaterm\ representations, 
we introduce, for every \lTRS~$\alTRS$, 
an expansion \TRS\ that makes use of the defining rules for the scope symbols in $\alTRS$. 
Then `denoted \lambdaterm\ representations' will be defined as normal forms of terms in the expansion \TRS.
It uses function symbols~$\sexpandi{i}$ with parameters $i$ 
for expanding a \lTRS\ term in a top--down manner.
Thereby the indices $i$ are used to guarantee that when an abstraction $\lambda\afovari{i}$ is created 
the indexed variable name $\afovari{i}$ is different from that of all abstractions $\lambda\avari{j}$ that have been created above it. 
In this way the arising \lambdaterm\ representation will be uniquely named at vertical positions.

\begin{definition}[expansion \TRS\ for a \lTRS]\label{def:expandTRS}
  Let $\alTRS = \pair{\asig}{\arules}$ be a \lTRS. 
  The \emph{expansion \TRS~$\,\expandTRSwrt{\alTRS} = \pair{\asigexp\cup\vars}{\rulesexp}$ for $\alTRS$}
  has the signature $\asigexp \,\defdby\, \asig \cup \asiglambda \cup \asigexpand$
  with 
  $\asigexpand \defdby \descsetexp{ \sexpandi{i} }{ i\in\nat }$
  where $\sexpandi{i}$ is unary 
  for $i\in\nat$, and $\asigmin \cap (\asiglambda \cup \asigexpand) = \emptyset$,
  and its set of rules $\rulesexp$ consists of the rules:
  \begin{center}
  $  
  \begin{aligned}
    \expandi{i}{\folapp{\avari{1}}{\avari{2}}}
      & \;\red\;
    \folapp{\expandi{i}{\avari{1}}}{\expandi{i}{\avari{2}}}  
    \\
    \expandi{i}{\afoscope{\avari{1},\ldots,\avari{k}}}
      & \;\red\;
    \folabs{\afovari{i}}{\expandi{i+1}{\afoscopecxtap{\avari{1},\ldots,\avari{k},\afovari{i}}}}
    & & \hspace*{1.5ex}\text{(where $\afoscopecxt$ is the scope context for $\afoscopesym$)}
    \\
    \expandi{i}{\folabs{\afovari{j}}{\avar}}
      & \;\red\;
    \folabs{\afovari{j}}{\expandi{\max\setexp{i,j}+1}{\avar}} 
    \\
    \expandi{i}{\afovari{j}} 
      & \;\red\;
    \afovari{j}
  \end{aligned}
  $
  \end{center}
  By $\sexpred$ we denote the rewrite relation of $\expandTRSwrt{\alTRS}$.
\end{definition}

\begin{lemma}\label{lem:expTRS:orthogonal:UN}
  The expansion \TRS~$\,\expandTRSwrt{\alTRS}$ 
  of a \lTRS~$\alTRS = \pair{\asig}{\arules}$  
  is an orthogonal TRS.
  Hence its rewrite relation $\sexpred$ is confluent,
  and normal forms of terms, whenever they exist, are unique.
\end{lemma}

Since expansion~\TRSs\ are orthogonal \TRSs, finite or infinite normal forms are unique.
Furthermore they are constructor~\TRSs, i.e.\ they have rules whose right-hand sides are guarded by constructors.
This can be used to show that all terms in an expansion~\TRS\ rewrite 
to a unique finite or infinite normal form.  

%
%
%
%
%
%
%
%
%
%

\begin{definition}[\lambdaterm\ representations denoted by \protect\lTRS\nb-terms]
  Let $\alTRS = \pair{\asig}{\arules}$ be a \lTRS. 
  For a term $\ater\in\termsover{\asig}$ we denote by $\denlterrepwrt{\alTRS}{\ater}$ 
  the finite or infinite $\sexpred$\nb-normal form of the term $\expandi{0}{\ater}$ in $\expandTRSwrt{\alTRS}$.
  If it is a \lambdaterm\ representation, we say that
  $\denlterrepwrt{\alTRS}{\ater}$ is the \emph{denoted \lambdaterm\ representation} of $\ater$,
  and write $\denlterwrt{\alTRS}{\ater}$ for the \lambdaterm~$\denlter{\denlterrepwrt{\alTRS}{\ater}}$.
\end{definition}

\begin{example}\label{ex:lTRS:expred}
  With the \lTRS~$\alTRS$ from Example~\ref{ex:lTRS}
  the \lambdaterm~$\alter$ in Example~\ref{ex:lopsimred:ltermrep} can be denoted
  as the term $\afoscope{\bfoscopesym,\cfoscopesym}$ expands 
  to a \lambdaterm\ representation of $\alter$
  (the final $\expmred$ step consists of two parallel $\expred$ steps):
  \begin{alignat*}{3}
    \expandi{0}{\afoscope{\bfoscopesym,\cfoscopesym}}
      & {} \;\expred\; {} & &
    \folabs{\afovari{0}}{\expandi{1}{\folapp{\bfoscopesym}{\folapp{\cfoscopesym}{\afovari{0}}}}}  
    \\
      & {} \;\expred\; {} & &
    \folabs{\afovari{0}}{\folapp{\expandi{1}{\bfoscopesym}}{\expandi{1}{\folapp{\cfoscopesym}{\afovari{0}}}}}  
    \displaybreak[0]\\
      & {} \;\expred\; {} & &
    \folabs{\afovari{0}}{\folapp{\folabs{\afovari{1}}{\expandi{2}{\afovari{1}}}}{\expandi{1}{\folapp{\cfoscopesym}{\afovari{0}}}}}  
    \displaybreak[0]\\
      & {} \;\expred\; {} & &
    \folabs{\afovari{0}}{\folapp{\folabs{\afovari{1}}{\afovari{1}}}{\expandi{1}{\folapp{\cfoscopesym}{\afovari{0}}}}}  
    \displaybreak[0]\\
      & {} \;\expred\; {} & &
    \folabs{\afovari{0}}{\folapp{\folabs{\afovari{1}}{\afovari{1}}}{\folapp{\expandi{1}{\cfoscopesym}}{\expandi{1}{\afovari{0}}}}} 
    \displaybreak[0]\\
      & {} \;\expred\; {} & &
    \folabs{\afovari{0}}{\folapp{\folabs{\afovari{1}}{\afovari{1}}}{\folapp{\expandi{1}{\cfoscopesym}}{\afovari{0}}}}  
    \displaybreak[0]\\
      & {} \;\expred\; {} & &
    \folabs{\afovari{0}}{\folapp{\folabs{\afovari{1}}{\afovari{1}}}{\folapp{\folabs{\afovari{1}}{\expandi{2}{\dfoscope{\afovari{1}}}}}{\afovari{0}}}}  
    \displaybreak[0]\\
      & {} \;\expred\; {} & &
    \folabs{\afovari{0}}
           {\folapp{\folabs{\afovari{1}}{\afovari{1}}}
                   {\folapp{\folabs{\afovari{1}}{\folabs{\afovari{2}}{\expandi{3}{\folapp{\afovari{2}}{\afovari{1}}}}}}
                           {\afovari{0}}}} 
    \displaybreak[0]\\
      & {} \;\expred\; {} & &
    \folabs{\afovari{0}}
           {\folapp{\folabs{\afovari{1}}{\afovari{1}}}
                   {\folapp{\folabs{\afovari{1}}{\folabs{\afovari{2}}{\folapp{\expandi{3}{\afovari{2}}}{\expandi{3}{\afovari{1}}}}}}
                           {\afovari{0}}}} 
    \displaybreak[0]\\
      & {} \;\expmred\; {} & &
    \folabs{\afovari{0}}
           {\folapp{\folabs{\afovari{1}}{\afovari{1}}}
                   {\folapp{\folabs{\afovari{1}}{\folabs{\afovari{2}}{\folapp{\afovari{2}}{\afovari{1}}}}}
                           {\afovari{0}}}}
  \end{alignat*}
  Hence
    $ \denlterrepwrt{\alTRS}{\afoscope{\bfoscopesym,\cfoscopesym}} =
      \folabs{\afovari{0}}
             {\folapp{\folabs{\afovari{1}}{\afovari{1}}}
                     {\folapp{\folabs{\afovari{1}}
                                     {\folabs{\afovari{2}}
                                             {\folapp{\afovari{2}}{\afovari{1}}}}}
                             {\afovari{0}}}}$.                         
  This \lambdaterm\ representation coincides with the term $\ater$ in Example~\ref{ex:lopsimred:ltermrep}
  `modulo \alphaconversion',
  and furthermore, for the denoted \lambdaterm\ it holds that
  $\denlterwrt{\alTRS}{\afoscope{\bfoscopesym,\cfoscopesym}}
     = \labs{x}{\lapp{(\labs{y}{y})}{(\lapp{(\labs{z}{\labs{w}{\lapp{w}{z}}})}{x})}}
     = \alter$.
\end{example}

\begin{proposition}\label{prop:lTRS-term:fin:nested:defines:lterrep}
  Let $\alTRS$ be a finitely nested \lTRS. 
  Then for every ground term $s$ of $\alTRS$, 
  $\denlterrepwrt{\alTRS}{\ater}$ is a finite ground term over $\asiglambda$, hence a \lambdaterm\ representation
  of the \lambdaterm\ $\denlterwrt{\alTRS}{\ater}$.
\end{proposition}

For proving termination and finiteness of the expansion process for terms and contexts in finitely nested \lTRSs,
we now define two measures: the `nesting depth' of scope symbols, and the `expansion size' of terms and of contexts.

\begin{definition}[nesting depth of a scope symbol, maximal nesting depth of contexts and terms]
  Let $\alTRS = \pair{\asig}{\arules}$ be a finitely nested \lTRS.   
  
  We define \emph{the nesting depth $\nestdepth{\afoscopesym}$} of a scope symbol $\afoscopesym\in\asigmin$
  by means of \wellfounded\ induction on $\sisnestedinto$, the converse of the nested-into relation $\sdependson$, as follows:
  \begin{equation*}
    \nestdepth{\afoscopesym}
      \defdby
        \begin{cases}
          0 & \text{ if } \lognot{\existsst{\bfoscopesym\in\asigmin}{\bigl(\, \afoscopesym \dependson \bfoscopesym \,\bigr)}} \punc{,}
          \\
          1 + \max \descsetexp{ \nestdepth{\bfoscopesym} }{ \afoscopesym \dependson \bfoscopesym,\, \bfoscopesym\in\asigmin }
            & \text{ if } \existsst{\bfoscopesym\in\asigmin}{\bigl(\, \afoscopesym \dependson \bfoscopesym \,\bigr)}  \punc{.}
        \end{cases}
  \end{equation*}
  Note that $\sisnestedinto$ is well-founded, since $\alTRS$ is finitely nested. 
  Furthermore the maximum in this clause is always taken over a finite set,
  because $\afoscopesym \dependson \bfoscopesym$ 
  means that $\bfoscopesym$ occurs in the scope context $\afoscopecxt$ of $\afoscopesym$, which is finite. 
 
  By the \emph{maximal nesting depth $\maxnestdepth{\acxt}$} of an \nary{n} context $\acxt\in\contextsnover{n}{\asig\cup\asiglambda}$,  
  we mean the maximal nesting depth of a scope symbol that occurs in $\acxt$.
  Similarly, by the \emph{maximal nesting depth $\maxnestdepth{\bter}$} of a term $\bter\in\termsover{\asig\cup\asiglambda}$  
  we mean the maximal nesting depth of a scope symbol that occurs in $\bter$.
\end{definition}

Next we introduce the `expansion size' of ground contexts (and thereby also of ground terms) over the signatures of an \lTRS~$\alTRS$, and the \lambdaterm\ representations.
We define it in such a way that the expansion size of a context $\acxt$ can later be recognized as the size of a normal form of $\expandi{0}{\acxt}$
in the expansion TRS~$\expandTRSwrt{\alTRS}$ when context holes, and remaining symbols $\sexpandi{j}$ are not counted.
First, however,
we will need this measure to show that every term $\expandi{0}{\acxt}$ has a normal form in $\expandTRSwrt{\alTRS}$ at all.

\begin{definition}[expansion size]
  Let $\alTRS = \pair{\asig}{\arules}$ be a finitely nested \lTRS.
  We define the \emph{expansion size~$\expsize{\acxt}$} of contexts $\acxt\in\contextsover{\asig\cup\asiglambda}$
  by induction on the structure of $\acxt$, thereby distinguishing the five possible cases of outermost symbols:  
  \begin{align*}
    \expsize{ \folapp{\acxti{1}}{\acxti{2}} }
      & \;\defdby\;
        1 + \expsize{\acxti{1}} + \expsize{\acxti{2}}  \punc{,}
    \\[-0.5ex]
    \expsize{ \afoscope{\acxti{1},\ldots,\acxti{k}} } 
      & \;\defdby\;
        1 + \expsize{ \afoscopecxtap{\holei{1},\ldots,\holei{k},\afovari{0}} }  
          + \sum_{i=1}^{k} 
              \expsize{\acxti{i}}
        & & \parbox{\widthof{(scope context for $\afoscopesym$ in $\alTRS$,)}}
                   {(where $\afoscopecxt$ is the
                    \\[-0.25ex]\phantom{(}%
                    scope context for $\afoscopesym$ in $\alTRS$),} 
    \displaybreak[0]\\[-0.5ex]
    \expsize{ \folabs{\afovari{j}}{\acxti{0}} }   
      & \;\defdby\;
        1 + \expsize{ \acxti{0} } 
    \\
    \expsize{ \afovari{j} }
      & \;\defdby\;
        1 
        & & \text{(for all $j\in\nat$)}  
    \\
    \expsize{\holei{j}}
      & \;\defdby\;
        0
        & & \text{(for all $j\in\nat$)} \punc{.}
  \end{align*}
  In particular, we apply \wellfounded\ induction on $\pair{\maxnestdepth{\acxt}}{\size{\acxt}}$
  with respect to the lexicographic order on $\nat\times\nat$,
  that is, induction on the nesting depth $\maxnestdepth{\acxt}$ of $\acxt$
  with a subinduction on the size $\size{\acxt}$ of $\acxt$.
  Note that, in particular, $\expsize{ \afoscope{\acxti{1},\ldots,\acxti{k}} }$ is \welldefined:
  for $\expsize{\afoscopecxtap{\holei{1},\ldots,\holei{k},\afovari{0}}}$ we can apply the induction hypothesis due to
  $\maxnestdepth{ \afoscopecxtap{\holei{1},\ldots,\holei{k},\afovari{0}} }
     <
   \maxnestdepth{ \afoscope{\acxti{1},\ldots,\acxti{k}} }$,
  since for all scope symbols $\bfoscopesym$ that occur in $\afoscopecxt$ it holds that $\bfoscopesym \dependson \afoscopesym$;
  and $\expsize{\acxti{i}}$ is \welldefined\ for $i\in\setexp{1,\ldots,k}$,
  because $\maxnestdepth{\acxti{i}} \le \maxnestdepth{\afoscope{\acxti{1},\ldots,\acxti{k}}}$, 
  and $\size{\acxti{i}} < \size{\afoscope{\acxti{1},\ldots,\acxti{k}}} $.
\end{definition}

\begin{lemma}\label{lem:expsize:cxts}
  $
  \expsize{ \acxtap{\acxti{1},\ldots,\acxti{n}} } 
    \; = \;
  \expsize{ \acxt }  
    + \sum_{i=1}^{n} 
        \expsize{\acxti{i}} 
  $      
  holds for contexts $\acxt\in\contextsnover{n}{\asig\cup\asiglambda}$ 
  and $\acxti{1},\ldots,\acxti{n}\in\contextsnover{m}{\asig\cup\asiglambda}$. 
\end{lemma}

\begin{proof}
  By a straightforward induction on the structure of the contexts~$\acxt\in\contextsnover{n}{\asig\cup\asiglambda}$.
\end{proof}

On the basis of these preparations we can now show that the expansion rewrite relation with respect to a finitely nested \lTRS\
always terminates (is strongly normalizing) on a term $\expandi{i}{\bter}$ with $i\in\nat$,
and with $\bter$ a term over the signature of the \lTRS\ and of \lambdaterm\ representations.

\begin{lemma}[termination of expansion in finitely nested \lTRSs]\label{lem:expred:terminates}
  Let $\alTRS = \pair{\asig}{\arules}$ be a {finitely nested} \lTRS. 
  Then the following statements hold for $\sexpred$ as defined in the expansion TRS~$\expandTRSwrt{\alTRS}$ of $\alTRS$:
  \begin{enumerate}[(i)]\setlength{\itemsep}{0ex}
    \item{}\label{it:1:lem:expred:terminates}
      $\sexpred$ terminates from $\expandi{i}{\acxt}$ for every $i\in\nat$, and every context $\acxt\in\contextsover{\asig\cup\asiglambda}$. 
    \item{}\label{it:2:lem:expred:terminates}  
      $\sexpred$ terminates from $\expandi{i}{\bter}$ for every $i\in\nat$, and every term $\bter\in\termsover{\asig\cup\asiglambda}$. 
  \end{enumerate}
\end{lemma}

\begin{proof}
  By inspection of the four rules of the expansion TRS~$\expandTRSwrt{\alTRS}$
  we find that in every step of the form 
  $\expandi{i}{\acxt} \expred \acxt'$ where $\acxt\in\contextsover{\asig\cup\asiglambda}$, $i\in\nat$,
                                            and $\acxt'\in\contextsover{\asig\cup\asiglambda\cup\asigexpand}$
  it holds for subexpressions $\expandi{i}{\bcxt}$ of $\acxt'$ 
  that $\expsize{\bcxt} < \expsize{\acxt}$. 
  This is immediate for the rules concerning application $\sfolapp$, named abstraction symbols $\sfolabs{\afovari{i}}$, and variable symbols $\afovari{i}$,
  for $i\in\nat$.  
  For the rule concerning scope symbols $\afoscopesym\in\asigmin$ 
  this can be checked by using Lemma~\ref{lem:expsize:cxts}.
  
  This reduction property justifies the induction step in a proof of statement~\eqref{it:1:lem:expred:terminates} of the lemma
  by induction on the expansion size $\expsize{\acxt}$ of contexts $\acxt\in\contextsover{\asig\cup\asiglambda}$.
  
  Statement~\eqref{it:2:lem:expred:terminates} of the lemma is a consequence of statement~\eqref{it:1:lem:expred:terminates},
  since terms over $\asig\cup\asiglambda$ can be viewed as contexts over $\asig\cup\asiglambda$ without hole occurrences. 
\end{proof}

Now that we know that the expansion process for terms over the signatures of \lTRSs, and of \lambdaterm\ representations
always terminate, and that $\sexpred$ normal forms are always unique, we introduce notation and a name for these normal forms.

\begin{definition}[expanded forms of contexts and terms]\label{def:expanded:form:cxts:terms}
  Let $\alTRS = \pair{\asig}{\arules}$ be a finitely nested \lTRS. 
  
  For every \nary{n} context $\acxt\in\contextsnover{n}{\asig\cup\asiglambda}$, where $n\in\nat$, and every term $\bter\in\termsover{\asig\cup\asiglambda}$ 
  we define \emph{the expanded form $\exprednfi{i}{\acxt}$ of $\acxt$}, and \emph{the expanded form $\exprednfi{i}{\bter}$ of $\bter$} 
  for all $i\in\nat$ by: 
  \begin{align*}
    \exprednfi{i}{\acxt} \: & \defdby\: \exprednf{\expandi{i}{\acxt}} \punc{,}
      &
    \exprednfi{i}{\ater} & \:\defdby\: \exprednf{\expandi{i}{\ater}} \punc{,}
  \end{align*}
  where $\exprednf{\bter}$ denotes the operation of taking \underline{the} $\expred$ normal form of $\bter$.
  This normal form is \welldefined, because $\sexpred$ normal forms of terms $\expandi{i}{\acxt}$ exist due to Lemma~\ref{lem:expred:terminates},
  and are unique due to Lemma~\ref{lem:expTRS:orthogonal:UN}.
\end{definition}

For reasoning with expanded forms of terms and context later in Section~\ref{sec:depth:increase}
we will need representations of the expanded form of arbitrary contexts,
and how this representation interacts with the context filling operation. 
The lemma below formulates the representation, and the subsequent lemma its property with respect to context filling.

\begin{lemma}[context representation of expanded forms of contexts]\label{lem:expanded:form:cxts}
  Let $\alTRS = \pair{\asig}{\arules}$ be a finitely nested \lTRS. 
  For every \nary{n} context $\acxt\in\contextsnover{n}{\asig\cup\asiglambda}$ (over the signature of $\alTRS$ and \lambdaterm\ representations), and $i\in\nat$, 
  the expanded form of $\acxt$ has a representation: 
  \begin{equation}\label{eq:lem:expanded:form:cxts}
    \begin{split}
      \exprednfi{i}{\acxt}
        \: & \syntequal\:
      \bcxtap{\expandi{i_1}{\holei{j_1}},\ldots,\expandi{i_m}{\holei{j_m}}}  
      \\
        \: & \syntequal\:
      \bcxtap{\exprednfi{i_1}{\holei{j_1}},\ldots,\exprednfi{i_m}{\holei{j_m}}}  
    \end{split}
  \end{equation}
  for some linear context $\bcxt\in\contextsnover{m}{\asiglambda}$ (over the signature of \lambdaterm\ representations), for $m\in\nat$,
  and $i_1,\ldots,i_m\in\nat$, and $j_1,\ldots,j_m\in\setexp{1,\ldots,n}$
  (therefore the context above on the right is \nary{n} just as $\acxt$).
\end{lemma}

\begin{proof}
  $\exprednfi{i}{\acxt}$ was \welldefined\ in Definition~\ref{def:expanded:form:cxts:terms},
    on the basis of Lemma~\ref{lem:expred:terminates} and Lemma~\ref{lem:expTRS:orthogonal:UN},
  as the unique $\sexpred$ normal form $\acxt'$ in $\expandTRSwrt{\alTRS}$ of $\expandi{i}{\acxt}$.
  Since $\acxt'$ is a normal form with respect to $\sexpred$,
  $\acxt'$ does not contain subexpressions of the form 
  $\expandi{j}{\folapp{\ccxti{1}}{\ccxti{2}}}$, 
  $\expandi{j}{\afoscope{\ccxti{1},\ldots,\ccxti{k}}}$,
  $\expandi{j}{\folabs{\afovari{l}}{\ccxti{0}}}$,
  or $\expandi{j}{\afovari{l}}$,
  with $j,k,l\in\nat$, and contexts $\ccxti{0},\ccxti{1},\ccxti{2},\ldots,\ccxti{k}$.
  The only possible occurrences of symbols $\sexpandi{j}$ in $\acxt$ must therefore be of the 
  `hole guarding' form
  $\expandi{j_1}{\cdots(\expandi{j_k}{\holei{l}})}$ with $j_1,\ldots,j_k,l\in\nat$.
  
  But proper stackings of symbols $\sexpandi{j}$ in $\acxt'$ are not possible. We argue as follows. 
  The form of the rules of $\expandTRSwrt{\alTRS}$ guarantees
  that if $\expandi{i}{\acxt} \expmred \ccxt$ holds, then $\ccxt$ does not contain symbols $\sexpandi{i}$ in nested positions.
  It follows that this also holds for the $\sexpred$ normal form $\acxt'$ of $\acxt$.
  
  Therefore the only  possible occurrences of symbols $\sexpandi{j}$ in $\acxt$ are of the 
  `linear hole guarding' form
  $\expandi{j}{\holei{l}}$ with some $j,l\in\nat$.
  As a consequence, $\acxt'$ can be written as of the form \eqref{eq:lem:expanded:form:cxts}
  for a linear context $\bcxt\in\contextsnover{m}{\asig}$, for $m\in\nat$,
  and $i_1,\ldots,i_m\in\nat$, and $j_1,\ldots,j_m\in\setexp{1,\ldots,n}$.
\end{proof}

\begin{lemma}[expanded form of filled contexts]\label{lem:expanded:form:cxtap}
  Let $\alTRS = \pair{\asig}{\arules}$ be a finitely nested \lTRS. 
  Then for all contexts $\acxt\in\contextsnover{n}{\asig\cup\asiglambda}$,
  and for all contexts $\acxti{1},\ldots,\acxti{n}\in\contextsnover{l}{\asig\cup\asiglambda}$ for some $l\in\nat$ it holds:
  \begin{equation}\label{eq:lem:expanded:form:cxtap}
    \exprednfi{i}{ \acxtap{ \acxti{1},\ldots,\acxti{n} } }
      \:\syntequal\:
    \bcxtap{ \exprednfi{i_1}{\acxti{j_1}}, \ldots, \exprednfi{i_m}{\acxti{j_m}} } 
  \end{equation}
  where $\bcxt$ is a linear context $\bcxt\in\contextsnover{m}{\asiglambda}$ (over the signature of \lambdaterm\ representations), for $m\in\nat$,
  that describes the expanded form $\exprednfi{i}{\acxt}$ of $\acxt$ via \eqref{eq:lem:expanded:form:cxts},
  for some $i_1,\ldots,i_m\in\nat$, and $j_1,\ldots,j_m\in\setexp{1,\ldots,n}$.
\end{lemma}

\begin{proof}
  By Lemma~\ref{lem:expanded:form:cxts} there is 
  a linear context $\bcxt\in\contextsnover{m}{\asiglambda}$, for $m\in\nat$,
  that describes the expanded form $\exprednfi{i}{\acxt}$ of $\acxt$ via \eqref{eq:lem:expanded:form:cxts},
  for some $i_1,\ldots,i_m\in\nat$, and $j_1,\ldots,j_m\in\setexp{1,\ldots,n}$,
  and hence with:
  \begin{equation*}
    \expandi{i}{\acxt}
      \:\expmred\;
    \bcxtap{\expandi{i_1}{\holei{j_1}},\ldots,\expandi{i_m}{\holei{j_m}}} \punc{.}
  \end{equation*} 
  By filling $\acxti{1},\ldots,\acxti{n}$ in the holes $\holei{1},\ldots,\holei{n}$ of $\acxt$
  in all contexts of this $\sexpred$ rewrite sequence
  we obtain another $\sexpred$ rewrite sequence
  that can furthermore be extended as follows:
  \begin{align*}
    \expandi{i}{\acxtap{\acxti{1},\ldots,\acxti{n}}}
      \: & \expmred\;
    \bcxtap{\expandi{i_1}{\acxti{j_1}},\ldots,\expandi{i_m}{\acxti{j_m}}}
    \\
      \: & \expmred\;
    \bcxtap{\exprednfi{i_1}{\acxti{j_1}},\ldots,\exprednfi{i_m}{\acxti{j_m}}} \punc{,}
  \end{align*}  
  by using the rewrite sequences $\expandi{i_l}{\acxti{j_l}} \expmred \exprednfi{i_l}{\acxti{j_l}}$
  for all $l\in\setexp{1,\ldots,m}$, which exist by the definition of $\exprednfi{i_l}{\acxti{j_l}}$.
  Since $\bcxt\in\contextsover{\asiglambda}$, it does not contain any symbol $\sexpandi{i}$, for $i\in\nat$.
  Therefore the resulting context $\bcxtap{\exprednfi{i_1}{\acxti{j_1}},\ldots,\exprednfi{i_m}{\acxti{j_m}}}$
  is indeed a $\sexpred$ normal form.
  In this way we have justified the form \eqref{eq:lem:expanded:form:cxtap} of $\exprednfi{i}{ \acxtap{ \acxti{1},\ldots,\acxti{n} } }$.
\end{proof}

Finally, in Section~\ref{sec:depth:increase} we will also need the following lemma.
It states that all expanded forms of terms over the signature of an \lTRS, and of the \lambdaterm\ representations
have the same unlabeled syntax tree, which entails that they have the same depth and size.

\begin{lemma}\label{lem:expanded:form:terms:cxts:difference}
  Let $\alTRS = \pair{\asig}{\arules}$ be a finitely nested \lTRS.
  Let $\acxt\in\contextsover{\asig\cup\asiglambda}$ be a context, let $\bter\in\termsover{\asig\cup\asiglambda}$ be a term
  over the signatures of $\alTRS$, and \lambdaterm\ representations.
  Let $i_1,i_2\in\nat$ with $i_1 \neq i_2$.  
  
  Then $\exprednfi{i_1}{\acxt}$ and $\exprednfi{i_2}{\acxt}$ have the same unlabeled syntax tree,
  and at a position $p$ they can only differ possibly in:
  \begin{enumerate}[(i)]\setlength{\itemsep}{0ex}
    \item{}\label{it:1:lem:expanded:form:terms:cxts:difference}
      a variable symbol $\afovari{j_1}$ at $p$ in $\exprednfi{i_1}{\acxt}$, and a variable symbol $\afovari{j_2}$ at $p$ in $\exprednfi{i_2}{\acxt}$,
    \item{}\label{it:2:lem:expanded:form:terms:cxts:difference}   
      an abstraction symbol $\sfolabs{\afovari{j_1}}$ at $p$ in $\exprednfi{i_1}{\acxt}$, and an abstraction symbol $\sfolabs{\afovari{j_2}}$ at $p$ in $\exprednfi{i_2}{\acxt}$,
    \item{}\label{it:3:lem:expanded:form:terms:cxts:difference}
      a symbol $\expandi{i'_1}{\holei{j}}$ at $p$ in $\exprednfi{i_1}{\acxt}$, and a symbol $\expandi{i'_2}{\holei{j}}$ at $p$ in $\exprednfi{i_2}{\acxt}$,
      where $j\in\setexp{1,\ldots,n}$.
  \end{enumerate} 
  Similarly,  $\exprednfi{i_1}{\bter}$ and $\exprednfi{i_2}{\bter}$ have the same syntax tree,
  and at a position $p$ they can only differ possibly as described 
  in items~\eqref{it:1:lem:expanded:form:terms:cxts:difference} and \eqref{it:2:lem:expanded:form:terms:cxts:difference}.
\end{lemma}

\begin{proof}
  From the rules of the expansion TRS $\expandTRSwrt{\alTRS}$
  we find that
  an expression $\expandi{i}{\acxt}$ is a redex of $\expandTRSwrt{\alTRS}$ irrespective of $i\in\nat$:
  indeed, it is a redex if and only if $\acxt$ is not a context hole. 
  The role of the index $i$ in a redex $\expandi{i}{\acxt}$ is only used to determine
  the index $j$ in variables $\afovari{j}$ or abstractions $\sfolabs{\afovari{j}}$ 
  that are possibly created by contracting this redex.  
  Therefore for every rewrite sequence $\expandi{i_1}{\acxt} \expmred \acxt'_1$ 
  has a `parallel' rewrite sequence $\expandi{i_2}{\acxt} \expmred \acxt'_2$
  where $\acxt'_1$ and $\acxt'_2$ have the same unlabeled syntax tree, and differ only possibly 
  in the aspects \eqref{it:1:lem:expanded:form:terms:cxts:difference}, \eqref{it:2:lem:expanded:form:terms:cxts:difference}, and \eqref{it:3:lem:expanded:form:terms:cxts:difference}
  of the lemma
  (with $\acxt'_1$ for $\exprednfi{i_1}{\acxt}$, and $\acxt'_2$ for $\exprednfi{i_2}{\acxt}$).
  Then this fact holds clearly also for the expanded forms $\exprednfi{i_1}{\acxt}$ and $\exprednfi{i_2}{\acxt}$ of $\acxt$.
\end{proof}

\section{Simulation of \protect\lo\ $\protect\beta$-reduction on \protect\lTRS-terms}
  \label{sec:lo-simulation:lTRSs}

We now adapt the \TRS\ for the simulation of \lo\ \betareduction\ rewrite sequences on \lambdaterm\ representations,
iterated in parallel positions (see page~\pageref{def:losim:TRS:ltermreps}), 
to a `\lopsimTRS' that facilitates such a simulation on terms of \lTRS{s}. 
For every \lTRS~$\alTRS$, we introduce a \lopsimTRS\ with 
rules that are similar as before but differ for steps involving abstractions.

A simulation starts on a term $\lopstart{\ater}$ where $\ater$ is a \lTRS\ ground term.
Therefore initially all abstractions are represented by scope symbols.
If under \lo\ evaluation a scope symbol $\afoscopesym$ is detected that does not have an argument,
then the top of the \lambdaabstraction\ it represents is stable (that is, it is part of a head normal form context).
Therefore it is expanded, giving rise to an abstraction representation~$\sfolabs{\afovari{i}}$,
and then \lo\ evaluation continues immediately below.
If, on the other hand, a scope symbol $\afoscopesym$ is detected that has at least one applicative argument term $\ater$,
then it represents the \lambdaabstraction\ of a \lo\ redex. In this case the \betareduction\ step for this \lo\ redex
is simulated by using the defining rule $\ruleof{\afoscopesym}$ of $\afoscopesym$ in the \lTRS, which involves filling the argument $\ater$
into the scope context $\afoscopecxt$ of $\afoscopesym$. 
The final term in an interated simulation of a \lo\ $\sbetared$ rewrite sequence 
                                                                           to a \lambdaterm\ normal form  
will be a normal form of the \lopsimTRS\
that is a \lambdaterm\ representation with named abstraction symbols, but without any scope symbols.
 
The changes in the adapted simulation TRS concern $\sdescendinfolabsred$ steps that descend into an abstraction,
and $\scontractred$ steps that simulate the reduction of \betaredex{es}.
In both cases prior to the step the pertaining abstractions are represented by terms with a scope symbol at the root.
Then in the steps the definition of the scope symbol in the underlying \lTRS\ is used.
Additional substitution rules are not necessary any more, because
the substitution involved in the contraction of a (represented) \betaredex\ 
can now be carried out by a single first-order rewrite step.
This is because such a step includes the transportation of the argument of a redex
into the scope context that defines the body of the abstraction. 
An additional parameter $i$ of the operation symbols $\slopni{n}{i}\,$ 
is used to prevent that any two nested abstractions refer to the same variable name,
safeguarding that rewrite sequences denote meaningful reductions on \lambdaterms.

\smallskip

\begin{definition}[\lopsimTRS\ for \lTRSs]%
    \label{def:losimTRS}
  Let $\alTRS = \pair{\asig}{\arules}$ be a \lTRS.
  The \emph{\lopsimTRS\ (\lopbetareduction\ simulation \TRS)~$\lopsimTRSwrt{\alTRS} = \pair{\asiglopsim}{\ruleslopsim}$ for $\alTRS$}
  has the signature $\asiglopsim \,\defdby\, \asig \cup \asiglambda \cup \asiglop$
  with 
    $
      \asiglop  
                     \,\defdby\, \setexp{ \slopstart } 
                                  \cup 
                                \descsetexp{ \slopni{n}{i} }{ n,i\in\nat }          
    $,                                               
  a signature of \emph{operation} symbols (for simulating \loreduction)
  consisting of the unary symbol $\slopstart$, 
  and the symbols $\slopni{n}{i}$ with arity $n+1$, for $n,i\in\nat$;
  the rule set $\ruleslopsim$ of $\lopsimTRSwrt{\alTRS}$ consists of the following (schemes of) rewrite rules,
  which are indexed by scope symbols $\afoscopesym\in\asigmin$,
  and where $\afoscopecxt$ is the scope context for scope symbol~$\afoscopesym\,$: 
  \begin{alignat*}{3}
    \lopstart{\avar}
      \; & \sred \;
    \lopni{0}{0}{\avar}    
      \tag*{$(\scriptinit)$}
    \\
    \lopni{n}{i}{\folapp{\avari{1}}{\avari{2}},\bvari{1},\ldots,\bvari{n}}
      \; & \sred \;
    \lopni{n+1}{i}{\avari{1},\avari{2},\bvari{1},\ldots,\bvari{n}}  
      \tag*{$\iap{(\scriptdescendinfolapp)}{n,i}$}
    \displaybreak[0]\\
    \lopni{0}{i}{\afoscope{\avari{1},\ldots,\avari{k}}}
      \; & \sred \;
    \folabs{\afovari{i}}{\lopni{0}{i+1}{\afoscopecxtap{\avari{1},\ldots,\avari{k},\afovari{i}}}}
      \tag*{$\bpap{(\scriptdescendinfolabs)}{i}{\afoscopesym}$}
    \displaybreak[0]\\
    \lopni{n+1}{i}{\afoscope{\avari{1},\ldots,\avari{k}},\bvari{1},\bvari{2},\ldots,\bvari{n+1}}
      \; & \sred \;              
    \lopni{n}{i}{\afoscopecxtap{\avari{1},\ldots,\avari{k},\bvari{1}},\bvari{2},\ldots,\bvari{n+1}} 
     \tag*{$\pbap{(\scriptcontract)}{\afoscopesym}{n+1}$}
    \displaybreak[0]\\  
    \lopni{0}{i}{\afovari{j}}
      \; & \sred \;
    \afovari{j}
      \tag*{$\bap{(\scriptvar)}{0,i}$}
    \displaybreak[0]\\
    \lopni{n+1}{i}{\afovari{j},\bvari{1},\ldots,\bvari{n+1}}
      \; & \sred \;
    \folapp{\ldots{\folapp{\afovari{j}}{\lopni{0}{i}{\bvari{1}}}}\ldots}
           {\lopni{0}{i}{\bvari{n+1}}}
      \tag*{$\bap{(\scriptvar)}{n+1,i}$}
  \end{alignat*}

  By $\slopsimred$ we denote the rewrite relation of $\lopsimTRSwrt{\alTRS}$.
  By $\scontractred$ we denote the rewrite relation that is induced by the rule scheme $\supap{(\scriptcontract)}{\afoscopesym}$
  where $\afoscopesym\in\asigmin$ ranges over scope symbols of $\alTRS$.  
  By $\ssearchred$ we denote  
  the rewrite relation that is induced by the other rules of $\lopsimTRSwrt{\alTRS}$.
  Similar as before,
  we denote by $\sinitred$, $\sdescendinfolappred$, $\sdescendinfolabsred$, and $\svarred$
  the rewrite relations that are induced by the rule schemes 
    $(\scriptinit)$, $\bap{(\scriptdescendinfolapp)}{n,i}$, $\pbap{(\scriptdescendinfolabs)}{\afoscopesym}{i}$, and $\bap{(\scriptvar)}{n}$, 
    respectively,
    where the parameters range over $\afoscopesym\in\asigmin$, and $n,i\in\nat$.
\end{definition} 


\begin{example}\label{ex:lopsimred}
  For the \lTRS~$\alTRS$ in Example~\ref{ex:lTRS}, 
  we reduce the term $\afoscope{\bfoscopesym,\cfoscopesym}$,
  which denotes the \lambdaterm~$\alter$ in Example~\ref{ex:lopsimred:ltermrep},
  in the \lopsimTRS~$\lopsimTRSwrt{\alTRS}$ for $\alTRS$:
  \begin{alignat*}{2}
    \lopstart{\afoscope{\bfoscopesym,\cfoscopesym}}
     & \;\,\initred\;\;\: & & 
       \lopni{0}{0}{\afoscope{\bfoscopesym,\cfoscopesym}}
     \\
     & \;\,\descendinfolabsred\;\;\: & &
       \folabs{\afovari{0}}{\lopni{0}{1}{\folapp{\bfoscopesym}{\folapp{\cfoscopesym}{\afovari{0}}}}}
     \displaybreak[0]\\
     & \;\,\descendinfolappred\;\;\: & &
       \folabs{\afovari{0}}{\lopni{1}{1}{\bfoscopesym,\folapp{\cfoscopesym}{\afovari{0}}}}
     \displaybreak[0]\\
     & \;\,\contractred\;\;\: & &
       \folabs{\afovari{0}}{\lopni{0}{1}{\folapp{\cfoscopesym}{\afovari{0}}}}
     \displaybreak[0]\\
     & \;\,\descendinfolappred\;\;\: & &
       \folabs{\afovari{0}}{\lopni{1}{1}{\cfoscopesym,\afovari{0}}}
     \displaybreak[0]\\
     & \;\,\contractred\;\;\: & &
       \folabs{\afovari{0}}{\lopni{0}{1}{\dfoscope{\afovari{0}}}}
     \displaybreak[0]\\
     & \;\,\descendinfolabsred\;\;\: & &
       \folabs{\afovari{0}}{\folabs{\afovari{1}}{\lopni{0}{1}{\folapp{\afovari{1}}{\afovari{0}}}}}
     \displaybreak[0]\\
     & \;\,\descendinfolappred\;\;\: & &
       \folabs{\afovari{0}}{\folabs{\afovari{1}}{\lopni{1}{2}{\afovari{1},\afovari{0}}}}
     \displaybreak[0]\\
     & \;\,\varnred{1}\;\;\: & &
       \folabs{\afovari{0}}{\folabs{\afovari{1}}{\folapp{\afovari{1}}{\lopni{0}{2}{\afovari{0}}}}}
     \displaybreak[0]\\
     & \;\,\varnred{0}\;\;\: & &
       \folabs{\afovari{0}}{\folabs{\afovari{1}}{\folapp{\afovari{1}}{\afovari{0}}}}
  \end{alignat*}
  We obtain an `\alphaequivalent' version of the \lambdaterm\ representation at the end
  of the simulated \lo\ reduction on \lambdaterm\ representations in Example~\ref{ex:lopsimred:ltermrep}.   
\end{example}

%
%
%
%

In order to define how terms in the \lopsimTRS\ denote \lambdaterm\ representations
we extend the expansion~\TRS\ from Definition~\ref{def:expandTRS} with rules
that deal with operation and named-abstraction symbols. 
We want expansion to be an `X-ray picture' of the current state of a term's evaluation. 
Therefore operation symbols $\slop$ and $\slopni{n}{i}$ will mainly be ignored.
However, indices $i$ in operation symbols $\slopni{n}{i}$ will be taken into account to,
to ensure unique naming at comparable positions in the expanded \lambdaterm\ representation.

\begin{definition}[expansion TRS for \lopsimTRS-terms, neglecting further evaluation]%
    \label{def:expandTRS:for:losimTRS}
  Let $\alTRS = \pair{\asig}{\arules}$ be a \lTRS. 
  The \emph{expansion \TRS} $\,\expandTRSlopsimwrt{\alTRS} = \pair{\asiglopsim \cup \asigexpand}{\rulesexp \cup \rulesexpprime}$ \emph{for \lopsimTRS-terms},
  which neglects further evaluation according to $\slopsimred$,
  has as its signature the union of the signature $\asiglopsim$ of $\lopsimTRSwrt{\alTRS}$ 
  and the signature $\asigexpand$ of $\expandTRSwrt{\alTRS}$,
  and as rules the rules $\rulesexp$ of $\expandTRSwrt{\alTRS}$,
  see both in Definition~\ref{def:expandTRS}, together with
  the set of rules $\rulesexpprime$ that consists of:
  \begin{align*}    
    \expandi{i}{\lopstart{\avar}}
      & \;\red\;
    \expandi{i}{\avar}
    \\
    \expandi{i}{\lopni{0}{j}{\avar}}
      & \;\red\;  
    \expandi{i'}{\avar}
    & & \text{for $i' \sdefdby \max\setexp{i,j}$} 
    \\
    \expandi{i}{\lopni{n+1}{j}{\avar,\bvari{1},\ldots\bvari{n+1}}} 
    & \;\red\;
    \expandi{i'}{\folapp{\cdots\folapp{\avar}{\bvari{1}}\ldots}{\bvari{n+1}}}
      & & \text{for $i' \sdefdby \max\setexp{i,j}$} 
  \end{align*}
  The rewrite relation of $\expandTRSlopsimwrt{\alTRS}$ will again be denoted by $\sexpred$.
\end{definition}

\begin{definition}[denoted \lambdaterm\ (representation), extended to \lopsimTRS-terms]
  Let $\alTRS = \pair{\asig}{\arules}$ be a \lTRS. 
  For terms $\ater\in\termsover{\asiglosim}$ in $\lopsimTRSwrt{\alTRS}$,  
  we also denote by $\denlterrepwrt{\alTRS}{\ater}$ 
  the finite or infinite $\sexpred$\nb-normal form of the term $\expandi{0}{\ater}$.
  If it is a \lambdaterm\ representation, then we say that $\denlterrepwrt{\alTRS}{\ater}$ is
  the \emph{denoted \lambdaterm\ representation}~of $\ater$,
  and we again write $\denlterwrt{\alTRS}{\ater}$ for the \lambdaterm~$\denlter{\denlterrepwrt{\alTRS}{\ater}}$.
\end{definition}


\section{Linear depth increase of \lo\ $\beta$-red.\ simulation}
  \label{sec:depth:increase}

In this section we establish that the depth increase of expanded terms
along an arbitrary rewrite sequences in a \lopsimTRS\ is linear in the number of $\scontractred$ steps.

In order to reason directly on terms and contexts of the \lopsimTRS, we define
the `expansion depth' of terms and contexts as the depth of their expanded forms without counting expansion symbols $\sexpandi{i}$. 

\begin{definition}[expansion depth of terms and contexts, expansion hole depth of contexts]
  Let $\alTRS = \pair{\asig}{\arules}$ be a \lTRS, 
  and let $\lopsimTRSwrt{\alTRS} = \pair{\asiglopsim}{\ruleslopsim}$ be the \lopsimTRS\ for $\alTRS$. 
  
  For terms $\bter\in\termsover{\asiglopsim,\vars}$ 
  and contexts $\acxt\in\contextsnover{n}{\asiglopsim,\vars}$, where $n\in\nat$,
  we define by:
  \begin{align*}
    \expdepth{\bter} 
      & {} \,\defdby\,
    \depthnotexp{ \exprednfi{0}{\bter} } \in\nat\cup\setexp{\infty} \punc{,}
    & 
    \expdepth{\acxt} 
      & {} \,\defdby\,
    \depthnotexp{ \exprednfi{0}{\acxt} } \in\nat\cup\setexp{\infty} \punc{,}
    \\
    & & 
    \expholedepth{\acxt} 
      & {} \,\defdby\,
    \holedepthnotexp{ \exprednfi{0}{\acxt} } \in\nat\cup\setexp{\infty} \punc{,}
  \end{align*}
  the \emph{expansion depth}~$\expdepth{\bter}$ of $\bter$, 
  the \emph{expansion depth}~$\expdepth{\acxt}$ of $\acxt$,
  and the \emph{expansion hole depth}~$\expholedepth{\acxt}$ of $\acxt$,
  namely 
  as the depth of the expanded form of $\bter$, 
     the depth of the expanded form of $\acxt$ while ignoring expansion symbols,
  and the hole depth of the expanded form of $\acxt$ while ignoring expansion symbols, respectively. 
  Here and below we denote by $\depthnotexp{\cdot}$ the operation that measures the depth
  of terms and contexts while ignoring symbols $\sexpandi{i}$ for $i\in\nat$.
  So the expansion depth $\expdepth{\acxt}$, and the expansion hole depth $\expholedepth{\acxt}$ of a context $\acxt$
  ignore the symbols $\sexpandi{i_l}$ in guarded hole expressions $\expandi{i_l}{\holei{j_l}}$ 
  in the representation of the expanded forms of $\acxt$ according to Lemma~\ref{lem:expanded:form:cxts}.
\end{definition}

\begin{lemma}\label{lem:finite:expdepth:fin:nested}
  Let $\lopsimTRSwrt{\alTRS} = \pair{\asiglopsim}{\ruleslopsim}$ be the \lopsimTRS{s} $\lopsimTRSwrt{\alTRS}$ for finitely nested \lTRS~$\alTRS$.
  
  Then for every term $\bter\in\termsover{\asiglopsim}$ the expansion depth $\expdepth{\bter}$ of $\bter$ is finite, that is, a natural number.  
  Also, for every context $\acxt\in\contextsover{\asiglopsim}$,
  the expansion depth $\expdepth{\acxt}$, and the expansion hole depth $\expholedepth{\acxt}$ of $\acxt$ are finite. 
\end{lemma}

\begin{proof}
  We argued in Definition~\ref{def:expanded:form:cxts:terms}, the expanded form of terms in $\termsover{\asig\cup\asiglambda}$
  and of contexts in $\contextsover{\asig\cup\asiglambda}$ are well-defined finite terms, and contexts, respectively.
  Now as the expansion depth of a term or context is defined as the depth of the expanded form of the term or context,
  it follows that that the expansion depth of the term or context in question is finite.
\end{proof}

Since a \lambdaterm\ representation $\ater$ and the \lambdaterm~$\denlter{\ater}$ denoted by it
have the same depth, the expansion depth of a term $\ater$ that denotes a \lambdaterm~$\alter$
coincides with the depth of $\alter$.

\begin{proposition}\label{prop:expdepth:lterrep:lter}
  Let $\alTRS = \pair{\asig}{\arules}$ be a \lTRS, 
  and let $\lopsimTRSwrt{\alTRS}$ be the \lopsimTRS\ for $\alTRS$. 
  If for a term $\ater$ in $\lopsimTRSwrt{\alTRS}$
  it holds that $\denlterwrt{\alTRS}{\ater} = \alter$ for a \lambdaterm~$\alter$,
  then $\expdepth{\ater} = \depth{\denlterrepwrt{\alTRS}{\ater}} = \depth{\denlterwrt{\alTRS}{\ater}} = \depth{\alter}$.  
\end{proposition}

The following lemma formulates clauses for the expansion depth
depending on the outermost symbol of a term in a \lopsimTRS. 
For finitely nested \lTRSs, these clauses can be read as an inductive definition. 
They can be proved in a straightforward manner by making use
of the definition  via the expansion \TRS\ of the \lambdaterm\ representations $\denlterrepwrt{\alTRS}{\ater}$
for terms $\ater$ of the \lopsimTRS\ for a \lTRS~$\alTRS$. 

\begin{lemma}[inductive clauses for the expansion depth of terms and contexts]\label{lem:expdepth}
  Let $\alTRS = \pair{\asig}{\arules}$ be a finitely nested \lTRS, 
  and let $\lopsimTRSwrt{\alTRS} = \pair{\asiglopsim}{\ruleslopsim}$ be the \lopsimTRS\ for $\alTRS$. 
  
  The expansition depth 
  $\expdepth{\acxt}$ of contexts $\acxt\in\contextsover{\asiglopsim,\vars}$ satisfies the following clauses:
  \begin{alignat*}{2}
    \expdepth{\avar} & = 0 
      & & (\avar\text{\nf\ variable in $\vars$})
    \\  
    \expdepth{\holei{i}} & = 0
      & & (i\in\nat)
    \\
    %
    \expdepth{\folapp{\acxti{1}}{\acxti{2}}} 
      & = 
    1 + \max \setexp{ \expdepth{\acxti{1}}, \expdepth{\acxti{2}} }
    \displaybreak[0]\\
    \expdepth{\afoscope{\acxti{1},\ldots,\acxti{k}}}
      & = 
    1 + \expdepth{ \afoscopecxtap{\acxti{1},\ldots,\acxti{k},\afovari{0}} }
    \displaybreak[0]\\
    \expdepth{\afovari{j}} 
      & = 
    0
      & & (j\in\nat)
    \displaybreak[0]\\
    \expdepth{\folabs{\afovari{j}}{\bter}}
      & = 
    1 + \expdepth{\bter}
    \displaybreak[0]\\
    \expdepth{\lopstart{\acxt}}
      & = 
    \expdepth{\acxt}  
    \displaybreak[0]\\
    %
      %
    \expdepth{\lopni{n}{i}{\acxti{0},\acxti{1},\ldots,\acxti{n}}}
      & =   
    \expdepth{\folapp{\cdots\folapp{\acxti{0}}{\acxti{1}}\ldots}{\acxti{n}}}   
  \end{alignat*}
  where $i,j,k,l,n\in\nat$, 
  $\acxti{0},\acxti{1},\ldots\in\contextsnover{l}{\asiglopsim,\vars}$ are contexts,
  and $\afoscopesym\in\asigmin$ with $\arityof{\afoscopesym} = k$ are scope symbols in $\alTRS$ with appertaining scope contexts $\afoscopecxt$ in $\alTRS$.
  These clauses specialize to analogous clauses for terms in $\termsover{\asiglopsim,\vars}$,
  because terms can be viewed as contexts without hole occurrences.
\end{lemma}

\begin{proof}
  The base cases of the inductive clauses can be verified as follows.
  For every $\avar\in\vars$, we have
  $\expdepth{\avar}  = \depthnotexp{ \expandi{0}{\avar} } = \depth{\avar} = 0$,
  and for every hole $\holei{i}\in\holes$ we find
  $\expdepth{\holei{i}} = \depthnotexp{ \expandi{0}{\holei{i}} } = \depth{ \holei{i} } = 0$, for $i\in\nat$.
  With the step $\expandi{0}{\afovari{i}} \expred \afovari{i}$
  we get
  $\expdepth{\afovari{i}} = \depthnotexp{ \exprednfi{0}{\afovari{i}} } = \depthnotexp{\afovari{i}} = 0$.
  
  Each of the other cases can be established by arguing with $\sexpred$ steps. We provide two examples.
  For the first one we consider $\acxt \syntequal \folapp{\acxti{1}}{\acxti{2}}$.
  Then
  $\expandi{0}{\folapp{\acxti{1}}{\acxti{2}} } 
     \expred
   \folapp{ \expandi{0}{\acxti{1}} }{ \expandi{0}{\acxti{2}} }$
  is an expansion step on $\expandi{0}{\acxt}$, 
  with which we can argue as follow:
  \begin{align*}
    \expdepth{ \folapp{\acxti{1}}{\acxti{2}} }
      & {} =
    \depthnotexp{ \exprednfi{0}{ \folapp{\acxti{1}}{\acxti{2}} } }
      =
    \depthnotexp{ \exprednf{(\expandi{0}{ \folapp{\acxti{1}}{\acxti{2}} })} }
    \\
      & {} =
    \depthnotexp{ \exprednf{(\folapp{ \expandi{0}{\acxti{1}} }{ \expandi{0}{\acxti{2}} })} }
    \displaybreak[0]\\
      & {} =
    \depthnotexp{ \folapp{ \exprednf{\expandi{0}{\acxti{1}}} }{ \exprednf{\expandi{0}{\acxti{2}}} }}
           =
    \depthnotexp{ \folapp{ \exprednfi{0}{\acxti{1}} }{ \exprednfi{0}{\acxti{2}} } }
    \\
      & {} =
    1 + \max \bigl\{ \depthnotexp{\exprednfi{0}{\acxti{1}}}, \depthnotexp{\exprednfi{0}{\acxti{1}}} \bigr\}
      =
    1 + \max \bigl\{ \expdepth{\acxti{1}}, \expdepth{\acxti{2}} \bigr\} \punc{.} 
  \end{align*}
  As a second example, we consider a context $\acxt \syntequal \afoscope{\acxti{1},\ldots,\acxti{k}}$.
  In this case there is an expansion step of the form
  $\expandi{0}{ \afoscope{\acxti{1},\ldots,\acxti{k}} }
     \expred
   \folabs{\afovari{0}}{ \exprednf{( \expandi{1}{ (\afoscopecxtap{\acxti{1},\ldots,\acxti{k},\afovari{0}}) } )} }$,
  with which we now argue as follows:
  \begin{align*}
    \expdepth{ \afoscope{\acxti{1},\ldots,\acxti{k}} }
      & {} =
    \depthnotexpbig{ \exprednfbigi{0}{ \afoscope{\acxti{1},\ldots,\acxti{k}} } }
      =
    \depthnotexpbig{ \exprednf{(\expandi{0}{ \afoscope{\acxti{1},\ldots,\acxti{k}} })} }
    \\ 
      & {} =
    \depthnotexpbig{ \exprednf{(\folabs{\afovari{0}}{ \expandi{1}{ (\afoscopecxtap{\acxti{1},\ldots,\acxti{k},\afovari{0}}) } })} } 
      & & \hspace*{-20ex} \text{(using the step here)}
    \displaybreak[0]\\
      & {} =
    \depthnotexpbig{ \folabs{\afovari{0}}{ \exprednf{( \expandi{1}{ (\afoscopecxtap{\acxti{1},\ldots,\acxti{k},\afovari{0}}) } )} } } 
    \displaybreak[0]\\
      & {} = 
    \depthnotexpbig{ \folabs{\afovari{0}}{ \exprednfbigi{1}{ (\afoscopecxtap{\acxti{1},\ldots,\acxti{k},\afovari{0}}) } } }  
    \displaybreak[0]\\
      & {} =
    \depthnotexpbig{ \folabs{\afovari{0}}{ \exprednfbigi{0}{ (\afoscopecxtap{\acxti{1},\ldots,\acxti{k},\afovari{0}}) } } }   
      & & \hspace*{-20ex} \text{(by using Lemma~\ref{lem:expanded:form:terms:cxts:difference})}
    \\
      & {} =
    1 + \depthnotexpbig{ \exprednfbigi{0}{ (\afoscopecxtap{\acxti{1},\ldots,\acxti{k},\afovari{0}}) } } 
      = 
    1 + \expdepth{ \afoscopecxtap{\acxti{1},\ldots,\acxti{k},\afovari{0}} }  \punc{,}
  \end{align*}
  in order to obtain the inductive clause for  $\acxt \syntequal \afoscope{\acxti{1},\ldots,\acxti{k}}$.
\end{proof}

We extend the concept of expansion depth also for scope symbols in a natural way. 

\begin{definition}[expansion depth of scope symbols, and of \lTRSs]\label{def:expdepth:foscopsym:lTRS}
  Let $\alTRS = \pair{\asig}{\arules}$ be a \lTRS.
  
  Let $\afoscopesym\in\asigmin$ be a scope symbol in $\alTRS$ with arity $k$. 
  The \emph{expansion depth} $\expdepth{\afoscopesym}$ of $\afoscopesym\in\asigmin$ 
  is defined as $\expdepth{\afoscope{\holei{1},\ldots,\holei{k}}}\in\nat\cup\setexp{\infty}$, that is,
  as the expansion depth of the \nary{k} context $\afoscope{\holei{1},\ldots,\holei{k}}$.
   
  We also define by 
  $ \expdepth{\alTRS} \defdby \max \descsetexp{ \expdepth{\afoscopesym} }{ \afoscopesym\in\asigmin} \in \nat\cup\setexp{\infty}$,
  the \emph{maximal expansion depth} of a scope symbol in $\alTRS$.
\end{definition}

Note that if a \lTRS~$\alTRS$ is finitely nested, then 
$\expdepth{\afoscopesym} = \expdepth{\afoscope{\holei{1},\ldots,\holei{k}}}\in\nat$ 
due to Lemma~\ref{lem:finite:expdepth:fin:nested}.
Furthermore, if in addition to being finitely nested 
$\alTRS$ is also finite, then $\expdepth{\alTRS} \in\nat$.


\begin{lemma}\label{lem:depth:losimTRS}
  Let $\alTRS = \pair{\asig}{\arules}$ be a \lTRS, 
  and let $\lopsimTRSwrt{\alTRS}$ be the \lopsimTRS\ for $\alTRS$. 
  If for a term $\ater$ in $\lopsimTRSwrt{\alTRS}$
  it holds that $\denlterwrt{\alTRS}{\ater} = \alter$ for a \lambdaterm~$\alter$,
  then 
  $\expdepth{\alTRS} \le \depth{\alter}$.   
\end{lemma}

\begin{figure}
\begin{center}
  \begin{tikzpicture}[thick,
    sym/.style={circle,draw,inner sep=.5mm,outer sep=0mm},
    term/.style={draw,inner sep=0.25mm,outer sep=0.2mm,isosceles triangle,isosceles triangle apex angle=40,shape border rotate=90},
    cxt/.style={draw,
                     inner sep=0.25mm,outer sep=0.25mm,isosceles triangle,isosceles triangle apex angle=110,shape border rotate=90}]
    
    \node (ap) [sym] {$\sfolapp$};
    \node (f) [sym] at (-1.25cm,-1cm) {$\afoscopesym$};
    \node (s1) [term,minimum height=1cm,anchor=north] at (-2cm,-2cm) {$\ateri{1}$};
    \node (dots) [anchor=north] at (-1.25,-1.8cm) {$\cdots$};
    \node (sk) [term,minimum height=1cm,anchor=north] at (-.5cm,-2cm) {$\ateri{k}$};
    \node (t) [term,minimum height=2.3cm,anchor=north] at (1.25cm,-0.7cm) {$\bter$};
    
    \draw ($(ap.north) + (0cm,0.25cm)$) -- (ap);
    \draw (ap) -- (f.north);
    \draw (f) -- (s1.north);
    \draw (f) -- (sk.north);
    \draw (ap) -- (t.north);
    
    \draw [|<->|,thin] (-3.25cm,.3cm) -- node [left] {$d$} (-3.25cm,-3.075cm);
    \draw [|<->|,thin] (-3.75cm,.3cm) -- node [left] {$1$} (-3.75cm,-0.7cm);
    \draw [<->|,thin] (-3.75cm,-0.7cm) -- node [left] {$d{-}1$} (-3.75cm,-3.075cm);
    \draw [|<->|,thin] (-2.35cm,-0.7cm) -- node [left] {$\expdepth{\afoscopesym}\!$} (-2.35cm,-2cm);
    
    
    \node (F) [cxt,minimum height=1cm,anchor=north] at (5.5cm,-0.1cm){$\afoscopecxt$};
    
    \node (s1') [term,minimum height=1cm,anchor=north] at (4.35cm,-1.5cm) {$\ateri{1}$};
    \node (dots') [anchor=north] at (5.05cm,-1.5cm) {$\cdots$};
    \node (sk') [term,minimum height=1cm,anchor=north] at (5.75cm,-1.5cm) {$\ateri{k}$};
    \node (t') [term,minimum height=2.3cm,anchor=north] at (6.65cm,-1.5cm) {$\bter$};

    \node (red) at (2.5cm,-0.75cm) {\scalebox{1.75}{$\sred_{\!\ruleof{\!\afoscopesym}}$}};

    \draw ($(F.north) + (0cm,0.3cm)$) -- (F);
    \draw ($(s1'.north) + (0cm,0.37cm)$) -- (s1');
    \draw ($(sk'.north) + (0cm,0.37cm)$) -- (sk');
    \draw ($(t'.north) + (0cm,0.37cm)$) -- (t');
    
    \draw [|<->|,thin] (3.5cm,.3cm) -- node [right,very near start] {$\le d{-}2$} (3.5cm,-2.6cm);
    \draw [|<->|,thin] (7.7cm,.3cm) -- node [right] {$\le \expdepth{\afoscopesym}{-}1$} (7.7cm,-1.5cm);
    \draw [<->|,thin] (7.7cm,-1.5cm) -- node [right] {$\le d{-}1$} (7.7cm,-3.875cm);
    \draw [|<->|,thin] (9.4cm,.3cm) -- node [right] {$\le d {+} (\expdepth{\afoscopesym}{-}2)$} (9.4cm,-3.875cm);
    

  \end{tikzpicture}
\end{center}
    \vspace*{-2.5ex}
  \caption{\label{fig:depth:increase}
    Illustration of the expansion depth increase that is caused by the simulation
    of a $\sbetared$ step at the root of a \lambdaterm~$\alter$
    on a \protect\lTRS\ term that denotes~$\alter$:
    the depth increase in a step
    $\folapp{\afoscope{\ateri{1},\ldots,\ateri{k}}}{\bter} \red \afoscopecxtap{\ateri{1},\ldots,\ateri{k},\bter}$
    is at most $\expdepth{\afoscopesym} - 2$, where
    $\ruleof{\afoscopesym} \funin
     \folapp{\afoscope{\avari{1},\ldots,\avari{k}}}{\bvar}
       \;\sred\;
     \afoscopecxtap{\avari{1},\ldots,\avari{k},\bvar}$
    is the defining rule for scope symbol~$\afoscopesym$.
    The subterm $\bter$ could be duplicated in the step and occur several times below $\afoscopecxt$,
    but only one such occurrence is displayed.
    }
\end{figure}

Next we establish expansion depth variants of the two easy context lemmas
as formulated in Section~\ref{prelims}:
of Lemma~\ref{lem:holedepth:vs:depth}, and Lemma~\ref{lem:depth:cxtap:vs:depth:holedepth}.
Those lemmas are also crucial for the proof of the lemmas below.

\begin{lemma}
  $\expholedepth{ \acxtap{\ateri{1},\ldots,\ateri{n},\Box} } \:\le\: \expdepth{\acxt}\:$
  holds, in a finitely nested \lTRS~$\pair{\asig}{\arules}$, 
  for all terms $\ateri{1},\ldots,\ateri{n}\in\termsover{\asig\cup\asiglambda}$,
  where $n\in\nat$,
  and all contexts $\acxt\in\contextsnover{n+1}{\asig\cup\asiglambda}$
  in which there is at least one occurrence of $\holei{n+1}$.   
    \label{lem:expholedepth:vs:expdepth} 
\end{lemma}

\begin{proof}
  Let $\ateri{1},\ldots,\ateri{n}\in\termsover{\asig\cup\asiglambda}$ be terms,
  and let $\acxt\in\contextsnover{n+1}{\asig\cup\asiglambda}$ be a context
  in which $\holei{n+1}$ has an occurrence.
  Then due to Lemma~\ref{lem:expanded:form:cxts} and Lemma~\ref{lem:expanded:form:cxtap} 
  the expanded forms of $\acxt$ and of $\acxtap{\ateri{1},\ldots,\ateri{n},\hole}$ 
  can be represented, with a linear context $\bcxt\in\contextsnover{m,1}{\asig\cup\asiglambda}$ 
  and $i_1,\ldots,i_m\in\nat$, and $j_1,\ldots,j_m\in\setexp{1,\ldots,n+1}$, 
  as follows:
  \begin{align}
    \exprednfi{0}{ \acxt }
      & \:\syntequal\:
    \bcxtap{ \exprednfi{i_1}{\holei{j_1}}, \ldots, \exprednfi{i_m}{\holei{j_m} } } \punc{,}
      \label{eq:1:prf:lem:expholedepth:vs:expdepth}
    \\
    \exprednfi{0}{ \acxtap{\ateri{1},\ldots,\ateri{n},\hole} }
      & \:\syntequal\:
    \bcxtap{ \exprednfi{i_1}{\ccxti{j_1}}, \ldots, \exprednfi{i_m}{\ccxti{j_m} } } \punc{,}
      \label{eq:2:prf:lem:expholedepth:vs:expdepth}
    \\
    \text{where }
      & 
      \ccxti{i} \defdby \begin{cases}
                          \ateri{i} & \text{ if $i\in\setexp{1,\ldots,n}$} 
                          \\
                          \hole    & \text{ if $i = n+1$} 
                        \end{cases}
                        \hspace*{1ex} \in\contextsnover{1}{\asig\cup\asiglambda}\punc{,}\text{ for $i\in\setexp{1,\ldots,n+1}$.}
      \notag
  \end{align}
  Since $\holei{n+1}$ occurs in $\acxt$, it follows that one of $\exprednfi{i_l}{\ccxti{j_l}}$ for $l\in\setexp{1,\ldots,m}$
  is of the form $\expandi{i_l}{\hole}$. We will use this in the application of Lemma~\ref{lem:holedepth:vs:depth}
  in the following argumentation that we now can perform on the basis of the preparation above: 
  \begin{align*}
    \expholedepthbig{ \acxtap{\ateri{1},\ldots,\ateri{n},\hole}}
      \: & = \:
    \holedepthnotexpbig{ \exprednfi{0}{ \acxtap{\ateri{1},\ldots,\ateri{n},\hole} } }
      & & \text{(by def.\ of $\expholedepth{\cdot}$)}
    \\
      & = \:
    \holedepthnotexpbig{ \bcxtap{ \exprednfi{i_1}{\ccxti{j_1}}, \ldots, \exprednfi{i_m}{\ccxti{j_m} } } }      
      & & \text{(by \eqref{eq:2:prf:lem:expholedepth:vs:expdepth})} 
    \displaybreak[0]\\[-0.35ex]
      & \le
    \depthnotexp{ \bcxt }
      & & \parbox{\widthof{(by using a $\depthnotexp{\cdot}$\nb-version}}
                 {(by using a $\depthnotexp{\cdot}$\nb-version\\\phantom{(by} of Lemma~\ref{lem:holedepth:vs:depth})}      
    \displaybreak[0]\\[-0.35ex]
      & =
    \depthnotexpbig{ \bcxtap{ \expandi{i_1}{\holei{j_1}}, \ldots, \expandi{i_m}{\holei{j_m} } } }
      & & \text{(by def. of $\depthnotexp{\cdot}$)}      
    \displaybreak[0]\\
      & =
    \depthnotexpbig{ \bcxtap{ \exprednfi{i_1}{\holei{j_1}}, \ldots, \exprednfi{i_m}{\holei{j_m} } } }
      & & \text{(by def.\ of $\sexprednfi{\cdot}$)}
    \displaybreak[0]\\
      & =
    \depthnotexpbig{ \exprednfi{0}{\acxt} }
      & & \text{(by \eqref{eq:1:prf:lem:expholedepth:vs:expdepth})} 
    \\
      & =
    \expdepth{ \acxt } 
      & & \text{(by def.\ of $\expholedepth{\cdot}$)} \punc{.}
  \end{align*}
  In this way we have established the inequality as stated by the lemma. 
\end{proof}

\begin{lemma} 
  $\expdepth{ \acxtap{\ater} }  \: = \: \max \setexp{ \expdepth{\acxt},\, \expholedepth{\acxt} + \expdepth{\ater} }$
  holds, in a finitely nested \lTRS~$\pair{\asig}{\arules}$, 
  for all contexts $\acxt\in\contextsnover{1}{\asig\cup\asiglambda}$, and all terms $\ater\in\termsover{\asig\cup\asiglambda}$.%
  \label{lem:expdepth:cxtap:vs:expdepth:expholedepth}
\end{lemma}

\begin{proof}
  Let $\acxt\in\contextsnover{1}{\asig\cup\asiglambda}$, and $\ater\in\termsover{\asig\cup\asiglambda}$.
  Due to Lemma~\ref{lem:expanded:form:cxts} and Lemma~\ref{lem:expanded:form:cxtap} 
  the expanded forms of $\acxt$ and of $\acxtap{\ater}$ 
  can be represented with a linear context $\bcxt\in\contextsnover{m,1}{\asig\cup\asiglambda}$ for $m\in\nat$,
  and $i_1,\ldots,i_m\in\nat$, and $j_1,\ldots,j_m\in\setexp{1,\ldots,n+1}$
  as follows:
  \begin{align}
    \exprednfi{0}{ \acxt }
      & \:\syntequal\:
    \bcxtap{ \exprednfi{i_1}{\holei{j_1}}, \ldots, \exprednfi{i_m}{\holei{j_m} } } \punc{,}
      \label{eq:1:lem:expdepth:cxtap:vs:expdepth:expholedepth}
    \\
    \exprednfi{0}{ \acxtap{\ater} }
      & \:\syntequal\:
    \bcxtap{ \exprednfi{i_1}{\ater}, \ldots, \exprednfi{i_m}{\ater} } \punc{.}
      \label{eq:2:lem:expdepth:cxtap:vs:expdepth:expholedepth}
  \end{align}
  By using \eqref{eq:1:lem:expdepth:cxtap:vs:expdepth:expholedepth} we obtain:
  \begin{gather}
    \expdepth{\acxt}
      =
    \depthnotexp{\exprednfi{0}{\acxt}}
      =
    \depthnotexp{\bcxtap{ \expandi{i_1}{\holei{j_1}}, \ldots, \expandi{i_m}{\holei{j_m} } }}    
      =
    \depth{\bcxt} 
      =
    \depthnotexp{\bcxt} \punc{,}
      \label{eq:3:lem:expdepth:cxtap:vs:expdepth:expholedepth}
    \\
    \expholedepth{\acxt}
      =
    \holedepthnotexp{\exprednfi{0}{\acxt}}
      =
    \holedepthnotexp{\bcxtap{ \expandi{i_1}{\holei{j_1}}, \ldots, \expandi{i_m}{\holei{j_m} } }}    
      =
    \holedepth{\bcxt} 
      =
    \holedepthnotexp{\bcxt} \punc{.}
      \label{eq:4:lem:expdepth:cxtap:vs:expdepth:expholedepth}  
  \end{gather}
  On the basis of these preparations we can now argue: 
  \begin{align*}
    \expdepth{ \acxtap{\ater} }
       \: & = \:
     \depthnotexp{ \exprednfi{0}{\acxtap{\ater}} }
       & & \text{(by def.\ of $\expdepth{\cdot}$)}
     \\
       & = \:
     \depthnotexp{ \bcxtap{ \exprednfi{i_1}{\ater}, \ldots, \exprednfi{i_m}{\ater}  } } 
       & & \text{(by \eqref{eq:2:lem:expdepth:cxtap:vs:expdepth:expholedepth})}   
     \displaybreak[0]\\
       & = \:
     \max \descsetexpbig{ \depthnotexp{\bcxt},\, \holedepthnotexp{\bcxt} + \depthnotexp{ \exprednfi{i_j}{\ater} } }   
                        { j\in\setexp{1,\ldots,n} }
       & & \parbox{\widthof{(by using a $\lvert{\cdot}\rvert^{\scriptnotexp}_{(\hole)}$\nb-version}}
                 {(by using a $\lvert{\cdot}\rvert^{\scriptnotexp}_{(\hole)}$\nb-version\\\phantom{(by} of Lemma~\ref{lem:depth:cxtap:vs:depth:holedepth})}        
     \displaybreak[0]\\
       & = \:
     \max \setexpbig{ \depthnotexp{\bcxt},\, \holedepthnotexp{\bcxt} + \expdepth{\ater} }
       & & \text{(by appeal to Lemma~\ref{lem:expanded:form:terms:cxts:difference})} 
     \displaybreak[0]\\
      & = \:
     \max \setexpbig{ \expdepth{\acxt},\, \expholedepth{\acxt} + \expdepth{\ater} }
       & & \text{(by \eqref{eq:3:lem:expdepth:cxtap:vs:expdepth:expholedepth}, and \eqref{eq:4:lem:expdepth:cxtap:vs:expdepth:expholedepth})} \punc{.}
  \end{align*}
  In this way we have shown the equation as stated by the lemma.
\end{proof}

For analyzing the depth increase of steps in \lopsimTRS{s},
the next two lemmas will be instrumental.
They relate the expansion depth of contexts filled with terms
to the expansion depths of occurring terms.

\begin{lemma}\label{lem:lifting:expdepth:le:in:cxt}
  Let $\alTRS = \pair{\asig}{\arules}$ be a finitely nested \lTRS.
  Then for all unary contexts $\acxt\in\contextsnover{1}{\asig}$, terms $\ater,\bter\in\termsover{\asig}$, and $d\in\nat$
  the following statements hold:
  \begin{align}
    \expdepth{\ater} \;\le\; \expdepth{\bter} + d
      \;\; & \Longrightarrow\;\;
    \expdepthbig{\acxtap{\ater}}
      \;\le\;
    \expdepthbig{\acxtap{\bter}} + d \punc{,}
      \label{eq:1:lem:lifting:expdepth:le:in:cxt}
    \\
    \expdepth{\ater} \;=\; \expdepth{\bter}
      \;\; & \Longrightarrow\;\;
    \expdepthbig{\acxtap{\ater}}
      \;=\;
    \expdepthbig{\acxtap{\bter}} \punc{.}
      \label{eq:2:lem:lifting:expdepth:le:in:cxt}
  \end{align}
\end{lemma}

\begin{proof}
  Let $\acxt\in\contextsnover{1}{\asig}$, $\ater,\bter\in\termsover{\asig}$, and $d\in\nat$.
  To verify \eqref{eq:1:lem:lifting:expdepth:le:in:cxt}
  we assume that $\expdepth{\ater} \;\le\; \expdepth{\bter} + d$ holds, and show the inequality on the right-hand side in \eqref{eq:1:lem:lifting:expdepth:le:in:cxt}. 
  For this we argue as follows:
  \begin{alignat*}{2}
    \expdepth{\acxtap{\ater}}
      & \:=\: 
    \max \setexp{ \expdepth{\acxt},\, \expholedepth{\acxt} + \expdepth{\ater} } 
      & \qquad & \text{(by Lemma~\ref{lem:expdepth:cxtap:vs:expdepth:expholedepth})} 
    \\
      & \:\le\: 
    \max \setexp{ \expdepth{\acxt},\, \expholedepth{\acxt} + \expdepth{\bter} + d }
      & \qquad & \text{(using the assumption)}   
    \displaybreak[0]\\
      & \:\le\: 
    \max \setexp{ \expdepth{\acxt} + d,\, \expholedepth{\acxt} + \expdepth{\bter} + d }
      & \qquad & \text{(possibly increasing the maximum)}  
    \displaybreak[0]\\
      & \: = \: 
    \max \setexp{ \expdepth{\acxt},\, \expholedepth{\acxt} + \expdepth{\bter} } + d
      & \qquad & \text{(simplifying the maximum expression)}  
    \\ 
      & \:=\: 
    \expdepth{\acxtap{\bter}} + d 
      & & \text{(by Lemma~\ref{lem:expdepth:cxtap:vs:expdepth:expholedepth})} \punc{.}  
  \end{alignat*} 
  Statement~\eqref{eq:2:lem:lifting:expdepth:le:in:cxt} follows by using \eqref{eq:1:lem:lifting:expdepth:le:in:cxt}
  with $d = 0$ in both directions. 
\end{proof}

\begin{lemma}\label{lem:cxt:contract-step}
  Let $\alTRS = \pair{\asig}{\arules}$ be a finitely nested \lTRS.
  Then for all contexts $\acxt\in\contextsnover{k+1}{\asig}$, and terms $\ateri{1},\ldots,\ateri{k},\cter\in\termsover{\asig}$, where $k\in\nat$, 
  the following statement holds:
  \begin{equation}\label{eq:lem:cxt:contract-step}
    \expdepth{ \acxtap{\ateri{1},\ldots,\ateri{k},\cter} }  
      \:\le\: 
    \max \setexp{ \expdepth{ \acxtap{\ateri{1},\ldots,\ateri{k},\hole} },\, 
                  \expdepth{\acxt} + \expdepth{\cter} }  \punc{.}   
  \end{equation}             
\end{lemma}

\begin{proof}
  For contexts $\acxt\in\contextsnover{k+1}{\asig}$, and terms $\ateri{1},\ldots,\ateri{k},\cter\in\termsover{\asig}$ with $k\in\nat$ we argue as follows.
  %
  If $\holei{n+1}$ occurs in $\acxt$, then we can argue as follows:
  \begin{alignat*}{2}
    &
    \expdepth{ \acxtap{\ateri{1},\ldots,\ateri{k},\cter} } 
    \\[-0.5ex]
      & \qquad \:=\: 
    \max \setexp{ \expdepth{ \acxtap{\ateri{1},\ldots,\ateri{k},\hole} },\, 
                  \expholedepth{ \acxtap{\ateri{1},\ldots,\ateri{k},\hole } } + \expdepth{\cter} }
      & \quad & \parbox{\widthof{(context $\acxtap{\ateri{1},\ldots,\ateri{k},\hole}$)}}
                       {(by Lemma~\ref{lem:expdepth:cxtap:vs:expdepth:expholedepth}, 
                         using
                        \\\phantom{(}%
                         context $\acxtap{\ateri{1},\ldots,\ateri{k},\hole}$)}            
    \displaybreak[0]\\ 
      & \qquad \:\le\: 
    \max \setexp{ \expdepth{ \acxtap{\ateri{1},\ldots,\ateri{k},\hole} },\, 
                  \expdepth{ \acxt } + \expdepth{\cter} }  
      & & \text{(by using Lemma~\ref{lem:expholedepth:vs:expdepth})} \punc{.}         
  \end{alignat*}  
  and have established the statement \eqref{eq:lem:cxt:contract-step}.
  If, on the other hand, $\holei{n+1}$ does not occurs in $\acxt$, then we argue:
  \begin{align*}
    \expdepth{ \acxtap{\ateri{1},\ldots,\ateri{k},\cter} } 
      =
    \expdepth{ \acxtap{\ateri{1},\ldots,\ateri{k},\hole} } 
      \:\le\: 
    \max \setexp{ \expdepth{ \acxtap{\ateri{1},\ldots,\ateri{k},\hole} },\, 
                  \expdepth{\acxt} + \expdepth{\cter} } \punc{,}
  \end{align*}    
  and have obtained \eqref{eq:lem:cxt:contract-step} again.
\end{proof}\pagebreak[4]

Now we can formulate, and prove, a crucial lemma (Lemma~\ref{lem:depth-increase:contract-step}). 
Its central statement is that the depth increase in a $\scontractred$ step   
(with respect to a \lopsimTRS) at the root of a term
is bounded by the depth of the scope context of the scope symbol that is involved in the step.
See Figure~\ref{fig:depth:increase} for an illustration of the underlying intuition 
for the analogous case of a step according to the defining rule of a scope symbol.  
Then
we obtain a lemma (Lemma~\ref{lem:expdepth:lopsimred:steps}) concerning the depth increase in general $\scontractred$ and $\ssearchred$ steps.

\begin{lemma}\label{lem:depth-increase:contract-step}
  Let $\alTRS = \pair{\asig}{\arules}$ be a finitely nested \lTRS.
  Then for every scope symbol $\afoscopesym\in\asigmin$ with arity $k$ and scope context $\afoscopecxt$,
  and for all terms $\ateri{1},\ldots,\ateri{k},\cter\in\termsover{\asig}$, and all $i\in\nat$, it holds:
  \begin{enumerate}[(i)]\setlength{\itemsep}{0ex}
    \item{}\label{it:1:lem:depth-increase:contract-step}
      $
      \expdepthbig{\afoscopecxtap{\ateri{1},\ldots,\ateri{k},\cter}}
        \;\le\;
      \expdepthbig{\folapp{\afoscope{\ateri{1},\ldots,\ateri{k}}}{\cter}}
        \;+\;
      \expdepth{\afoscopesym} - 2 \punc{.}
      $
    \item{}\label{it:2:lem:depth-increase:contract-step}
      $
      \begin{aligned}[t]
        &   
        \expdepthbig{ \lopni{n}{i}{\afoscopecxtap{\ateri{1},\ldots,\ateri{k},\cteri{1}},\cteri{2},\ldots,\cteri{n+1}} } 
          \; \le  
        \\
        & \hspace*{3ex}
        \le\;  
        \expdepthbig{ \lopni{n+1}{i}{\afoscope{\ateri{1},\ldots,\ateri{k}},\cteri{1},\ldots,\cteri{n+1}} }
          +
        \expdepth{\afoscopesym} - 2 \punc{.}
      \end{aligned}   
      $
  \end{enumerate}
\end{lemma}

\begin{proof}
  We let $\afoscopesym$, $\afoscopecxt$, $\ateri{1},\ldots,\ateri{k},\bter$, and $i$ be as assumed 
  in the lemma.
  We establish statement~\eqref{it:1:lem:depth-increase:contract-step} as follows:
  \begin{alignat*}{2}
    &
    \expdepthbig{\afoscopecxtap{\ateri{1},\ldots,\ateri{k},\cter}}
    \\
    & \quad \;\,\parbox[t]{\widthof{${}\le{}$}}{${}\le{}$}\;  
    \max \setexpbig{ \, \expdepth{\afoscopecxtap{\ateri{1},\ldots,\ateri{k},\hole}}, \:
                        \expdepth{\afoscopecxt} + \expdepth{\cter} \, }
      & & \qquad \text{(by Lemma~\ref{lem:depth-increase:contract-step})}                     
    \displaybreak[0]\\
    & \quad \;\,\parbox[t]{\widthof{${}\le{}$}}{${}={}$}\;  
    \max \setexpbig{ \, \expdepth{\afoscope{\ateri{1},\ldots,\ateri{k}}} - 1, \:
                        \expdepth{\afoscopesym} - 1 + \expdepth{\cter} \, } 
      & & \qquad \text{(by Def.~\ref{def:expdepth:foscopsym:lTRS}, and Lemma~\ref{lem:expdepth})}                      
    \displaybreak[0]\\
    & \quad \;\,\parbox[t]{\widthof{${}\le{}$}}{${}\le{}$}\;  
    \max \setexpbig{ \, \expdepth{\afoscope{\ateri{1},\ldots,\ateri{k}}} + \expdepth{\afoscopesym} - 1, \:
                        \expdepth{\afoscopesym} - 1 + \expdepth{\cter} \, }
      & & \qquad \parbox{\widthof{(by possibly increasing}}
                        {(by possibly increasing
                         \\\phantom{(}%
                         the maximum)}                    
    \displaybreak[0]\\
    & \quad \;\,\parbox[t]{\widthof{${}\le{}$}}{${}={}$}\;  
    ( \max \setexpbig{ \, \expdepth{\afoscope{\ateri{1},\ldots,\ateri{k}}}, \:
                          \expdepth{\cter} \, } )
     + \expdepth{\afoscopesym} - 1 
       & & \qquad \text{(simplification)}
    \displaybreak[0]\\
    & \quad \;\,\parbox[t]{\widthof{${}\le{}$}}{${}={}$}\;  
    ( 1 + \max \setexpbig{ \, \expdepth{\afoscope{\ateri{1},\ldots,\ateri{k}}}, \:
                           \expdepth{\cter} \, } )
     + \expdepth{\afoscopesym} - 2 
      & & \qquad \text{(rearrangement)}
    \\
    & \quad \;\,\parbox[t]{\widthof{${}\le{}$}}{${}={}$}\;   
      \expdepthbig{\folapp{\afoscope{\ateri{1},\ldots,\ateri{k}}}{\cter}}
        \;+\;
      \expdepth{\afoscopesym} - 2 
      & & \qquad \text{(by Lemma~\ref{lem:expdepth})}  \punc{.}
  \end{alignat*}
  
  For showing statement~\eqref{it:2:lem:depth-increase:contract-step} 
  we proceed by lifting the inequality in statement~\eqref{it:1:lem:depth-increase:contract-step} into a context
  by means of Lemma~\ref{lem:lifting:expdepth:le:in:cxt}.
  More precisely, we argue as follows by means of the inductive clauses in Lemma~\ref{lem:expdepth},
  and by appealing to Lemma~\ref{lem:lifting:expdepth:le:in:cxt} 
  for the context $\acxt \defdby \folapp{\cdots\folapp{{\hole}}{\cteri{1}}\cdots}{\cteri{n+1}}$:
  \begin{align*}
    &
    \expdepthbig{ \lopni{n}{i}{\afoscopecxtap{\ateri{1},\ldots,\ateri{k},\cteri{1}},\cteri{2},\ldots,\cteri{n+1}} } 
    \\
    & \qquad                
      \; = \;
    \expdepthbig{ \folapp{\cdots\folapp{\afoscopecxtap{\ateri{1},\ldots,\ateri{k},\cteri{1}}}{\cteri{2}}\ldots}{\cteri{n+1}} } 
    \displaybreak[0]\\
    & \qquad                
      \; \le \;
    \expdepthbig{ \folapp{\cdots\folapp{\folapp{\afoscope{\ateri{1},\ldots,\ateri{k}}}{\cteri{1}}}{\cteri{2}}\ldots}{\cteri{n+1}} }
      + \expdepth{\afoscopesym} - 2 
    \\
    & \qquad                
      \; = \;
    \expdepthbig{ \lopni{n+1}{i}{\afoscope{\ateri{1},\ldots,\ateri{k}},\cteri{1},\ldots,\cteri{n+1}} }
      + \expdepth{\afoscopesym} - 2 \punc{.}
  \end{align*}
  In this way we have now also justified the inequality in statement~\eqref{it:2:lem:depth-increase:contract-step}.
\end{proof}

\begin{lemma}\label{lem:expdepth:lopsimred:steps}
  Let $\lopsimTRSwrt{\alTRS} = \pair{\asiglopsim}{\ruleslopsim}$ be the \lopsimTRS\ 
  for a finitely nested \lTRS~$\alTRS = \pair{\asig}{\arules}$.
  
  Then every $\ssearchred$ step in $\lopsimTRSwrt{\alTRS}$ preserves the expansion depth,
  and every $\scontractred$ step increases the expansion depth by less that the expansion depth
  of the scope symbol $\afoscopesym$ involved in the contraction. 
  More precisely, the following statements hold for all $\bteri{1},\bteri{2}\in\termsover{\asiglopsim}$:
  \begin{enumerate}[(i)]\setlength{\itemsep}{0.5ex} 
    \item{}\label{it:1:lem:expdepth:lopsimred:steps}
      If $\,\bteri{1} \searchred \bteri{2}\,$, then $\,\expdepth{\bteri{1}} \;=\; \expdepth{\bteri{2}}\,$.
    \item{}\label{it:2:lem:expdepth:lopsimred:steps}
      If $\,\bteri{1} \contractred \bteri{2}\,$, then $\,\expdepth{\bteri{2}} \;\le\; \expdepth{\bteri{1}} + \expdepth{\afoscopesym}-2 \,$,
      where $\afoscopesym$ is the scope symbol involved in the step.   
  \end{enumerate}
\end{lemma}

\begin{proof}
  We first reduce the proof obligation for both items of the lemma
  to statements that pertain to rewrite steps that take place at the root of the term $\bteri{1}$. 
  This is because for non-root $\ssearchred$ and $\scontractred$ steps  
  the corresponding property can be lifted into a rewriting context by using Lemma~\ref{lem:lifting:expdepth:le:in:cxt}.
  For instance, consider a step $\bteri{1} \contractred \bteri{2}$ that does not take place at the root of $\bteri{1}$. 
  As such it is of the form $\bteri{1} \syntequal \acxtap{\bteri{10}} \contractred \acxtap{\bteri{20}} \syntequal \bteri{2}$
  for some \nontrivial\ unary context $\acxt \notsyntequal \hole$ and subterms $\bteri{10}$ and $\bteri{20}$ of $\bteri{1}$ and $\bteri{2}$, respectively,
  such that $\bteri{10} \contractred \bteri{20}$ is a root step. 
  Now under the assumption that \eqref{it:2:lem:expdepth:lopsimred:steps} holds for root $\scontractred$ steps,
  we have $\expdepth{\bteri{20}} \le \expdepth{\bteri{10}} + \expdepth{\afoscopesym} - 2$. 
  Then by using equation \eqref{eq:1:lem:lifting:expdepth:le:in:cxt} in Lemma~\ref{lem:lifting:expdepth:le:in:cxt}
  we obtain the desired inequality as follows:
  \begin{equation*}
    \expdepth{\bteri{2}} = \expdepth{\acxtap{\bteri{20}}} 
                       \le \expdepth{\acxtap{\bteri{10}}} + \expdepth{\afoscopesym} - 2  
                       = \expdepth{\bteri{1}} + \expdepth{\afoscopesym} - 2 \punc{,}
  \end{equation*}
  For non-root $\ssearchred$ steps, preservation of expansion depth can be argued analogously
  by using equation \eqref{eq:2:lem:lifting:expdepth:le:in:cxt} in Lemma~\ref{lem:lifting:expdepth:le:in:cxt},
  under the assumption that root  $\ssearchred$ steps preserve expansion depth.
  
  It remains to show that the statements in \eqref{it:1:lem:depth-increase:contract-step} and \eqref{it:2:lem:depth-increase:contract-step} 
  hold for root steps. 
  We start with showing this for item~\eqref{it:1:lem:depth-increase:contract-step},
  by inspecting the rules of the \lopsimTRS,
  and by using the clauses of expansion depth in Lemma~\ref{lem:expdepth}.
  The case of a root $\sinitred$ step is straightforward. 
  Now we consider the case of a root $\sdescendinfolappred$ step, which is of the form: 
  \begin{equation*}
    \bteri{1} 
      \syntequal 
    \lopni{n}{i}{\folapp{\ateri{1}}{\ateri{2}},\cteri{1},\ldots,\cteri{n}}
      \; \descendinfolappred \;
    \lopni{n+1}{i}{\ateri{1},\ateri{2},\cteri{1},\ldots,\cteri{n}}  
      \syntequal
    \bteri{2} \punc{,}
  \end{equation*}
  for some $n,i\in\nat$.
  Here we easily conclude with the clauses for the expansion depth in Lemma~\ref{lem:expdepth}: 
  \begin{align*}
    \expdepth{\bteri{1}}
      \: = \:
    \expdepthbig{ \lopni{n}{i}{\folapp{\ateri{1}}{\ateri{2}},\cteri{1},\ldots,\cteri{n}} } 
      & \: = \:
    \expdepthbig{ \folapp{\cdots\folapp{\folapp{\ateri{1}}{\ateri{2}}}{\cteri{1}}\ldots}{\cteri{n}} }  
    \\
      & \: = \:
    \expdepthbig{ \lopni{n+1}{i}{\ateri{1},\ateri{2},\cteri{1},\ldots,\cteri{n+1}} }
      \: = \:
    \expdepth{\bteri{2}} \punc{.}
  \end{align*}
  Next we consider a root $\sdescendinfolabsred$ step. With some $i\in\nat$ it is of the form:
  \begin{equation*}
    \bteri{1}
      \syntequal
    \lopni{0}{i}{\afoscope{\ateri{1},\ldots,\ateri{k}}}
      \; \descendinfolabsred \;
    \folabs{\afovari{i}}{\lopni{0}{i+1}{\afoscopecxtap{\ateri{1},\ldots,\ateri{k},\afovari{i}}}} 
      \syntequal
    \bteri{2} \punc{.}
  \end{equation*}
  Here we argue as follows by using clauses for the expansion depth in Lemma~\ref{lem:expdepth}:
  \begin{align*}
    \expdepth{\bteri{2}}
      =
    \expdepthbig{ \folabs{\afovari{i}}{\lopni{0}{i+1}{\afoscopecxtap{\ateri{1},\ldots,\ateri{k},\afovari{i}}}} }
      & \: = \:
    1 + \expdepthbig{ \lopni{0}{i+1}{\afoscopecxtap{\ateri{1},\ldots,\ateri{k},\afovari{i}}} }  
    \displaybreak[0]\\
      & \: = \:
    1 + \expdepthbig{ \afoscopecxtap{\ateri{1},\ldots,\ateri{k},\afovari{i}} }
    \displaybreak[0]\\
      & \: = \:
    \expdepthbig{ \folabs{\afovari{i}}{\afoscopecxtap{\ateri{1},\ldots,\ateri{k},\afovari{i}}} } 
    \displaybreak[0]\\
      & \: = \:
    \expdepth{ \afoscope{\ateri{1},\ldots,\ateri{k}} }
    \\
      & \: = \:
    \expdepthbig{ \lopni{0}{i}{\afoscope{\ateri{1},\ldots,\ateri{k}}} }
      =
    \expdepth{\bteri{1}} \punc{.}
  \end{align*}
  The case of a root $\svarred$ step is again easy,
  both according to the rule $\bap{(\scriptvar)}{0,i}$,
  also according to the rule $\bap{(\scriptvar)}{n+1,i}$,
  by using the clauses for $\lopni{0}{i}{\avar}$, and for $\lopni{n+1}{i}{\avar,\bvari{1},\ldots,\bvari{n+1}}$, respectively. 
  
  For showing the restriction of item~\eqref{it:2:lem:depth-increase:contract-step}
  to root steps, we consider a $\scontractred$ steps at the root.
  Such a step is of the form:
  \begin{equation*}
    \bteri{1}
      \syntequal
    \lopni{n+1}{i}{\afoscope{\ateri{1},\ldots,\ateri{k}},\cteri{1},\ldots,\cteri{n+1}}
      \;\;\scontractred\;\;
    \lopni{n}{i}{\afoscopecxtap{\ateri{1},\ldots,\ateri{k},\cteri{1}},\cteri{2},\ldots,\cteri{n+1}} 
      \syntequal
    \bteri{2} \punc{.}  
  \end{equation*}
  Then the desired expansion depth inequality 
  $\, \expdepth{ \bteri{2} }
        \le
      \expdepth{ \bteri{1} }
        +
      \expdepth{ \afoscopesym }
        -
      2 $
  follows from Lemma~\ref{lem:depth-increase:contract-step},~\eqref{it:2:lem:depth-increase:contract-step}.
\end{proof}

By a direct application of this lemma we obtain our main result concerning
the depth increase of terms in $\slopsimred$ rewrite sequences.

\begin{theorem}\label{thm:main:lTRS}
  Let $\alTRS = \pair{\asig}{\arules}$ be a finite, and finitely nested \lTRS,
  and let $D \defdby \expdepth{\alTRS}$. 
  Let $\arewseq$ be a finite or infinite $\slopsimred$ rewrite sequence~$\arewseq$ with initial term $\ater$.
  Then $\arewseq$ can be construed as a sequence of $\ssearchmred$ and $\scontractred$ steps: 
  \begin{center}
    $
  \begin{aligned}
    \arewseq \;\funin\; 
       \ater 
         = \cteri{0} 
             \searchmred
           \cteracci{0}
             \contractred  
           \cteri{1} 
             \searchmred
           \cteracci{1} & {}
             \contractred
               \cdots  
           \\
               \cdots & {}    
             \contractred
           \cteri{n} 
             \searchmred
           \cteracci{n}
             \; ( {}
             \contractred  
           \cteri{n+1} 
             \searchmred
               \cdots
             \: )
           \punc{,}
  \end{aligned}
    $   
  \end{center}
  and then the following statements hold for all $n\in\nat$ with $n\le l\,$ where $l\in\nat\cup\setexp{\infty}$ is the length of $\arewseq\,$:
  \begin{enumerate}[(i)]\setlength{\itemsep}{0ex}
    \item{}\label{it:1:thm:main:lTRS}
      $\expdepth{\cteri{n}} = \expdepth{\cteracci{n}}$,
      and 
      $\expdepth{\cteri{n+1}} 
         \le 
       \expdepth{\cteracci{n}} + (D-2)$
       if $n + 1\le l$,
      that is more verbally, the expansion depth remains the same in the $\ssearchred$ steps, and
      it increases by at most $D-2$ in the $\scontractred$ steps. 
    \item{}\label{it:2:thm:main:lTRS}  
      $\expdepth{\cteri{n}}, \expdepth{\cteracci{n}} \;\le\; \expdepth{\ater} + (D-2)\cdot n\,$,
      that is,
      the increase of the expansion depth along $\arewseq$ is linear
      in the number of $\scontractred$ steps performed, with $(D-2)$ as multiplicative constant. 
  \end{enumerate}  
\end{theorem}

\begin{proof}
  Statement~\eqref{it:1:thm:main:lTRS} 
  follows directly from Lemma~\ref{lem:expdepth:lopsimred:steps},~\eqref{it:1:lem:expdepth:lopsimred:steps}, and \eqref{it:2:lem:expdepth:lopsimred:steps}. 
  Statement~\eqref{it:2:thm:main:lTRS} follows by adding up the uniform bound $D$ 
  on the expansion depth increase in the $n$ $\scontractred$ steps of the rewrite sequence $\arewseq$.    
\end{proof}


\section{Transfer to \lo\ 
         $\beta$\nb-reduction in the \lambdacalculus}
  \label{sec:transfer:lambda-calculus}

In this section we sketch how the linear-depth-increase result can be transferred from simulating rewrite sequences on terms of the \lopsimTRS~$\lopsimTRSwrt{\alTRS}$
for a \lTRS~$\alTRS$ to \lo\ \betareduction\ rewrite sequences on terms of the \lambdacalculus. 
We formulate correspondence statements via projection and lifting. 
In particular, we formulate statements about the projections of $\slopsimred$ steps to \betareduction\ steps on \lambdaterms,
where the projection takes place via expansion to expanded-form \lambdaterm\ representations,
and about the lifting of \lo\ \betareduction\ rewrite sequences to \lo\ rewrite sequences in \lopsimTRS{s},
where the lifting has to be defined via fully-lazy \lambdalifting. 
We do not prove these statements here, but we illustrate them by means of our running example. 
On the basis of such correspondences between rewrite sequences, the linear-depth-increase result for \lo\ \betareduction\ in the \lambdacalculus\
follows from the linear-depth-increase result for \lopsimTRS\ in Section~\ref{sec:depth:increase}. 


The first correspondence statement concerns the projection of $\slopsimred$ steps to $\sbetared$ steps or empty steps on \lambdaterms\
with the property that \lo\ $\scontractred$ steps project to \lo\ $\sbetared$ steps.

\begin{proposition}[Projection of $\slopsimred$ steps via $\denlterwrt{\alTRS}{\cdot}$]
    \label{prop:projection}
  Let $\lopsimTRSwrt{\alTRS} = \pair{\asiglopsim}{\ruleslopsim}$ be the \lopsimTRS\ for 
  a \lTRS~$\alTRS = \pair{\asig}{\arules}$.
  Let $\ater\in\termsover{\asiglopsim}$
  be a term in $\lopsimTRSwrt{\alTRS}$ such that $\denlterwrt{\alTRS}{\ater} = \alter$ for a \lambdaterm~$\alter$.
  
  Then the following statements hold concerning the projection of $\slopsimred$ steps via $\denlterwrt{\alTRS}{\cdot}$ 
  to steps on \lambdaterms, for all $\ater,\ateri{1}\in\termsover{\asiglosim}$:
  \begin{enumerate}[(i)]\setlength{\itemsep}{0ex} 
    \item
      If $\ater \searchred \ateri{1}$, then $\denlterwrt{\alTRS}{\ater} = \denlterwrt{\alTRS}{\ateri{1}}$.
      That is, the projection of a $\ssearchred$ step via $\denlterwrt{\alTRS}{\cdot}$ is a trivial step.
    \item  
      If $\ater \contractred \ateri{1}$, then $\denlterwrt{\alTRS}{\ater} \betared \denlterwrt{\alTRS}{\ateri{1}}$.
      That is, the projection of a $\scontractred$ step via $\denlterwrt{\alTRS}{\cdot}$ is a $\sbetared$ step.
    \item  
      If $\ater \contractred \ateri{1}$ is a \lo\ step, then $\denlterwrt{\alTRS}{\ater} \lobetared \denlterwrt{\alTRS}{\ateri{1}}$ holds.
      That is, the projection of a \lo\ $\scontractred$ step via $\denlterwrt{\alTRS}{\cdot}$ is $\slobetared$ steps.
  \end{enumerate}
\end{proposition}

A \emph{proof} of this statement can be obtained by defining the projection via the expansion rewrite relation $\sexpred$,
and in particular, via the reduction $\sexprednf$ to expanded forms, which yields \lambdaterm\ representations.
Then it can be shown that $\ssearchred$ steps do not change the expanded form, and that $\scontractred$ steps
correspond to the contraction of \betaredex{es} on the represented \lambdaterms. 

\begin{example}
  We illustrate the projection of $\slopsimred$ rewrite sequences in a \lopsimTRS\ to $\slobetared$ sequences in the \lambdacalculus\
  at our standard example.
  For this, we consider the \lambdaterm~$\alter \defdby \labs{x}{\lapp{(\labs{y}{y})}{(\lapp{(\labs{z}{\labs{w}{\lapp{w}{z}}})}{x})}}$
  from Example~\ref{ex:lopsimred:ltermrep}, and the \lTRS~$\alTRS = \pair{\asig}{\arules}$
  with $\asigmin = \setexp{ \afoscopesym, \bfoscopesym, \cfoscopesym, \dfoscopesym }$
  as defined in Example~\ref{ex:lTRS}, for which $\denlter{\afoscope{\bfoscopesym,\cfoscopesym}} = \alter$ holds, 
  that is, the \lTRS\nb-term $\afoscope{\bfoscopesym,\cfoscopesym}$ represents the \lambdaterm~$\alter$.  
  
  Then the leftmost (and \lo) $\slopsimred$ rewrite sequence in $\lopsimTRSwrt{\alTRS}$ from Example~\ref{ex:lopsimred} 
  projects to the \lo\ $\slobetared$ rewrite sequence in the \lambdacalculus\ from Example~\ref{ex:lopsimred:ltermrep}
  as follows, where we indicate the projection by writing the denoted \lambdaterms\ beneath 
  the corresponding \lambdaterm\ representations:
  \begin{center}
  $
  \begin{array}{clclclc}
    \lopstart{\afoscope{\bfoscopesym,\cfoscopesym}}
      & \ssearchred &
    \lopni{0}{0}{\afoscope{\bfoscopesym,\cfoscopesym}}
      & \ssearchred &
    \folabs{\afovari{0}}{\lopni{0}{1}{\folapp{\bfoscopesym}{\folapp{\cfoscopesym}{\afovari{0}}}}} 
    \\[0.75ex]
    \alter 
      & \ssyntequal &
    \labs{x}{\lapp{(\labs{y}{y})}{(\lapp{(\labs{z}{\labs{w}{\lapp{w}{z}}})}{x})}}
      & \ssyntequal &
    \labs{x}{\lapp{(\labs{y}{y})}{(\lapp{(\labs{z}{\labs{w}{\lapp{w}{z}}})}{x})}}
    \\[2.5ex]
      & \ssearchred & 
    \folabs{\afovari{0}}{\lopni{1}{1}{\bfoscopesym,\folapp{\cfoscopesym}{\afovari{0}}}}  
      & \scontractred & 
    \folabs{\afovari{0}}{\lopni{0}{1}{\folapp{\cfoscopesym}{\afovari{0}}}}  
    \\[0.75ex]
      & \ssyntequal &
    \labs{x}{\underline{\lapp{(\labs{y}{y})}{(\lapp{(\labs{z}{\labs{w}{\lapp{w}{z}}})}{x})}}}
      & \slobetared &
    \labs{x}{\lapp{(\labs{z}{\labs{w}{\lapp{w}{z}}})}{x}}
    \\[2.5ex]
      & \searchred & 
    \folabs{\afovari{0}}{\lopni{1}{1}{\cfoscopesym,\afovari{0}}}
      & \scontractred & 
    \folabs{\afovari{0}}{\lopni{0}{1}{\dfoscope{\afovari{0}}}}
    \\[0.75ex]
      & \ssyntequal &
    \labs{x}{\lapp{(\labs{z}{\labs{w}{\lapp{w}{z}}})}{x}}
      & \slobetared &
    \labs{x}{\lapp{(\labs{z}{\labs{w}{\lapp{w}{z}}})}{x}}
    \\[2.5ex]   
      & \ssearchred & 
    \folabs{\afovari{0}}{\folabs{\afovari{1}}{\lopni{0}{1}{\folapp{\afovari{1}}{\afovari{0}}}}}
      & \ssearchred & 
    \folabs{\afovari{0}}{\folabs{\afovari{1}}{\lopni{1}{2}{\afovari{1},\afovari{0}}}}
    \\[0.75ex]
      & \ssyntequal &
    \labs{x}{\underline{\lapp{(\labs{z}{\labs{w}{\lapp{w}{z}}})}{x}}}
      & \ssyntequal &
     \labs{x}{\labs{w}{\lapp{w}{x}}}  
    \\[2.5ex]    
      & \ssearchred & 
    \folabs{\afovari{0}}{\folabs{\afovari{1}}{\folapp{\afovari{1}}{\lopni{0}{2}{\afovari{0}}}}}
      & \ssearchred &
    \folabs{\afovari{0}}{\folabs{\afovari{1}}{\folapp{\afovari{1}}{\afovari{0}}}}
    \\[0.75ex]
      & \ssyntequal &
    \labs{x}{\labs{w}{\lapp{w}{x}}}
      & \ssyntequal &
    \labs{x}{\labs{w}{\lapp{w}{x}}}
  \end{array}
  $
  \end{center}
  As in Example~\ref{ex:lopsimred:ltermrep} we have underlined redexes that are contracted in $\slobetared$ steps.
  This parallelization of steps can help to recognize, for the latter ones quite directly,
  that projection takes place by taking the expanded form of the \lopsimTRS\ term, and interpreting that as a \lambdaterm\
  (modulo \alphaconversion).%
  \label{ex:prop:projection}
\end{example}


The next lemma states that every \lo\ \betareduction\ step $\alter \lobetared \alteri{1}$
can be lifted to a sequence $\ater \mathrel{\ssearchmred \cdot \scontractred} \ateri{1}$ of leftmost steps in a \lopsimTRS,
provided that $\ater$ denotes $\alter$, and $\ater$ has been obtained by the simulation of a $\slobetared$ rewrite sequence.

\begin{lemma}[Lifting of $\slobetared$ steps to $\ssearchmred \cdot \scontractred$ steps w.r.t.\ $\denlterwrt{\alTRS}{\cdot}$]%
    \label{lem:lifting}
  Let $\alTRS = \pair{\asig}{\arules}$ be a \lTRS.
  Let $\ater\in\termsover{\asig}$ be a ground term 
  such that $\denlterwrt{\alTRS}{\ater} = \alteri{0}$ for a \lambdaterm~$\alteri{0}$.
  Furthermore let $\cter\in\termsover{\asiglopsim}$ 
  with $\denlterwrt{\alTRS}{\cter} = \alter$ for a \lambdaterm~$\alter$
  be the final term of a \lo\ rewrite sequence $\lopstart{\ater} \lopsimmred \cter$.
  
  Then for a $\slobetared$ step $\arewstep \funin \denlterwrt{\alTRS}{\cter} \,{=}\, \alter \lobetared \alteri{1}$ 
  with \lambdaterm~$\alteri{1}$ as target
  there are terms $\cteracc,\cteri{1}\in\termsover{\asiglopsim}$ 
  and a \lo\ $\slopsimred$ rewrite sequence 
  $\Hat{\rho} \funin \cter \searchmred \cteracc \contractred \cteri{1}$
  whose projection via $\denlterwrt{\alTRS}{\cdot}$ amounts to the step $\arewstep$,
  and hence, $\denlterwrt{\alTRS}{\cteracc} = \alter$, and $\denlterwrt{\alTRS}{\cteri{1}} = \alteri{1}$.
\end{lemma}

We note that in the lemma `\lo' in `\lo\ rewrite sequence $\lopstart{\ater} \lopsimmred \cter$'
and `\lo\ $\slopsimred$ rewrite sequence $\Hat{\rho} \funin \cter \searchmred \cteracc \contractred \cteri{1}$'
could both be replaced by `leftmost'. The reason is as follows. 
In a \lopsimTRS~$\lopsimTRSwrt{\alTRS}$ for a \lTRS~$\alTRS$ 
it holds for all rewrite sequences $\lopstart{\bter} \lopsimmred \cter$ for a ground term $\bter$ over the signature of $\alTRS$
and of \lambdaterm\ representations that $\cter$ does not have occurrences of operation symbols $\slopstart,\slopni{n}{i}\in\asiglop$ in nested positions
(but only at identical or parallel positions). 
From this it follows that all redexes of the \lopsimTRS\ in $\cter$ are outermost,
and hence that all \lo\ steps from $\cter$ in $\lopsimTRSwrt{\alTRS}$ arise by contracting leftmost redexes. 

We expect that Lemma~\ref{lem:lifting} can be proved in close analogy to the correctness statement for fully-lazy \lambdalifting.
In particular, it is possible to use the correspondence between weak \betareduction\ steps on \lambdaterms\
and combinator reduction steps on supercombinator representations obtained by fully-lazy \lambdalifting. 
The latter result was formulated and proved by Balabonski in \cite{bala:2012}.

Now by using Lemma~\ref{lem:lifting} in a proof by induction on the length of a $\slobetared$ rewrite sequence 
the theorem below can be obtained. It justifies the use of \lopsimTRS{s} for the simulation of 
$\slobetared$ rewrite sequences.

\begin{proposition}[Lifting of $\slobetared$ to \lo\ $\slopsimred$ rewrite sequences]%
    \label{prop:lifting:lobeta:lo-losim:rewseqs}
  Let $\alTRS = \pair{\asig}{\arules}$ be a \lTRS.
  Let $\ater\in\termsover{\asig}$ be a ground term with $\denlterwrt{\alTRS}{\ater} = \alter$ for a \lambdaterm~$\alter$.
  Then every $\slobetared$ rewrite sequence:
  \begin{center}
    $
    \arewseq \funin \alter = \clteri{0} \lobetared \clteri{1} \lobetared \cdots \lobetared \clteri{k} \; (\,\lobetared \clteri{k+1} \lobetared \cdots\,)
    $ 
  \end{center}
  of finite or infinite length $l\in\nat\cup\setexp{\infty}$ lifts via $\denlterwrt{\alTRS}{\cdot}$ to a \lo\ $\slopsimred$ rewrite sequence: 
%
  \begin{align*}
    \Hat{\arewseq} \funin \lopstart{\ater} 
      = 
    \cteri{0} 
      \mathrel{\ssearchmred \cdot \scontractred} & \:
    \cteri{1} 
      \searchmred
    \;\;\cdots\;\;
    \\
    \;\;\cdots\;\;
      \contractred & \:
    \cteri{k}   
    \; (\,
      \mathrel{\ssearchmred \cdot \scontractred} 
    \cteri{k+1}
      \searchmred
    \cdots \,)  
  \end{align*}
  with precisely $l$ $\scontractred$ steps
  such that furthermore
  $\denlterwrt{\alTRS}{\cteri{i}} = \clteri{i}$ holds for all $i\in\{0,1,\ldots,l\}$.
\end{proposition}

For the same reason as argued above for Lemma~\ref{lem:lifting}, the formulation
`\lo\ $\slopsimred$ rewrite sequence' in this proposition could be replaced by `leftmost $\slopsimred$ rewrite sequence'.

Now by using the lifting of $\slobetared$ rewrite sequences to $\slopsimred$ rewrite sequences
(Proposition~\ref{prop:lifting:lobeta:lo-losim:rewseqs}),
that \lambdaterm\ and \lambdaterm\ representation depths coincide (Proposition~\ref{prop:expdepth:lterrep:lter}),
and that the depth of an \lTRS\ that can represent a \lambdaterm~$\alter$ is bounded by the depth of $\alter$ (Lemma~\ref{lem:depth:losimTRS}),
the theorem above entails our main theorem,
the linear-depth-increase result for \lo\ \betareduction\ rewrite sequences.


\begin{theorem}[Linear depth increase in $\lobetared$-rewrite sequences]\label{thm:main}
  Let $\alter$ be a \lambdaterm.
  Then for every finite or infinite \lo\ rewrite sequence 
  $\arewseq \funin \alter = \clteri{0} \lobetared \clteri{1} \lobetared \cdots \lobetared \clteri{k} \; (\,\lobetared \clteri{k+1} \lobetared \cdots\,)$ 
  from $\alter$  with length $l\in\nat\cup\setexp{\infty}$
  it holds:
  \begin{enumerate}[(i)]\setlength{\itemsep}{0ex}
    \item
      $\depth{\clteri{n+1}} \le \depth{\clteri{n}} + \depth{\alter}\,$
      for all $n\in\nat$ with $n+1\le l$,
      that is,
      the depth increase in each step of $\arewseq$ is uniformly bounded by $\depth{\alter}$. 
    \item  
      $\depth{\clteri{n}} \le \depth{\alter} + n\cdot \depth{\alter} = (n+1)\cdot \depth{\alter}$,
      and hence $\depth{\clteri{n}} - \depth{\alter} \in \bigOmicron{n}$,
      for all $n\in\nat$ with $n\le l$,
      that is,
      the depth increase along $\arewseq$ to the $n$\nb-th reduct is linear in $n$,
      with $\depth{\alter}$ as multiplicative constant.
  \end{enumerate}
\end{theorem}


\section{Idea for a graph rewriting implementation}
  \label{sec:idea:graph:implementation}

The linear-depth-increase result suggests a directed-acyclic-graph implementation of \lo\ \betareduction\
that is based on the following idea.
It keeps subterms shared as much as possible, particularly in the search for the
representation of the next \lo\ redex.
Steps that are used in the search for the next \lo\ redex
do not perform any unsharing, but only use markers to organize the search, and to keep track of its progress.
All search steps together increase the size of the graph only by at most a constant multiple.
Then the number of search steps that are necessary for finding the next \lo\ redex is linear in the size of the current graph. 
Unsharing of the graph only takes place once the next (representation of the) \lo\ redex is found:
then the part of the graph between this redex and the root is unshared (copied),
and subsequently the (represented) redex is contracted. 

The idea is to develop a graph rewriting calculus $\graphcalwrt{\asig}$ 
such that its rewrite relation $\sgraphred$ implements
\lo\ $\scomprewrels{\ssearchmred}{\scontractred}$ rewrite sequences in the corresponding \lopsimTRS~$\lopsimTRSwrt{\asig}$.
We know from Proposition~\ref{prop:projection}
that those \lo\ $\scomprewrels{\ssearchmred}{\scontractred}$ rewrite sequences in turn implement
$\slobetared$ rewrite sequences in the \lambdacalculus.
Starting from a \lambdaterm~$\alteri{0}$,
a \lo\ \betareduction\ rewrite sequence from $\alteri{0}$
is thus first lifted to a $\scomprewrels{\ssearchmred}{\scontractred}$ rewrite sequence 
from a \lambdaterm\ representation $\bteri{0}$ of $\alter$
in an \lopsimTRS~$\lopsimTRSwrt{\asig}$ 
for a \lTRS~$\alTRS = \pair{\asig}{\arules}$ with $\denlter{\bteri{0}} = \alteri{0}$,
and then to a $\sgraphred$ rewrite sequence from a directed-acyclic graph $\agraphi{0}$ that represents~$\bteri{0}$:\vspace*{-1.5ex}
\begin{center}
  $
  \begin{array}{ccccccccccc@{\hspace*{2ex}}c}
    \alteri{0}
      & \slobetared &
    \alteri{1}
      & \slobetared &
    \alteri{2}  
      & \slobetared &
        \ldots
      & \slobetared &
    \alteri{n-1}
      & \slobetared &
    \alteri{n}  
        &  (\lambdacal)
    \\[0.75ex]
    \bteri{0}
      & \scomprewrels{\sseamred}{\sconred} &
    \bteri{1}  
      & \scomprewrels{\sseamred}{\sconred} &
    \bteri{2}
      & \scomprewrels{\sseamred}{\sconred} &
        \ldots
      & \scomprewrels{\sseamred}{\sconred} &
    \bteri{n-1}  
      & \scomprewrels{\sseamred}{\sconred} &
    \bteri{n}
        &  (\lopsimTRSwrt{\asig})
    \\[0.75ex]
    \agraphi{0}
      & \sgraphred &
    \agraphi{1}  
      & \sgraphred &
    \agraphi{2}
      & \sgraphred &
        \ldots
      & \sgraphred &
    \agraphi{n-1}  
      & \sgraphred &
    \agraphi{n}
        &  (\graphcalwrt{\asig})
  \end{array}
  $
\end{center}
(here we have shortened the subscripts in $\scomprewrels{\ssearchred}{\scontractred}$ steps)
where it holds for all $i\in\setexp{1,\ldots,n}\,$:
\begin{center}
  $
  \begin{aligned}
    \denlter{\bteri{i}}
      & =
    \alteri{i} \punc{,}
      & 
    \agraphi{i} \; & \text{represents} \; \bteri{i} \punc{,}
      & 
    \depth{\alteri{i}} & \le (i+1) \cdot \depth{\alteri{0}} \punc{,}
      &
    \depth{\bteri{i}} 
      \le
    \expdepth{\bteri{i}} & = \depth{\alteri{i}} \begin{aligned}[t]
                                                 &
                                                 \le (i+1) \cdot \depth{\bteri{0}}  
                                                 \\[-0.5ex]
                                                 &
                                                 =   (i+1) \cdot \depth{\alteri{0}} \punc{.} 
                                               \end{aligned}                   
  \end{aligned}
  $
\end{center}
Here we have used the linear-depth-increase results Theorem~\ref{thm:main} for \lambdaterms, 
and Theorem~\ref{thm:main:lTRS} for \lambdaterm\ respresentations.  
Now it seems feasible to develop the graph rewrite calculus $\graphcalwrt{\asig}$ in such a way
that the depth $\depth{\agraphi{i}}$ of the (acyclic) graph representations $\agraphi{i}$ of the \lopsimTRS\ terms $\bteri{i}$ 
are bounded by a constant $c$ multiplied with the depth of $\bteri{i}$,
and consequently also bounded by $c$ multiplied with the \lambdadepth\ of $\bteri{i}$, or the depth of $\alteri{i}$:
\begin{center}
  $
  \begin{aligned}
    \depth{\agraphi{i}} & \le c \cdot \depth{\bteri{i}} \le c \cdot \expdepth{\bteri{i}} = c \cdot \depth{\alteri{i}} \punc{,}
      & 
    \text{hence: }
    \depth{\agraphi{i}} & \le c \cdot (i+1) \cdot \depth{\alteri{0}} 
    & & \text{(for all $i\in\setexp{0,1,\ldots,n}$)}  \punc{.}    
  \end{aligned}
  $
\end{center}
The reason for the possible depth increase in the graph representations 
consists in the use of additional controle nodes for keeping track of the progress of \lo\ evaluation: 
links will be used in order to indicate positions to which the \lo\ evaluation needs to backtrack
after having reduced a subexpression to a normal form, or having detected that a subexpression is a normal form. 
The depth of the graphs $\agraphi{i}$ are well-defined because they are acylic. 

The idea for simulating a step $\bteri{i} \comprewrels{\ssearchmred}{\scontractred} \bteri{i+1}$
consists in unsharing the graph representation $\agraphi{i}$ of $\bteri{i}$ only 
between the graph's root and the representation of the \betaredex\ in the $\scontractred$ step,
and then carrying out the representation of the $\scontractred$ step that involves   
replacing the symbol $\afoscopesym$ by a graph version of its scope context $\afoscopecxt$,
together with adapting links accordingly. 
We can expect the size increase in the graph rewrite step $\agraphi{i} \graphred \agraphi{i+1}$ 
to be bounded linearly in the depth $\depth{\agraphi{i}}$ of $\agraphi{i}$, for the first part,
and to be bounded by linearly the size $\size{\afoscopecxt}$, and hence the size $\size{\alteri{0}}$ of $\alteri{0}$,
for the contraction part.   
That is, we want to guarantee that for all $i\in\setexp{0,1,\ldots,n-1}$ it holds:
\begin{center}
  $
\begin{aligned}
  \size{\agraphi{i+1}} & {} \le \size{\agraphi{i}} + d \cdot \depth{\agraphi{n}} + c \cdot \size{\alteri{0}} \punc{,}
  \\
    \text{hence: } 
  \size{\agraphi{i+1}} - \size{\agraphi{i}} & {} \le d \cdot \depth{\agraphi{i}}  + c \cdot \size{\alteri{0}}
                                            \\      
                                            & {} \le c \cdot d \cdot (i+1) \cdot \depth{\alteri{0}} 
                                                     + c \cdot \size{\alteri{0}} \punc{.}    
\end{aligned}
  $
\end{center}  
We may also assume that $d\in\nat$ is at the same time a multiplicative constant for bounding
the size of $\agraphi{0}$ by the sizes of $\bteri{0}$ and $\alteri{0}$:
\begin{center}
  $
  \size{\agraphi{0}} \le d \cdot \max \setexp{ \size{\bteri{0}}, \, \size{\alteri{0}} } \punc{.} 
  $
\end{center}  
On the basis of these assumptions a bound for the size of the $n$\nb-th graph $\agraphi{n}$
of the graph rewrite sequence can be calculated as follows:
\begin{center}
  $
  \begin{aligned}
    \size{\agraphi{n}}
      & \; = \;
    \Bigl(\, \sum_{i=0}^{n-1} (\size{\agraphi{i+1}} - \size{\agraphi{i}}) \,\Bigr)  
      +
    \size{\agraphi{0}}
    \\[-1.25ex]
      &
       \; = \;
    \size{\agraphi{0}}
      +
    \sum_{i=0}^{n-1} 
      \bigl( c \cdot d \cdot (i+1) \cdot \depth{\alteri{0}} 
             + c \cdot \size{\alteri{0}} \bigr)
    \\[-1.25ex]
      & \; = \;
    \size{\agraphi{0}}
      +
    c \cdot d \cdot \depth{\alteri{0}} \cdot  
      \bigl( \sum_{i=1}^{n} i \bigr)
     + c \cdot n \cdot \size{\alteri{0}} 
    \\[-0.75ex]
      &
       \; \le \;
    d \cdot \size{\alteri{0}}
      +
    \frac{1}{2} \cdot c d \cdot \depth{\alteri{0}} \cdot n (n + 1)
      + 
    c \cdot n \cdot \size{\alteri{0}}
    \\
      &
       \; \le \;
    d \cdot \size{\alteri{0}}
      +
    0.5 \cdot c d \cdot \size{\alteri{0}} \cdot n (n + 1) 
      + 
    c \cdot n \cdot \size{\alteri{0}}
    \;\; \in \;\; \bigOmicron{ \size{\alteri{0}} \cdot n^2 } \punc{.}
  \end{aligned}
  $
\end{center}
Now the time for computing the $i$\nb-th rewrite step $\agraphi{i} \graphred \agraphi{i+1}$
will consist of two parts: 
the time $\timei{\fap{\scriptsearch}{\agraphi{i}}}$ for searching the occurrence of the representation of the \lo\ redex in $\agraphi{i}$,
and the time $\timei{\fap{\scriptcontract}{\agraphi{i}}}$ for performing the graph representation of the $\scontractred$ step.
The search part $\timei{\fap{\scriptsearch}{\agraphi{i}}}$ can be organized as a graph traversal of $\agraphi{i}$,
and therefore can be expected to be performed in time that depends linearly on the size of $\agraphi{i}$.
The contraction part $\timei{\fap{\scriptcontract}{\agraphi{i}}}$ consists of the necessary unsharing of the $\agraphi{i}$ between its root
and the represented \lo\ redex, and by performing the $\scontractred$ step on the shared representation $\agraphi{i}$. 
The first subpart necessitates copying work of time that is linearly dependent on the depth $\depth{\agraphi{i}}$ of $\agraphi{i}$.
The second subpart involves the addition of a graph context from $\agraphi{0}$ that corresponds to the scope context $\afoscopecxt$
of the scope symbol $\afoscopesym$ that is part of the $\scontractred$ redex that is contracted; 
it therefore requires copying $\afoscopecxt$, and since $\afoscopecxt$ occurs already in $\agraphi{0}$,
this can be expected to be work that depends linearly on the size of $\agraphi{0}$.
Together we obtain that for some $e,f\in\nat$ it holds: 
\begin{center}
  $
  \begin{aligned}
    \timei{(\agraphi{i} \graphred \agraphi{i+1})}
      & \; = \;
    \timei{\fap{\scriptsearch}{\agraphi{i}}}
      +
    \timei{\fap{\scriptcontract}{\agraphi{i}}}
    \\
      & \; \le \;
    e \cdot \size{\agraphi{i}}
      +
    f \cdot ( \depth{\agraphi{i}} + \size{\agraphi{0}} )
    \\
      & \; \le \;
    (e + 2 f) \cdot \size{\agraphi{i}}
    \;\; \in \;\; \bigOmicron{ \size{\alteri{0}} \cdot i^2 } \punc{.}
  \end{aligned}
  $
\end{center}
From this we now obtain the following rough estimate of the time needed to implement 
the \lo\ rewrite sequence $\alteri{0} \lobetaredn{n} \alteri{n}$ 
by the graph rewrite sequence $\agraphi{0} \graphredn{n} \agraphi{n}$, for some $g\in\nat\,$:
\begin{center}
  $
  \begin{aligned}
    \timei{(\agraphi{0} \graphredrtc \agraphi{n})}
      \; & {} = \;
    \sum_{i=1}^{n-1}   
        \timei{(\agraphi{i} \graphredrtc \agraphi{i+1})}
    \\    
      \; & {} \le \;
    \sum_{i=0}^{n-1}
        g \cdot \size{\alteri{0}} \cdot i^2 
      \; = \;
    g \cdot \size{\alteri{0}} \cdot  \sum_{i=0}^{n-1} i^2
    \;\; \in \;\; \bigOmicron{ \size{\alteri{0}} \cdot n^3 } \punc{.}
  \end{aligned}
  $
\end{center}
This can yield a polynomial cost function for the work that is needed
to faithfully implement a \lo\ \betareduction\ rewrite sequence of length $n$
by `atomic' graph manipulation steps. 

An implementation of such graph rewriting representations of \lo\ \betareduction\ sequences,
broken down into the atomic steps of a port graph rewrite system \cite{stew:2002}, on a reasonable machine
could lead to an alternative proof of the invariance result of Accattoli and Dal Lago.

\enlargethispage{6ex}
\paragraph{Acknowledgment.}
  This article is an extension of my not reviewed contribution~\cite{grab:2016:lindepthincrease:liber:alberti}
  to the Liber Alberti Festschrift on the occasion of the retirement of Albert Visser from Utrecht University in 2016. 
  I want to thank: 
  Vincent van Oostrom, for familiarizing me with \TRSrepresentation{s} of \lambdaterms,
  and with the simulation of weak-$\beta$ reduction by orthogonal \TRSs;
  Dimitri Hendriks, for comments on my drafts of \cite{grab:2016:lindepthincrease:liber:alberti},
  and for his questions about it that helped me;
  J\"{o}rg Endrullis, for help with typsetting Figure~\ref{fig:depth:increase} with TikZ;
  and Luca Aceto for detailed comments about the present version.

\bibliographystyle{plainnat}
\bibliography{ldi-lobr}

\begin{thebibliography}{12}
\providecommand{\natexlab}[1]{#1}
\providecommand{\url}[1]{\texttt{#1}}
\expandafter\ifx\csname urlstyle\endcsname\relax
  \providecommand{\doi}[1]{doi: #1}\else
  \providecommand{\doi}{doi: \begingroup \urlstyle{rm}\Url}\fi

\bibitem[Accattoli and
  Dal~Lago(2014)]{acca:lago:2014:beta-reduction-invariant:LICS}
Beniamino Accattoli and Ugo Dal~Lago.
\newblock {Beta Reduction is Invariant, Indeed}.
\newblock In \emph{Proceedings of the joint conference CSL-LICS '14}, pages
  8:1--8:10, New York, NY, USA, 2014. ACM.
\newblock \doi{10.1145/2603088.2603105}.

\bibitem[Accattoli and Lago(2016)]{acca:lago:2016}
Beniamino Accattoli and Ugo~Dal Lago.
\newblock {(Leftmost-Outermost) Beta Reduction is Invariant, Indeed}.
\newblock \emph{{Logical Methods in Computer Science}}, {Volume 12, Issue 1},
  Mar 2016.
\newblock \doi{10.2168/LMCS-12(1:4)2016}.

\bibitem[Asperti and Levy(2013)]{aspe:levy:2013}
Andrea Asperti and Jean-Jacques Levy.
\newblock {The Cost of Usage in the \protect$\lambda$-Calculus}.
\newblock In \emph{Proceedings of LICS~2013}, LICS '13, pages 293--300,
  Washington, DC, USA, 2013. IEEE Computer Society.
\newblock \doi{10.1109/LICS.2013.35}.

\bibitem[Balabonski(2012)]{bala:2012}
Thibaut Balabonski.
\newblock {A Unified Approach to Fully Lazy Sharing}.
\newblock In \emph{Proceedings of the Symposium POPL~'12}, pages 469--480, New
  York, NY, USA, 2012. ACM.
\newblock \doi{10.1145/2103656.2103713}.

\bibitem[Blanc et~al.(2005)Blanc, L\'{e}vy, and Maranget]{blan:levy:mara:2005}
Tomasz Blanc, Jean-Jacques L\'{e}vy, and Luc Maranget.
\newblock {Sharing in the Weak Lambda-Calculus}.
\newblock In \emph{Processes, Terms and Cycles: Steps on the Road to Infinity.
  Essays dedicated to Jan Willem Klop}, number 3838 in LNCS. Springer, 2005.

\bibitem[Grabmayer(2016)]{grab:2016:lindepthincrease:liber:alberti}
Clemens Grabmayer.
\newblock {Linear Depth Increase of Lambda Terms in Leftmost-Outermost
  Beta-Reduction Rewrite Sequences}.
\newblock In Jan van Eijck, Rosalie Iemhoff, and Joost~J. Joosten, editors,
  \emph{Liber Amicorum Alberti (A Tribute to Albert Visser)}, Tributes, pages
  46--60. College Publications, 2016.
\newblock Also available as report
  \href{http://arxiv.org/abs/1604.07030}{\texttt{arXiv:1604.07030}}.

\bibitem[Grabmayer and van Oostrom(2015)]{grab:oost:2015}
Clemens Grabmayer and Vincent van Oostrom.
\newblock {Nested Term Graphs (Work In Progress)}.
\newblock In Aart Middeldorp and Femke~van Raamsdonk, editors, \emph{{\rm
  Proceedings of the TERMGRAPH~2014}}, volume 183 of \emph{EPTCS}, pages
  48--65. Open Publishing Association, 2015.
\newblock \doi{10.4204/EPTCS.183.4}.

\bibitem[Hughes(1982{\natexlab{a}})]{hugh:1982}
John Hughes.
\newblock {Super-combinators a New Implementation Method for Applicative
  Languages}.
\newblock In \emph{Proceedings of the 1982 ACM Symposium on LISP and Functional
  Programming}, LFP '82, pages 1--10, New York, NY, USA, 1982{\natexlab{a}}.
  ACM.
\newblock ISBN 0-89791-082-6.
\newblock \doi{10.1145/800068.802129}.

\bibitem[Hughes(1982{\natexlab{b}})]{hugh:1982:report}
John Hughes.
\newblock {Graph Reduction with Supercombinators}.
\newblock Technical Report PRG28, Oxford University Computing Laboratory, June
  1982{\natexlab{b}}.

\bibitem[{Peyton~Jones}(1987)]{peyt:1987}
Simon~L. {Peyton~Jones}.
\newblock \emph{The Implementation of Functional Programming Languages}.
\newblock Prentice-Hall, Inc., 1987.

\bibitem[Stewart(2002)]{stew:2002}
Charles Stewart.
\newblock {Reducibility between Classes of Port Graph Grammar}.
\newblock \emph{Journal of Computer and System Sciences}, 65\penalty0
  (2):\penalty0 169 -- 223, 2002.
\newblock ISSN 0022-0000.
\newblock \doi{10.1006/jcss.2002.1814}.

\bibitem[Terese(2003)]{terese:2003}
Terese.
\newblock \emph{{Term Rewriting Systems}}, volume~55 of \emph{Cambridge Tracts
  in Theoretical Computer Science}.
\newblock Cambridge University Press, 2003.

\end{thebibliography}
\label{sec:biblio}

\end{document}